\documentclass[11pt]{article}
\usepackage[sort&compress,numbers]{natbib}
\usepackage{fullpage,bm,times}
\usepackage{amsmath,amssymb,latexsym, amsopn, amsfonts}
\usepackage[boxed,ruled,vlined,algosection,linesnumbered]{algorithm2e}
\usepackage{titlesec}
\usepackage{caption}
\usepackage{graphicx}
\usepackage{wrapfig,frame}
\usepackage{framed}\usepackage{xcolor}
\usepackage{enumitem}

\newcommand{\remove}[1]{}
\newcommand{\rep}{state}
\newcommand{\config}{config}
\newcommand{\assert}{lemma}
\newcommand{\noReconfig}{noReco}
\newcommand{\configEstab}{estab}
\newcommand{\notif}{prp}
\newcommand{\notifSet}{notifSet}

\newtheorem{theorem}{Theorem}[section]
\newtheorem{lemma}[theorem]{Lemma}
\newtheorem{claim}[theorem]{Claim}

\newtheorem{remark}{Remark}
\newtheorem{corollary}[theorem]{Corollary}
\newtheorem{definition}{Definition}[section]
\newenvironment{proof}{\noindent{\bf Proof.}}{\hfill$\Box$}

\newenvironment{proofsketch}{\noindent{\bf Proof Sketch.}}{\hfill$\Box$}

\newcommand{\boldsubparagraph}[1]{\subparagraph{\bf #1}}
\titleformat*{\subparagraph}{\it \raggedleft}

\begin{document}
\begin{titlepage}
\renewcommand{\thefootnote}{\fnsymbol{footnote}}


\title{Self-Stabilizing Reconfiguration\\ \large{(Technical Report)}
}


\author{Shlomi Dolev~\footnote{Department of Computer Science, Ben-Gurion University of the Negev, Beer-Sheva, Israel. Email: {\tt dolev@cs.bgu.ac.il} 
}
\and Chryssis Georgiou~\footnote{ Department of Computer Science, University of Cyprus, Nicosia, Cyprus. Email: {\tt {\char '173}chryssis, imarco01{\char '175} @cs.ucy.ac.cy} . Supported by the University of Cyprus.} \and~Ioannis Marcoullis$^\dag$ \and Elad M.\ Schiller~\footnote{Department of Computer Science and Engineering, Chalmers University of Technology, Gothenburg, SE-412 96, Sweden. Email: {\tt elad@chalmers.se}.}}

\date{}

\maketitle

\thispagestyle{empty}

\begin{abstract}
Current reconfiguration techniques are based on starting the system in a consistent configuration, in which all participating entities are in a predefined state. Starting from that state, the system must preserve consistency as long as a predefined churn rate of processors joins and leaves is not violated, and unbounded storage is available.
Many working systems cannot control this churn rate and do not have access to unbounded storage. System designers that neglect the outcome of violating the above assumptions may doom the system to exhibit illegal behaviors. We present the first automatically recovering reconfiguration scheme that recovers from transient faults, such as temporal violations of the above assumptions. Our self-stabilizing solutions regain safety automatically by assuming temporal access to reliable failure detectors. 
Once safety is re-established, the failure detector reliability is no longer needed. Still, liveness is conditioned by the failure detector's unreliable signals. We show
that our self-stabilizing reconfiguration techniques 
can serve as the basis for the
implementation of several dynamic services over message passing systems. 
Examples include  
self-stabilizing reconfigurable virtual synchrony, 
which, in turn, can be used for implementing a self-stabilizing  
reconfigurable state-machine replication and self-stabilizing reconfigurable emulation of shared memory.
\vspace{2em}

\noindent {\bf Keywords:} Self-stabilization, Dynamic Participation, Reconfiguration, Virtual Synchrony, State Machine Replication.\vspace{2em}





\end{abstract}

\end{titlepage}

\renewcommand{\thefootnote}{\arabic{footnote}}


\section{Introduction} 
\paragraph{Motivation.} 
We consider distributed systems that work in dynamic asynchronous environments, such as a shared storage system~\cite{DBLP:journals/cacm/MusialNS14}. Quorum configurations~\cite{DBLP:journals/dc/PelegW97,DBLP:series/synthesis/2012Vukolic}, 
i.e., set of active processors (servers or replicas), are typically used to provide services to the participants of the system. Since over time, the (quorum) configuration may gradually lose active participants due to voluntary leaves and stop failures, there is a need to allow the participation of newly arrived processors and from time to time to \emph{reconfigure} so that the new configuration is built on a more recent participation group. Over the last years, 
a number of reconfiguration techniques have been proposed, mainly for state machine replication and emulation of atomic memory 
(e.g.,~\cite{DynaStore,birmanMR2010,DBLP:journals/corr/BortnikovCPRSS15,RAMBO,DBLP:journals/eatcs/AguileraKMMS10,spiegelmandynamic,DBLP:journals/sigact/LamportMZ10,DBLP:conf/wdag/AttiyaCEKW15,DBLP:conf/wdag/GafniM15,DBLP:conf/wdag/JehlVM15,Baldoni09,DBLP:journals/jpdc/ChocklerGGMS09}). These reconfiguration techniques are based on starting the system in a consistent configuration, in which all processors are in their initial state. Starting from that state, the system must preserve consistency as long as a predefined churn rate of processors' joins and leaves is not violated and unbounded storage is available. Furthermore, they do not tolerate {\em transient faults} that can cause an arbitrary corruption of the system's state.

Many working systems cannot control their churn rate and do not have access to unbounded storage. System designers that neglect the outcome of violating the above assumptions may doom the system to forever exhibit a behavior that does not satisfy the system requirements. Furthermore, the dynamic and difficult-to-predict nature of distributed systems gives rise to many fault-tolerance issues and requires efficient solutions. Large-scale message passing networks are asynchronous and they are subject to transient faults due to hardware or software temporal malfunctions, short-lived violations of the assumed failure rates or violation of correctness invariants, such as the uniform agreement among all current participants about the current configuration. Fault tolerant systems that are {\em self-stabilizing}~\cite{D2K} can recover after the occurrence of transient faults (as long as the program's code is still intact).


\paragraph{Our contributions and approach.}
We present the first automatically recovering reconfiguration scheme that recovers from transient faults, such as temporary violations of the predefined churn rate or the unexpected activities of processors and communication channels. Our blueprint for self-stabilizing reconfigurable distributed systems can withstand a temporal violation of such assumptions, and recover once conditions are resumed. It achieves this with only a bounded amount of local storage and message size. Our self-stabilizing solutions regain safety automatically by assuming temporal access to reliable failure detectors. Once safety is re-established, the failure detector reliability is no longer needed; still, liveness is conditioned by the failure detector's unreliable signals.  
We now overview our approach.\vspace{.3em} 

\noindent{\em Reconfiguration scheme:} Our scheme comprises of two layers that appear as a single ``black-box" module to an application that uses the reconfiguration service. The objective is to provide to the application a {\em conflict-free} configuration,
such that no two alive processors consider different configurations. The first layer, called {\em Reconfiguration Stability Assurance} or {\em recSA} for short (detailed in Section~\ref{sec:RSA}), is mainly responsible for detecting configuration conflicts (that could be a result of transient faults). It deploys a {\em brute-force} technique for converging to a conflict-free new configuration. 
It also employs another technique for {\em delicate} configuration replacement when a processor notifies that it wishes to replace the current configuration with a new set of participants. 
For both techniques, processors use a failure detector (detailed in Section~\ref{sec:settings}) to obtain membership information, and configuration convergence is reached when failure detectors have temporal reliability. 
Once a uniform configuration is installed, the failure detectors' reliability is no longer needed and from then on our liveness conditions consider unreliable failure detectors.  
The decision for requesting a delicate reconfiguration is controlled by the other layer, called {\em Reconfiguration Management} or {\em recMA} for short (detailed in Section~\ref{sec:reconMan}).

Specifically, if a processor suspects that the dependability of the current configuration is under jeopardy, it seeks to obtain a majority approval from the alive {\em members} of the current configuration, and request a (delicate) reconfiguration from \emph{recSA}. Moreover, in the absence of such a majority (e.g., configuration replacement was not activated ``on time'' or the churn assumptions were violated), the {\em recMA} can aim to control the recovery via an \emph{recSA} reconfiguration request. Note that the current participant set can, over time, become different than the configuration member set. As new members arrive and others leave, changing the configuration based on system membership would imply a high frequency of (delicate) reconfigurations, especially in the presence of high churn. We avoid unnecessary reconfiguration requests by requiring a weak liveness condition:  if a majority of the configuration set has not collapsed, then there exists at least one processor that is known to trust this majority in the failure detector of each alive processor.
%
%
%
Such active configuration members can aim to replace the current configuration with a newer one (that would provide an approving majority for prospective reconfigurations) without the use of the brute-force stabilization technique. \vspace{.3em}

\noindent{\em Joining mechanism:} We complement our reconfiguration scheme with a self-stabilizing joining mechanism (detailed in Section~\ref{sec:join}) that manages and controls the inclusion of new processors into the system. Here extra care needs to be taken so that newly joining processors do not ``contaminate" the system state with stale information (due to arbitrary faults). For this, together with other techniques, we follow a snap-stabilizing data link protocol (see Section~\ref{sec:settings}). We have designed our joining mechanism so that the decision of whether new members should be included in the system or not is {\em application-controlled}. In this way, the churn (regarding new arrivals) can be ``fine-tuned" based on the application requirements; we have modeled this
by having joining processors obtaining approval from a majority of the members of the current configuration (if no reconfiguration is taking place). These, in turn, provide such approval if the application's (among other) criteria are met. 
%
%
We note that in the event of transient faults, such as an unavailable approving majority, {\em recSA} ensures recovery via brute-force stabilization that includes all alive processors.\vspace{.3em}

\noindent {\em Applications:} We demonstrate the usability and modularity of our self-stabilizing reconfiguration scheme and joining mechanism by using them to develop self-stabilizing dynamic participation versions of several algorithms: a label
algorithm for providing a bounded self-stabilizing labeling scheme (Section~\ref{sec:label}); a self-stabilizing counter increment algorithm (Section~\ref{sec:counter}); a self-stabilizing virtual synchrony algorithm that leads to self-stabilizing state machine replication and a self-stabilizing MWMR emulation of shared memory (Section~\ref{sec:VS}). 
These algorithms are derived by combining our reconfiguration scheme and joining mechanism with the corresponding self-stabilizing algorithms developed for static membership systems in~\cite{SSVS}.

\paragraph{Related work.}
\sloppy{As mentioned, existing solutions for providing reconfiguration in dynamic systems, such as~\cite{RAMBO} and~\cite{DynaStore}, do not consider  transient faults and self-stabilization, because their correctness proofs (implicitly) depend on a coherent start~\cite{DBLP:journals/cacm/MusialNS14} and also assume that fail-stops can never prevent the (quorum) configuration to facilitate configuration updates.}
They also often use unbounded counters for ordering consensus messages (or for shared memory emulation) and by that facilitate configuration updates, e.g.,~\cite{RAMBO}. 
Our self-stabilizing solution recovers after the occurrence of transient faults, which we model as an arbitrary starting state, and guarantees a consistent configuration that provides (quorum) services, e.g., allowing reading from and writing to distributed shared memory objects, and at the same time managing the configuration providing these services.

Significant amount of research was dedicated in characterizing the fault-tolerance guarantees that can be provided by
difference quorum system designs; see~\cite{DBLP:series/synthesis/2012Vukolic} for an in depth discussion. 
In this paper we use majorities, which is regarded as the simplest form of a quorum system (each set composed of a majority of the processors is a quorum). 
Our reconfiguration scheme can be modified to support more complex, quorum systems, as long as processors have access to a mechanism (a function actually) that given a set of processors can generate the specific quorum system. 
Another important design decision is \emph{when} a reconfiguration (delicate in our case) must take place; see the related discussion
in~\cite{DBLP:journals/cacm/MusialNS14}. One simple decision would be to reconfigure when a fraction (e.g., 1/4th) of the members of a configuration appear to have failed. More complex decisions could use prediction mechanisms (possibly based on statistics). 
This issue is outside of the scope of this work; however, we have designed our reconfiguration scheme (specifically  the {\em recMA} layer) to be able to use any decision mechanism imposed by the application (via an application interface).

\section{System Settings}
\label{sec:settings} \label{s:sys}
\paragraph{Processing entities.} We consider an asynchronous message-passing system of processors. 
Each processor $p_i$ has a unique identifier, $i$, taken from a totally-ordered 
set of identifiers $P$. 
The number of live and connected processors at any time of the computation is bounded by some $N$ such that $N \ll |P|$.  
We refer to such processors as \emph{active}. 
We assume that the processors have knowledge of the upper bound $N$, but not of the actual number of active processors.
Processors may stop-fail by crashing; a processor may crash at any point without warning, and in this event a crashed processor takes no further steps and never rejoins the computation. 
%
%
(For readability sake, we model rejoins as transient faults rather than considering them explicitly. 
Self-stabilization inherently deals with rejoins by regarding the past join information as possibly corrupted.)
New processors may join the system (using a joining procedure) at any point in time with an identifier drawn from $P$, such that this identifier is only used by this processor forever.
A \emph{participant} is an active processor that has joined the computation.
Note that $N$ accounts for all active processors, both participants and those that are still joining. 
\vspace{.3em}

%

\paragraph{Communication.} The network topology is that of a fully connected graph, and links have a bounded capacity $cap$.
Processors exchange low-level messages called \emph{packets} to enable a reliable delivery of high level
\emph{messages}. 
Packets sent may be lost, reordered, or duplicated but not arbitrarily created, although the channels may initially (after transient faults) contain stale packets, which due to the boundedness of the channels are also bounded in a number that is in $O(N^2cap)$.
We assume the availability of self-stabilizing protocols for reliable FIFO end-to-end message delivery over unreliable channels with bounded capacity, such as the ones of~\cite{DBLP:journals/ipl/DolevDPT11} or~\cite{DBLP:conf/sss/DolevHSS12}.

Specifically, when processor $p_i$ sends a packet, $pkt_1$, to processor $p_j$, the operation $send$ inserts a copy of $pkt_1$ into the FIFO queue representing the communication channel from $p_i$ to $p_j$. 
Since links are bounded in capacity, the new packet might be omitted or some already sent packet may be lost. 
While we assume that packets can spontaneously be omitted, i.e., lost from the channel, a packet that is sent infinitely often is received infinitely often. Namely, the communication channels provide \textit{fair communication}.
The policy of acknowledging is that acknowledgments are sent only when a packet arrives, and not spontaneously. 
Packet $pkt_1$ is retransmitted until more than the total capacity acknowledgments arrive, and then $pkt_2$ starts being transmitted. 
This forms an abstraction of token carrying messages between the two processors.
In this way the two processors (sender and receiver) can continuously exchange a ``token''.
%
%
We use this token exchange technique to implement a {\em heartbeat} for detecting  whether a processor is active or not; when a processor in no longer active, the token will not be returned back to the other processor.

Due to the possibility of arbitrary faults and of the dynamic nature of the network, we cannot assume that processors have knowledge of the identifier of the processor with which they are communicating.
We employ two anti-parallel data-link protocols, where every packet of one data-link is identified by the identifiers of the sender and receiver of the data link it participates in. 
For example, if the communication link connects $p_i$ and $p_j$, packets of the data link in which $p_i$ ($p_j$) acts as the sender that traverse from $p_i$ to $p_j$ ($p_j$ to $p_i$) are identified by the label $p_i$ ($p_j$), while the label of packets traversing from $p_j$ ($p_i$) are extended by adding $p_j$ ($p_i$) to the label to form the label $p_x,p_j$ ($p_x,p_i$, respectively).  Any packet $p_x,p_y$ arriving to $p_i$ ($p_j$) where $x \neq i$ ($x \neq j$) is ignored. Thus, eventually the data link in which $p_i$ is the sender is implemented by packets with label $p_i$ ($p_j$) traversing from $p_i$ to $p_j$ ($p_j$ to $p_i$). The analogous holds for the packets implementing the data link in which $p_j$ serves as the sender. Thus, both parties will eventually know the identifier of the other party and regard the token of the data link in which the sender has the greater identifier among them, to be the used token. 

Using the underlying packet exchange protocol described, 
a processor $p_i$ that has received a packet from some processor $p_j$ which did not belong to $p_i$'s failure detector, engages in a two phase protocol with $p_j$ in order to ``clean'' their intermediate link. 
This is done before any messages are delivered to the algorithms that handle reconfiguration, joining and applications.
We follow the snap-stabilizing data link protocol detailed in \cite{DBLP:journals/tcs/DolevT09}.
A \emph{snap-stabilizing} protocol is one which allows the system (after faults cease) to behave according to its specification upon its first invocation.
We require that every data-link established between two processors is initialized and cleaned straight after it is established.
In contrast to~\cite{DBLP:journals/tcs/DolevT09} where the protocol is run on a tree and initiated from the root, our case requires that each pair of processors takes the responsibility of cleaning their intermediate link.
Snap-stabilizing data links do not ignore signals indicating the existence of new connections, possibly some physical carrier signal from the port. 
In fact, when such a connection signal is received by the newly connected parties, they start a communication procedure that uses the bound on the packet in transit and possibly in buffers too, to clean all unknown packets in transit, by repeatedly sending the same packet until more than the round trip capacity acknowledgments arrive.\vspace{.3em}
%

%

\paragraph{$(N,\Theta)$-failure detector.} We consider the $(N,\Theta)$-failure detector that uses the token exchange and heartbeat detailed above.
This is an extension of the $\Theta$-failure detector used in \cite{Blanchard2013SSPaxos}. 
It allows each processor $p_i$ to order other processors according to how recently they have communicated.
Each processor $p_i$ maintains an ordered heartbeat count vector $nonCrashed$, with an entry corresponding to each processor $p_k$ that exchanges the token (i.e., sends a heartbeat) with $p_i$.
Specifically, whenever $p_i$ receives the token from $p_j$, it sets the count corresponding to $p_j$ to $0$ and increments the count of every other processor by one.
In this way, $p_i$ manages to rank every processor $p_k$ according to the token exchanges that it has performed with $p_i$ in relation to the token exchanges that it has performed with some other processor $p_j$.
So the processor that has most recently contacted $p_i$ is the first in $p_i$'s vector.

The technique enables $p_i$ to obtain an estimate on the number of processors $n_i$ that are active in the system; $n_i \leq N$. 
Assuming that $p_c$ is the most recently crashed processor, then every processor other than $p_c$ will eventually exchange the token with $p_i$ many times, and their heartbeat count will be set to zero, while $p_c$'s will be increasing continuously.
Eventually, every other processor's count (given they remain alive and communicating) will become lower than $p_c$'s and $p_c$ will be ranked last in $nonCrashed$.
Moreover, while difference between heartbeat counts of non-crashed processors does not become large, the difference of these counts and that of $p_c$ increases to form a significant ever-expanding ``gap''.
The last processor before the gap is the ${n_i}^{th}$ processor and this provides an estimate on the number of active processors.
Since there are at most $N$ processors in the computation at any given time, we can ignore any processors that rank below the $N^{th}$ vector entry.
If, for example, the first 30 processors in the vector have corresponding counters of up to 30, then the $31^{st}$ will have a count much greater than that, say 100; so $n_i$ will be estimated at 30.
This estimation mechanism is suggested in \cite{DBLP:journals/cjtcs/DolevH97} and in \cite{DBLP:journals/tmc/DolevSW06}. \\


\paragraph{The interleaving model and self-stabilization.}
A program is a sequence of {\em (atomic) steps}. Each atomic step starts with local computations and ends with a communication operation, i.e., packet $send$ or $receive$. We assume the standard interleaving model where at most one step is executed in every given moment.  An input event can either be the arrival of a  packet or a periodic timer triggering $p_i$ to (re)send. Note that the system is asynchronous and the rate of the timer is totally unknown. The {\em state}, $c_i$, consists of $p_i$'s variable values and the content of $p_i$'s incoming communication channels. A step executed by $p_i$ can change the state of $p_i$. The tuple of the form $(c_1, c_2, \cdots, c_n)$ is used to denote the {\em system state}. An {\em execution (or run)} $R={c_0,a_0,c_1,a_1,\ldots}$ is an alternating sequence of system states $c_x$ and steps $a_x$, such that each $c_{x+1}$, except the initial system state $c_0$, is obtained from $c_x$ by the execution of $a_x$. A practically infinite execution is an execution with many steps, where many is defined to be proportional to the time it takes to execute a step and the life-span time of a system. The system's task is a set of executions called {\em legal executions} ($LE$) in which the task's requirements hold. An algorithm is {\em self-stabilizing} with respect to $LE$ when every (unbounded) execution of the algorithm has a suffix that is in $LE$.

\section{Self-stabilizing Reconfiguration Scheme}
\label{sec:reconf}

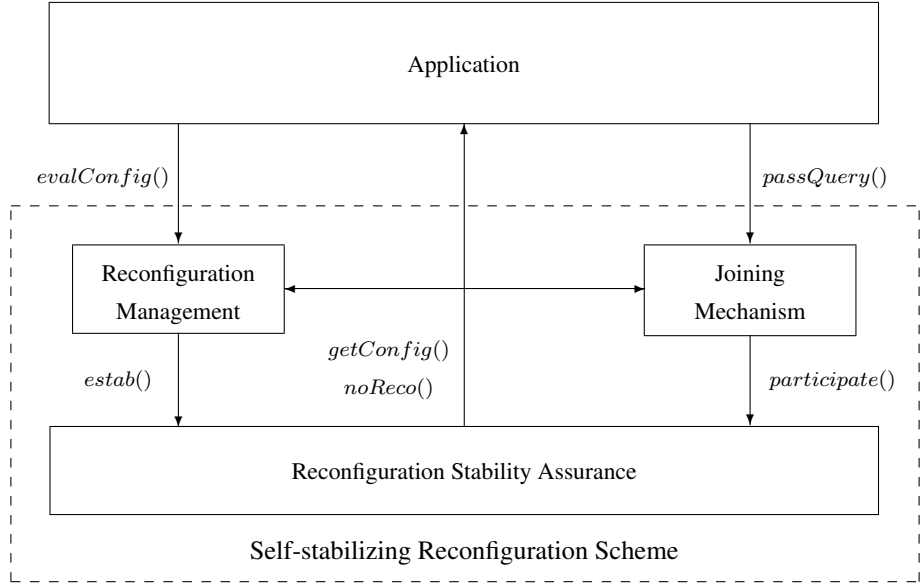
\begin{figure}[t!]
\center
\captionsetup{margin=10pt,font=small,labelfont=bf}
\setlength{\unitlength}{4144sp}%
\begingroup\makeatletter\ifx\SetFigFont\undefined%
\gdef\SetFigFont#1#2#3#4#5{%
  \reset@font\fontsize{#1}{#2pt}%
  \fontfamily{#3}\fontseries{#4}\fontshape{#5}%
  \selectfont}%
\fi\endgroup%
\begin{picture}(5457,3489)(2227,-5923)
\thinlines
{\color[rgb]{0,0,0}\put(2476,-5506){\framebox(4950,521){}}
}%
{\color[rgb]{0,0,0}\put(3241,-3166){\vector( 0,-1){720}}
}%
{\color[rgb]{0,0,0}\put(2471,-3166){\framebox(4955,720){}}
}%
{\color[rgb]{0,0,0}\put(6031,-4435){\framebox(1260,537){}}
}%
{\color[rgb]{0,0,0}\put(2611,-4421){\framebox(1260,523){}}
}%
{\color[rgb]{0,0,0}\put(4951,-4983){\vector( 0, 1){1817}}
}%
{\color[rgb]{0,0,0}\put(6661,-3166){\vector( 0,-1){720}}
}%
{\color[rgb]{0,0,0}\put(6661,-4434){\vector( 0,-1){540}}
}%
{\color[rgb]{0,0,0}\put(4951,-4156){\vector( 1, 0){1082}}
}%
{\color[rgb]{0,0,0}\put(4951,-4156){\vector(-1, 0){1075}}
}%
{\color[rgb]{0,0,0}\put(3241,-4421){\vector( 0,-1){565}}
}%
{\color[rgb]{0,0,0}\put(2239,-5911){\dashbox{57}(5433,2250){}}
}%
\put(4951,-5298){\makebox(0,0)[b]{\smash{{\SetFigFont{9}{10.8}{\rmdefault}{\mddefault}{\updefault}{\color[rgb]{0,0,0}Reconfiguration Stability Assurance}%
}}}}
\put(4948,-2858){\makebox(0,0)[b]{\smash{{\SetFigFont{9}{10.8}{\rmdefault}{\mddefault}{\updefault}{\color[rgb]{0,0,0}Application}%
}}}}
\put(3241,-4336){\makebox(0,0)[b]{\smash{{\SetFigFont{9}{10.8}{\rmdefault}{\mddefault}{\updefault}{\color[rgb]{0,0,0}Management}%
}}}}
\put(6661,-4111){\makebox(0,0)[b]{\smash{{\SetFigFont{9}{10.8}{\rmdefault}{\mddefault}{\updefault}{\color[rgb]{0,0,0}Joining}%
}}}}
\put(6661,-4336){\makebox(0,0)[b]{\smash{{\SetFigFont{9}{10.8}{\rmdefault}{\mddefault}{\updefault}{\color[rgb]{0,0,0}Mechanism}%
}}}}
\put(3241,-4111){\makebox(0,0)[b]{\smash{{\SetFigFont{9}{10.8}{\rmdefault}{\mddefault}{\updefault}{\color[rgb]{0,0,0}Reconfiguration}%
}}}}
\put(2791,-3526){\makebox(0,0)[b]{\smash{{\SetFigFont{8}{9.6}{\rmdefault}{\mddefault}{\itdefault}{\color[rgb]{0,0,0}$evalConfig()$}%
}}}}
\put(7111,-3526){\makebox(0,0)[b]{\smash{{\SetFigFont{8}{9.6}{\rmdefault}{\mddefault}{\itdefault}{\color[rgb]{0,0,0}$passQuery()$}%
}}}}
\put(7156,-4741){\makebox(0,0)[b]{\smash{{\SetFigFont{8}{9.6}{\rmdefault}{\mddefault}{\itdefault}{\color[rgb]{0,0,0}$participate()$}%
}}}}
\put(2881,-4741){\makebox(0,0)[b]{\smash{{\SetFigFont{8}{9.6}{\rmdefault}{\mddefault}{\itdefault}{\color[rgb]{0,0,0}$estab()$}%
}}}}
\put(4501,-4561){\makebox(0,0)[b]{\smash{{\SetFigFont{8}{9.6}{\rmdefault}{\mddefault}{\itdefault}{\color[rgb]{0,0,0}$getConfig()$}%
}}}}
\put(4501,-4786){\makebox(0,0)[b]{\smash{{\SetFigFont{8}{9.6}{\rmdefault}{\mddefault}{\itdefault}{\color[rgb]{0,0,0}$noReco()$}%
}}}}
\put(4951,-5776){\makebox(0,0)[b]{\smash{{\SetFigFont{10}{12.0}{\rmdefault}{\mddefault}{\updefault}{\color[rgb]{0,0,0}Self-stabilizing Reconfiguration Scheme}%
}}}}
\end{picture}%

\caption{
The reconfiguration scheme modules internal interaction and the interaction with the application. 
The Reconfiguration Stability Assurance ($recSA$) layer provides information on the current configuration and on whether a reconfiguration is not taking place using the $getConfig()$ and $\noReconfig()$ interfaces.
This is based of local information.
The Reconfiguration Management ($recMA$) layer uses the prediction mechanism $evalConfig()$ which is application based to evaluate whether a reconfiguration is required.
If a reconfiguration is required, $recMA$ initiates it with $\configEstab()$.
Joining only proceeds if a configuration is in place and if no reconfiguration is taking place.
When the joining mechanism has received a permission to access the application (using  $passQuery()$) it can then join via $participate()$.
The direction of an arrow from a module $A$ to a module $B$ illustrates the transfer of the specific information from $A$ to $B$.
}

\label{fig:modules}
\end{figure}

The reconfiguration scheme is composed of the Reconfiguration Stability Assurance ($recSA$) layer (Section~\ref{sec:RSA}), the Reconfiguration Management ($recMA$) layer (Section~\ref{sec:reconMan}), and is accompanied by the Joining Mechanism (Section~\ref{sec:join}).
Figure~\ref{fig:modules} depicts the interaction between the modules and with the application. 
The Reconfiguration Stability Assurance ($recSA$) layer ensures that participants eventually have a common configuration set.
It also introduces processors that want to join the computation and provides information on the current configuration and on whether a reconfiguration is not taking place using the $getConfig()$ and $\noReconfig()$ interfaces, respectively.

The Reconfiguration Management ($recMA$) layer strives to maintain a majority of active processors of the configuration set, to this end, and may also request a reconfiguration from $recSA$ via the $estab()$ interface. 
This is done when a configuration majority is suspected as collapsed or if a majority of active processor configuration members appears to require a reconfiguration based on some application-defined prediction function ($evalConf()$).
A joining mechanism gives the application the leverage required to control participation and ensure that processors enter the computation with the most recent state.
A joiner becomes a participant via $participate()$ only if $passQuery()$ of a majority of configuration members is reported as $\sf True$. 
We now proceed with the details. 



\subsection{Reconfiguration Stability Assurance}
\label{sec:RSA}
We present the Reconfiguration Stability Assurance layer ($recSA$), a self-stabilizing algorithm for assuring correct configuration while allowing the updates from the Reconfiguration Management layer (Section~\ref{sec:reconMan}). 
We first describe the algorithm (Section~\ref{sec:recSAdescr}) and then we prove its correctness (Section~\ref{sec:recSAproof}).

\begin{algorithm*}[t!]

\caption{Self-stabilizing Reconfiguration Stability Assurance; code for processor $p_i$}
\label{alg:disCongif}
\begin{footnotesize}
{\bf Variables}:
The following arrays consider both $p_i$'s own value (entry $i$) and $p_j$'s most recently received value (entry $j$).  

$\textnormal{config}[]$: an array in which $\textnormal{config}[i]$ is $p_i$'s view on the current configuration.
%
%
Note that $p_i$ assigns the \emph{empty (configuration)} value $\bot$ after receiving a conflicting (different) non-empty configuration value. 

$\textnormal{FD}[]$: an array in which $\textnormal{FD}[i]$ represents $p_i$'s failure detector. 
%
Note that we consider only the trusted processors rather than the suspected ones. Namely, crashed processors are eventually suspected. 

$\textnormal{FD}[].part$ is the participant set, where $\textnormal{FD}[i].part$ is an alias for $\{ p_j \in \textnormal{FD}[i] : \textnormal{config}[j] \neq \sharp \}$ and $\textnormal{FD}[j].part$ refers to the last value received from $p_j$. Namely, the $\textnormal{FD}$ field of every message encodes also this participation information.

$\textnormal{\notif}[]$ is an array of pairs $\langle phase \in \{0,1,2\}, set \subseteq P \rangle$, where $\textnormal{\notif}[i]$ refers to $p_i$'s configuration replacement notifications.
%
%
In the pair $\textnormal{\notif}[i]$, the field $set$ can either be $\bot$ (`no value') or the proposed processor set.

$\textnormal{echo}[]$ is an array in which $\textnormal{echo}[i]$ is $(\textnormal{FD}[i].part, \textnormal{\notif}[i], \textnormal{all}[i])$'s alias and $\textnormal{echo}[j]$ refers to the most recent value that $p_i$ received from $p_j$ after $p_j$ had responded to $p_i$ with the most recent values it got from $p_i$.

$\textnormal{all}[]$ is an array of Booleans, where $\textnormal{all}[i]$ refers to the case in which $p_i$ observes that all trusted processors had noticed its current (maximal) notification and they hold the same notification. 
%

$\textnormal{allSeen}$: a list of processors $p_k$ for which $p_i$ received the $\textnormal{all}[k]$ indication. 

\vspace{0.35em}

{\bf Interface functions}:

{\bf function} \label{ln:chsConfig} 
$chsConfig()=$ 
{\bf return(choose}$(\{ \textnormal{config}[k] \}_{p_k \in \textnormal{FD}[i]} \setminus \{ \sharp \}))$, where {\bf choose}$(\emptyset)=\bot$ else {\bf choose}$(set)\in set$;\

{\bf function} \label{ln:getConfig} $getConfig()=$ \{{\bf if} $\noReconfig()$ {\bf then return}$(chsConfig())$ {\bf else return(}$\textnormal{config}[i])\}$;\



{\bf function} \label{ln:\noReconfig} $\noReconfig() = ((p_i \notin (\cap_{p_j \in \textnormal{FD}[i]  \setminus \{ p_i \} } \textnormal{FD}[i])) \lor 
(|  \{ \textnormal{config}[k] \}_{p_k \in \textnormal{FD}[i]} \setminus \{ \sharp \} | > 1) \lor (\{ \textnormal{FD}[i].part \} \neq \{ \textnormal{FD}[k].part, \textnormal{echo}[k].part \}_{p_k \in \textnormal{FD}[i]}) \lor
(\nexists p_k \in \textnormal{FD}[i] :   config[k] = \bot) \lor (\notif[k] \neq dfltNtf))$ \hfill$/*$ invariant tests $*/$


{\bf function}  \label{ln:configEstab} $\configEstab(\textnormal{set})=$  \{\lIf{$(\noReconfig() \land (set \notin \{ \textnormal{config}[i], \emptyset \} )  )$}{$\textnormal{\notif}[i] \gets \langle 1, \textnormal{set} \rangle\}$ 
} 



{\bf function}  \label{ln:participate} $participate()=$  \{\lIf{$(\noReconfig())$}{$\textnormal{config}[i] \gets chsConfig()\}$ 
}

\vspace{0.35em}

{\bf Constants and macros}: $dfltNtf = \langle 0,  \bot \rangle$ \label{ln:dfltNtf} \hfill$/*$ {the default notification tuple} $*/$


{\bf macro} $degree(k)$ $=$ {$(2 \cdot \textnormal{\notif}[k].phase + |\{1:\textnormal{all}[k]\}|)$} \hfill$/*$ {$p_k$'s most-recently-received \notif ~degree} $*/$\label{ln:degree}

{\bf macro} \label{ln:eachoNoAll} $echoNoAll(k)$ $= $ \Return{$( \{( \textnormal{FD}[i].part, \textnormal{\notif}[i]) \} = \{ (\textnormal{echo}[j].part, \textnormal{echo}[j].\textnormal{\notif}) \}_{p_j \in \textnormal{FD}[i].part})$ }  

{\bf macro} \label{ln:echoDef} $echo()$ $=$ \Return{$( \{( \textnormal{FD}[i].part, \textnormal{\notif}[i], \textnormal{all}[i]) \} = \{ \textnormal{echo}[j] \}_{p_j \in \textnormal{FD}[i].part})$ } 

{\bf macro} \label{ln:sameK} $same(k)$ $=$ \Return{$( \{( \textnormal{FD}[i].part, \textnormal{\notif}[i]) \} = \{ (\textnormal{FD}[k].part, \textnormal{\notif}[k]) \})$ }

{\bf macro} \label{ln:maxNotif}
 $maxNtf()=$ 
 \{\lIf{$\{ \textnormal{\notif}[k]  \}_{ p_k \in FD[i].part}=\{dfltNtf\}$}{\Return{$\bot$} {\bf else return }$\textnormal{max}_{\textnormal{lex}} \{ \textnormal{\notif}[k]\}_{  p_k \in FD[i].part}\}$}


{\bf macro} \label{ln:confSetVal} $configSet(val)$ $=\{${\bf foreach} {$p_k$} {\bf do} {$(\textnormal{config}[k], \textnormal{\notif}[k]) \gets (val, dfltNtf)\}$ \hfill$/*$ {access to $p_i$'s $\textnormal{config}$}   $*/$}

{\bf macro} \label{ln:incrementphs} $increment(phs)$ $=\{${\bf select}$(phs)$ {\bf case $0$:}  {\bf return} $0$; {\bf case $1$:}  {\bf return} $2$;   {\bf case $2$:}  {\bf return} $0$;$\}$ 

{\bf macro} \label{ln:allSaw} $allSeen() = (\textnormal{FD}[i].part \subseteq (\textnormal{allSeen} \cup \{ p_i :\textnormal{all}[i]\}))$\;

\vspace{0.35em}

{\bf do forever} \label{ln:doForever} \Begin{
\lForEach{$p_k \notin \textnormal{FD}[i].part$}{$(\textnormal{config}[k],\textnormal{\notif}[k]) \gets (\sharp, dfltNtf);$ \hfill$/*$ {clean after crashes} $*/$\label{ln:clean}\DontPrintSemicolon}
%
\PrintSemicolon
\lIf{$((\exists p_k${\em :}$((\textnormal{\notif}[k]=\langle 0, s \rangle) \land (s \neq \bot)) \lor (\textnormal{config}[k] \in \{ \bot, \emptyset \})))   \lor (\nexists p_k, p_{k'} \in \textnormal{FD}[i].part${\em :}$ |degree(k)-degree(k')|> 1) \lor (\exists p_k \in \textnormal{FD}[i].part ${\em :}$ ((\textnormal{\notif}[i].phase =x) \land (\textnormal{\notif}[k].phase =(x+1)(\bmod ~3)) \land (p_k \notin allSeen) \land (x \in \{1, 2 \} )) \lor $ $ (|\notifSet|>1) \land ((\{(\textnormal{FD}[i],\textnormal{FD}[i].part)\}=\{(\textnormal{FD}[k],\textnormal{FD}[k].part)\}_{p_k \in \textnormal{FD}[i].part}) \land ((\textnormal{config}[i]\cap\textnormal{FD}[i].part) = \emptyset)))$ {\bf where} $\notifSet=\{ \textnormal{\notif}[k].set ${\em :}$ \exists p_{k'} \in \textnormal{FD}[i].part${\em :}$\textnormal{\notif}[k']$ $ = \langle 2, \bullet \rangle \}_{p_k \in \textnormal{FD}[i].part}$}{\label{ln:stale}$configSet(\bot)$} 
%
%
%
%
\If(\hfill $/*$ {when no notification arrived} $*/$){$(maxNtf() = \bot)$\label{ln:noMaxNotif}}{
%
\lIf{$|  \{ \textnormal{config}[k] \}_{p_k \in \textnormal{FD}[i]} \setminus \{ \bot, \sharp \} | > 1$}{$configSet(\bot);$ \hfill$/*$ {nullify the configuration upon conflict} $*/$\DontPrintSemicolon }\label{ln:nullConfig} 
\lIf{$ (\textnormal{config}[i] = \bot \land | \{  \textnormal{FD}[j]: p_j \in \textnormal{FD}[i] \} | = 1)$}{$configSet(\textnormal{FD}[i]);$ \hfill$/*$ {reset during admissible runs} $*/$\label{ln:restartConfig}}}
\Else{
%
$\textnormal{all}[i] \gets \bigwedge_{p_k \in \textnormal{FD}[i].part}(echoNoAll(k) \land same(k));$\label{ln:adoptAll} 
\hfill $/*$ {test all-the-same reports} $*/$
\lForEach{$p_k \in \textnormal{FD}[i].part: (echoNoAll(k) \land same(k))$}{$\textnormal{allSeen} \gets \textnormal{allSeen} \cup \{p_k\}$\label{ln:addSaw}}
%
\lIf{$echo() \land allSeen()$}{$( \textnormal{\notif}[i].phase, \textnormal{allSeen}) \gets (increment(\textnormal{\notif}[i].phase),\emptyset)$\label{ln:readyToReplace}}
%
%
%
%
\{{\bf select}$(\textnormal{\notif}[i].phase)$ {\bf case $0$:} {$\textnormal{\notif}[i] \gets dfltNtf$}, {\bf case $1$:} {$\textnormal{\notif}[i] \gets maxNtf()$}, {\bf case $2$:} \label{ln:automatonStep}
 {$\textnormal{config}[i] \gets \textnormal{\notif}[i].set$}\};
%
%
}
\lIf{$\textnormal{config}[i]\neq\sharp$}{{\bf foreach} {$p_j \in \textnormal{FD}[i]$}{ {\bf do send}$(\langle \textnormal{FD}[i], \textnormal{config}[i], \textnormal{\notif}[i], \textnormal{all}[i], (\textnormal{FD}[j].part, \textnormal{\notif}[j], \textnormal{all}[j]) \rangle)$\label{ln:send}}}}
{\bf upon receive} $m=\langle \textnormal{FD}, \textnormal{config}, \textnormal{\textnormal{\notif}},  \textnormal{all}, \textnormal{echo} \rangle $ {\bf from} $p_j$ {\bf do} 
$(\textnormal{FD}[j], \textnormal{config}[j], \textnormal{\notif}[j], \textnormal{all}[j], \textnormal{echo}[j]) \gets m$ \label{ln:receive}\;
%
{\bf upon interrupt} $\bm{p_i}${\bf 's booting} {\bf do} 
\lForEach{$p_k$}{$\textnormal{echo}[k] \gets (\textnormal{config}[k], \textnormal{\notif}[k], \textnormal{all}[k]) \gets (\sharp, dfltNtf, false)$\label{ln:join}}
\end{footnotesize}
\end{algorithm*}

\subsubsection{Algorithm Description}
\label{sec:recSAdescr}


We first present an overview of the algorithm and then proceed to a line-by-line description. 

\subsubsection*{Overview}
The $recSA$ layer uses a self-stabilizing algorithm (Algorithm~\ref{alg:disCongif})  for assuring correct configuration while allowing the updates from the reconfiguration management layer.
Algorithm~\ref{alg:disCongif} guarantees that (1) all active processors have eventually identical copies of a single configuration, (2) when participants notify the system that they wish to replace the current configuration with another, the algorithm selects one proposal and replaces the current configuration with it, and (3) joining processors can become participants eventually. 

The algorithm combines two techniques: one for \emph{brute force stabilization} that recovers from stale information and a complementary technique for \emph{delicate (configuration) replacement}, where participants jointly select a single new configuration that replaces the current one. 
As long as a given processor is not aware of ongoing configuration replacements, Algorithm~\ref{alg:disCongif} merely monitors the system for stale information, e.g., it makes sure that all participants have a single (non-empty) configuration. During these periods the algorithm allows the invocation of configuration replacement processes (via the $estab(\textnormal{set})$ interface, triggered by the Reconfiguration Management layer) as well as the acceptance of joining processors as participants (via the $participate()$ interface, triggered by the Joining layer). During the process of configuration replacement, the algorithm selects a single configuration proposal and replaces the current one with that proposal before returning to monitor for configuration disagreements.

While the system reconfigures, there is no immediate need to allow joining processors to become participants. By temporarily disabling this functionality, the algorithm can focus on completing the configuration replacement using the current participant set. To that end, 
only participants broadcast their states at the end of the do forever loop (line~\ref{ln:send}), and only their messages arrive to the other active processors (line~\ref{ln:receive}). Joining processors receive such messages, but cannot broadcast before their safe entry to the participants' set via the function $participate()$ (line~\ref{ln:participate}), which enables $p_i$'s broadcasting. 
Note that non-participants monitor the intersection between the current configuration and the set of active participants (line~\ref{ln:stale}). In case it is empty, the processors (participants or not) essentially begin a brute-force stabilization (outlined below) where there is no more blocking of joining processors to become participants.



\begin{figure*}[t!] 
	\center
	\includegraphics[scale=0.45,bb=96 401 486 755]{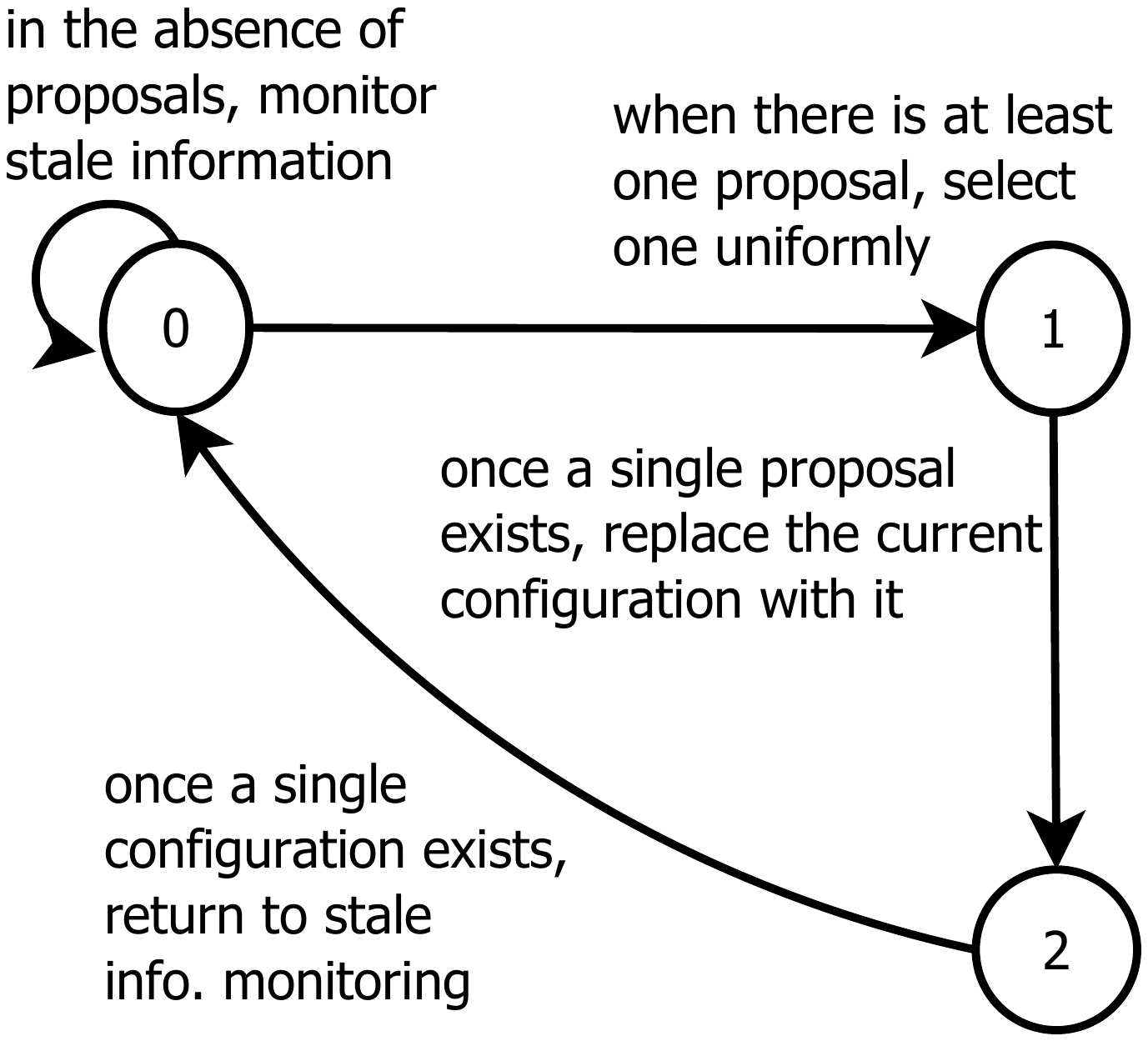}
	\caption{The configuration replacement automaton}
	\label{fig:auto}
\end{figure*}

\noindent {\bf\em Brute-force stabilization.}
The algorithm detects the presence of stale information and recovers from these transient faults. {\em Configuration conflicts} are one of several kinds of such stale information and they refer to differences in the field $\textnormal{config}$, which stores the (quorum) configuration values. Processor $p_i$ can signal to all processors that it had detected stale information by assigning $\bot$ to $\textnormal{config}_i$ and by that start a reset process that nullifies all $\textnormal{config}$ fields in the system (lines~\ref{ln:stale} and~\ref{ln:nullConfig}). 
Algorithm~\ref{alg:disCongif} uses the brute-force technique for letting processor $p_i$ to assign to $\textnormal{config}_i$ its set of trusted processors (line~\ref{ln:restartConfig}), which the failure detector $\textnormal{FD}_i$ at processor $p_i$ provides. Note that by the end of the brute-force process, all active processors (joining or participant) become  participants. We show that 
eventually all active processors share identical (non-$\bot$) $\textnormal{config}$ values by the end of this process.

%
%
%
\noindent {\bf\em Delicate (configuration) replacement.}
Participants can propose to replace the current configuration with a new one, $\textnormal{set}$, via the $estab(\textnormal{set})$ interface. This replacement uses the {\em configuration replacement} process, which for the purposes of the overview, we abstract it as the automaton depicted in Figure~\ref{fig:auto}. 
When the system is free from stale information, the configuration uniformity invariant (of the $\textnormal{config}$ field values) holds. Then, any number of calls to the $estab(\textnormal{set})$ interface starts the configuration replacement process, which controls the configuration replacement using the following three phases: (1) selecting (deterministically and uniformly) a single proposal (while verifying the eventual absence of ``unselected'' proposals), (2) replacing (deterministically and uniformly) all $\textnormal{config}$ fields with the jointly selected proposal, and (3) bringing back the system to a state in which it merely tests for stale information. 
%

The configuration replacement process 
requires coordinated phase transition.
%
%
Algorithm~\ref{alg:disCongif} lets processor $p_i$ to represent proposals as $\textnormal{prp}_i[j]=(phase, set)$, where $p_j$ is the processor from which $p_i$ received the proposal, $phase \in \{0,1, 2\}$ and $set$ is a processor set or the null value, $\bot$. The \emph{default proposal}, $\langle 0, \bot \rangle$, refers to the case in which $\textnormal{prp}$ encodes ``no proposal''. 
When $p_i$ calls the function $estab(\textnormal{set})$, it changes $\textnormal{prp}$ to $\langle 1, set \rangle$ (line~\ref{ln:configEstab}) as long as $p_i$ is not aware of an ongoing configuration replacement process, i.e., $noReco()$ returns true. Upon this change, the algorithm disseminates $\textnormal{prp}_i[i]$ and by that guarantees eventually that $noReco()$ returns false for any processor that calls it. Once that happens, no call to $estab(\textnormal{set})$ adds a new proposal for configuration replacement and no call to $participate()$ lets a joining processor to become a participant (line~\ref{ln:participate}). The algorithm can then use the lexical value of the $\textnormal{prp}_i[]$'s tuples for selecting one of them deterministically (Figure~\ref{fig:auto}). To that end, each participant makes sure that all other participants report the same tuples by waiting until they ``echo'' back the same values as the ones it had sent to them. Once that happens, the participant $p_i$ makes sure that the communication channels do not include other ``unselected'' proposals by raising a flag ($\textnormal{all}_i=true$) and waiting for the echoed values of these three fields, i.e., participant set, $\textnormal{prp}_i[i]$ and $\textnormal{all}_i$. This waiting continues until the echoed values match the values of any other active participant in the system (while monitoring their well-being). Before this participant proceeds, it makes sure that all active participants have noticed its phase completion (line~\ref{ln:addSaw}). Each processor $p$ maintains the $\textnormal{allSeen}$ variable; a set of participants that have noticed $p$'s phase completion (line~\ref{ln:addSaw}) and are thus added to $p$'s $\textnormal{allSeen}$ set.

The above mechanism for phase transition coordination allows progression in a unison fashion. Namely, no processor starts a new phase before it has seen that all other active participants have completed the current phase and have noticed that all other have done so (because they have identical participants' set, $\textnormal{prp}$ and $\textnormal{all}[]$ values). This is the basis for emulating every step of the configuration replacement process (line~\ref{ln:automatonStep}) and making sure that the phase 2 replacement occurs correctly before returning to phase 0, in which the system simply tests for stale information. We show that since the failure detectors monitor the participants' well-being, this process terminates.

\subsubsection*{Detailed Descripion}

We now proceed to a detailed, line-by-line description of Algorithm~\ref{alg:disCongif}.

\paragraph{Variables.}
The algorithm uses a number of fields
that each active participant broadcasts to all other system processors. The processor stores the values that they receive in arrays. Namely, we consider both $p_i$'s own value (the $i$-th entry) and $p_j$'s most recently received value (the $j$-th entry).  

\begin{itemize}

\item The field $\textnormal{config}[]$: an array in which $\textnormal{config}[i]$ is $p_i$'s view on the current configuration. Note that $p_i$ assigns the \emph{empty (configuration)} value $\bot$ after receiving a conflicting (not the same) non-empty configuration value, i.e., the received configuration is different than $p_i$'s configuration. We use the symbol $\sharp$ for denoting that processor $p_i$ is not a participant, i.e., $\textnormal{config}_i[i]=\sharp$.

\item The field $\textnormal{FD}[]$ is  an array in which $\textnormal{FD}[i]$ represents $p_i$'s failure detector of trusted processors (without the list of processors that are suspected to be crashed).

\item $\textnormal{FD}[].part$ is the participants' set, where $\textnormal{FD}[i].part$ is an alias for $\{ p_j \in \textnormal{FD}[i] : \textnormal{config}[j] \neq \sharp \}$ and $\textnormal{FD}[j].part$ refers to the last value received from $p_j$. Namely, the $\textnormal{FD}$ field of every message encodes also this participation information.

\item The field $\textnormal{\notif}[]$ is an array of pairs $\langle phase \in \{0,1,2\}, set \subseteq P \rangle$, where $\textnormal{\notif}[i]$ refers to $p_i$'s configuration replacement notifications. In the pair $\textnormal{\notif}[i]$, the field $set$ can either be $\bot$ (`no value') or the proposed processor set. 

\item The field $\textnormal{echo}[]$ is an array in which $\textnormal{echo}[i]$ is an alias of  $(\textnormal{FD}[i].part, \textnormal{\notif}[i], \textnormal{all}[i])$ and $\textnormal{echo}[j]$ refers to the most recent value that $p_i$ has received from $p_j$ after $p_j$ had responded to $p_i$ with the most recent values it got from $p_i$.

\item The field $\textnormal{all}[]$ is an array of Booleans, where $\textnormal{all}[i]$ refers to the case in which $p_i$ observes that all trusted processors have noticed its current (maximal) notification and they hold the same notification. 
%

\item The variable $\textnormal{allSeen}$ is a set that includes the processors $p_k$ for which $p_i$ received the $\textnormal{all}[k]$ indication. 

\end{itemize}

\paragraph{Constants, functions and macros.}
The constant $dfltNtf$ (line~\ref{ln:dfltNtf}) denotes the default notification tuple $\langle 0,  \bot \rangle$.
The following functions define the interface between Algorithms~\ref{alg:disCongif} and~\ref{alg:SSQR} (Reconfiguration Management layer) and the Joining Mechanism (Algorithm~\ref{alg:join}). Note that the behavior which we specify below considers legal executions.  

\begin{itemize}

\item The function $chsConfig()$ (line~\ref{ln:chsConfig}) returns $\textnormal{config}$ whenever there is a single such non-$\sharp$ value. Otherwise, $\bot$ is returned.

\item The function $getConfig()$ (line~\ref{ln:getConfig}) allows Algorithms~\ref{alg:SSQR}~and~\ref{alg:join} to retrieve the value of the current quorum configuration, i.e., $\textnormal{config}_i[i]$. We note that during legal executions, this value is a set of processors whenever $p_i$ is a participant. However, this value can be $\sharp$ whenever $p_i$ is not a participant and $\bot$ during the process of configuration reset.

\item The function $\noReconfig()$ (line~\ref{ln:\noReconfig}) returns $true$ whenever (1) $p_i$ was not recognized as a trusted processor by a processor that $p_i$ trusts, (2) there are configuration conflicts, (3) the participant sets have yet to stabilize, (4) there is an on-going process of brute force stabilization, or (5) there is a delicate (configuration) replacement in progress. This part of the interface allows Algorithms~\ref{alg:SSQR}~and~\ref{alg:join} to test for the presence of local evidence according to which Algorithm~\ref{alg:disCongif} shall disable delicate (configuration) replacement and joining to the participant set.

\item The function $\configEstab(\textnormal{set})$ (line~\ref{ln:configEstab}) provides an interface that allows the \emph{recMA} layer to request from Algorithm~\ref{alg:disCongif} to replace the current quorum configuration with the proposed $\textnormal{set}$, which is a non-empty group of participants. Note that Algorithm~\ref{alg:disCongif} disables this functionality whenever $\noReconfig()=false$ or $set = \textnormal{config}[i]$.

\item The function $participate()$ (line~\ref{ln:participate}) provides an interface that allows the Joining Mechanism to request from Algorithm~\ref{alg:disCongif} to let $p_i$ join the participant set, which is the group that can participate in the configuration and request the replacement of the current configuration with another (via the $\configEstab(\textnormal{set})$ function). Note that Algorithm~\ref{alg:disCongif} disables this functionality whenever $\noReconfig()=false$ and thus there exists a single configuration in the system, i.e., the call to $chsConfig()$ (line~\ref{ln:chsConfig}) returns the single configuration that all active participants store as their current quorum configuration. This is except for case in which $(\{ \textnormal{config}_i[k] \}_{p_k \in \textnormal{FD}_i[i]} \setminus \{ \sharp \})=\emptyset$. Here, $chsConfig()$ returns $\bot$, which starts a reset process in order to deal with a complete collapse where the quorum system includes no active participants.

\end{itemize}

Algorithm~\ref{alg:disCongif} uses the following macros. 

\begin{itemize}

\item The macro $degree(k)$ (line~\ref{ln:degree}) calculates the degree of $p_k$'s most-recently-received notification degree, which is twice the notification phase plus one whenever all participants are using the same notification (and zero otherwise).

\item The macros $echoNoAll(k)$ and $echo(k)$ (lines~\ref{ln:eachoNoAll}, and~\ref{ln:echoDef} respectively) test whether $p_i$ was acknowledged by all participants for the values it has sent. The former function considers just the fields that are related to its own participants' set and notification, whereas the latter considers also the field $\textnormal{all}[]$. 

\item The macro $same(k)$ (line~\ref{ln:sameK}) performs a similar tests to the one of $echoNoAll(k)$, but considers only processor $p_k$'s most-recently-received values rather than all participants' echoed ones.

\item The macro $maxNtf()$ (line~\ref{ln:maxNotif}) selects a notification with the maximal lexicographical value or returns $\bot$ in the absence of notification that is not the default (phase 0) notification. We define our lexicographical order   between $\textnormal{\notif}_1$ and $\textnormal{\notif}_2$, as $\textnormal{\notif}_1 \leq_{lex}\textnormal{\notif}_2 \iff ((\textnormal{\notif}_1.phase <\textnormal{\notif}_2.phase) \lor ((\textnormal{\notif}_1.phase =\textnormal{\notif}_2.phase) \land (\textnormal{\notif}_1.set \leq_{lex}\textnormal{\notif}_2.set)))$, where $\textnormal{\notif}_1.set \leq_{lex}\textnormal{\notif}_2.set$ can be defined using a common lexical ordering and by considering sets of processors as ordered tuples that list processors in, say, an ascending order with respect to their identifiers.

\item The macro $configSet(val)$ (line~\ref{ln:confSetVal}) acts as a wrapper function for accessing $p_i$'s local copies of the field $\textnormal{config}$. This macro also makes sure that there are no (local) active notifications.

\item The macro $increment(phs)$ (line~\ref{ln:incrementphs}) performs the transition between the phases of the delicate configuration replacement technique.

\item The macro $allSeen()$ (line~\ref{ln:allSaw}) tests whether all active participants have noticed that all other participants have finished the current phase.

\end{itemize}

\boldsubparagraph{The do forever loop.} 
A line-by-line walkthrough of the pseudocode of Algorithm~\ref{alg:disCongif} follows.

\subparagraph{Cleaning up, removal of stale information and invariant testing.} The do forever loop starts by making sure that non-participant nodes cannot have an impact on $p_i$'s computation (line~\ref{ln:clean}) before testing that $p_i$'s state does not include any stale information (line~\ref{ln:stale}). Algorithm~\ref{alg:disCongif} tests for several types of stale information (c.f. Definition~\ref{def:type}): (type-1) notifications in phase 0 must not have a $set \neq \emptyset$, (type-2) configurations that refer to the empty set, execute configuration reset, or conflicting  configurations, (type-3) the degree gap between two notifications is greater than one, there are participants in different  phases but the one in the more advanced phase does not appear in the $\notifSet$ set, or the local set of notifications includes more than one notification and at least one of them is in phase 2, and (type-4)
the quorum configuration includes at least one active participant (to avoid false positive error by testing only when processor $p_i$ has a stable view of set of trusted processors and participants). In case any of these tests fails, the algorithm starts a configuration reset process by calling $configSet(\bot)$.

\subparagraph{The brute-force stabilization technique.}
As long as no active notifications are present locally (line~\ref{ln:noMaxNotif}), every processor performs this technique for transient fault recovery. In the presence of configuration conflicts, the algorithm starts the configuration reset process (line~\ref{ln:nullConfig}). Moreover, during the configuration reset process, the algorithm waits until all locally trusted processors report that they trust the same set of processors (line~\ref{ln:restartConfig}).

\subparagraph{The delicate replacement technique --- phase synchronization.}
This technique synchronizes the system phase transitions by making sure that all active participants work with the same notification.

Each active participant tests whether all other trusted participants have echoed their current participants' set and notifications and have the same values with respect to these two fields     
(line~\ref{ln:adoptAll}). The success of this test assigns $true$ to the field $\textnormal{all}_i[i]$. The algorithm then extends this test to include also the field $\textnormal{all}[]$ (line~\ref{ln:addSaw}). Upon the success of this test with respect to participant $p_k$, the algorithm adds $p_k$ to the set $\textnormal{allSeen}_i$.
Once processor $p_i$ receives reports from all participants that the current phase is over, it moves to the next phase (line~\ref{ln:readyToReplace}).

\subparagraph{The delicate replacement technique --- finite-state-machine emulation.}
Each of the three phases represent an automaton state (line~\ref{ln:automatonStep}); recall Figure~\ref{fig:auto} and the configuration replacement process discussed
in the Overview.  
During phase $1$, the system converges to a single notification. During phase $2$, the system replaces the current configuration with the proposed one. Next, the system returns to its ideal state, i.e., phase $0$, which allows new participants to join, as well as further reconfigurations.

\subparagraph{Message exchange and the control of newly arrived processors.}
When a participant finishes executing the do forever loop, it broadcasts its entire state (line~\ref{ln:send}). Once these messages arrive, processor $p_i$ stores them (line~\ref{ln:receive}). The only way for a newly arrived processor to start executing Algorithm~\ref{alg:disCongif} is by responding to an interrupt call (line~\ref{ln:join}). This procedure notifies the processor state and makes sure that it cannot broadcast messages (line~\ref{ln:send}). The safe entry of this newly arrived processor to the participants' set, is via the function $participate()$ (line~\ref{ln:participate}), which enables $p_i$'s broadcasting. Thus, non-participants merely follow the system messages until the function $\noReconfig()$ returns true and allows its join to the participant set by a call (from the Joining Mechanism) to the function $participate()$ (line~\ref{ln:participate}).

\subsubsection{Correctness}
\label{sec:recSAproof}

We first provide an outline of the proof and then proceed in steps to establish the correctness of the algorithm.
\subsubsection*{Outline} 
We say that system state $c$ has no stale information when (1) all the notifications are valid, (2) there are no configuration  conflicts or active reset process, (3) the phase information (including the set $allSeen$) are not out of synch, and (4) there are active participants  in $\textnormal{config}$. The correctness proof of Algorithm~\ref{alg:disCongif} shows that eventually there is no stale information, because they are all cleaned (line~\ref{ln:clean}) or detected and cause configuration reset by calling $configSet(\bot)$ (line~\ref{ln:stale}).
 

In the absence of notifications (line~\ref{ln:noMaxNotif}), 
Algorithm~\ref{alg:disCongif} merely tests for configuration conflicts (line~\ref{ln:nullConfig}) and during a configuration reset, it waits until all locally trusted processors report that they trust the same processors (line~\ref{ln:restartConfig}) before these processors become the configuration. 
I.e., it applies the brute force stabilization technique. The proof here shows that during admissible executions, the $\bot$ symbol propagates to all of the $\textnormal{config}$ fields in the system until all active processors assign to $\textnormal{config}$ either $\bot$ or the entire set of trusted processors. 
This reset process ends when all processors assign merely the latter value to $\textnormal{config}$.

In the presence of notifications, Algorithm~\ref{alg:disCongif}  synchronizes the system phase transitions by making sure that all active participants work with the same notification (lines~\ref{ln:adoptAll} to~\ref{ln:addSaw}) before moving to the next phase (line~\ref{ln:readyToReplace}). The algorithm's three phases represent an automaton state (line~\ref{ln:automatonStep}). 
During phase $1$, the system converges to a single notification. While $phase=2$ holds, the system replaces the current configuration with the proposed one. Next, the system returns to its no-notifications state, i.e., $phase=0$, which allows new participants as well as further reconfigurations. 
When a participant finishes the do forever loop, it broadcasts its entire state (line~\ref{ln:send}).
Once these messages arrive, the receiving processor stores them (line~\ref{ln:receive}). 

The only way for a newly arrived processor to start executing Algorithm~\ref{alg:disCongif} is by responding to an interrupt call (line~\ref{ln:join}). This procedure nullifies the state of the newly arrived processor with $\textnormal{config}=\sharp$ and by that makes sure that it cannot broadcast (line~\ref{ln:send}). The safe entry of this processor to the participating set, is via the function $participate()$ (line~\ref{ln:participate}), which enables $p_i$'s broadcasting. Thus, non-participants merely follow the system messages until the function $\noReconfig()$ returns true and allows its join to the participant set by a call to $participate()$ (line~\ref{ln:participate}). 
By controlling the new joins, the three phase structure is the basis for 
proving the final theorem stating that eventually a common configuration set is adopted by every active processor. We now proceed to the detailed proof.


%
%
\subsubsection*{Configuration conflicts and stale information}
%

We begin by classifying the stale information into four types.

\begin{definition}
\label{def:type}
We say that processor $p_i$ in system state $c$ has a stale information in $c$ of:
\begin{itemize}[topsep=2pt,itemsep=-.5ex,partopsep=.5ex,parsep=1ex,leftmargin=.5cm]
	\item[]{\bf type-1:} when 
$(\exists p_k : ((\textnormal{prp}[k]=\langle 0, s \rangle) \land (s \notin \{ \emptyset, \bot\})))$ (cf. line~\ref{ln:stale}). 
	\item[]{\bf type-2:} when $(\exists p_k : (\textnormal{config}[k] \in \{ \bot, \emptyset\}))$ (cf. line~\ref{ln:stale}) or $c$ encodes a \emph{(configuration) conflict}, i.e., there are two active processors $p_i$ and $p_j$ for which $\textnormal{config}_i[i] \neq \textnormal{config}_j[j]$, $\textnormal{config}_i[i] \neq \textnormal{config}_i[j]$, or $m_{j,i}.\textnormal{config}\neq\textnormal{config}_i[i]$ in any message in the communication channel from $p_i$ to $p_j$.
	\item[]{\bf type-3:} when $(\nexists p_k, p_{k'} \in \textnormal{FD}_i[i].part: |degree_i(k)-degree_i(k')|> 1) \lor (\exists p_k \in \textnormal{FD}_i[i] : ((\textnormal{\notif}_i[i].phase =x) \land (\textnormal{\notif}_i[k].phase =(x+1)(\bmod ~3)) \land (p_k \notin allSeen_i) \land (x \in \{1, 2 \} ))\lor (|\notifSet_i|>1))$, where $\notifSet_i=\{ \textnormal{\notif}_i[k].set : \exists p_{k'} \in \textnormal{FD}_i[i].part:\textnormal{\notif}_i[k']$ $ = \langle 2, \bullet \rangle \}_{p_k \in \textnormal{FD}_i[i]}$. 
	\item[]{\bf type-4:} when $((\{(\textnormal{FD}_i[i],\textnormal{FD}_i[i].part)\}=\{(\textnormal{FD}_i[k],\textnormal{FD}_i[k].part)\}_{p_k \in \textnormal{FD}_i[i].part}) \land ((\textnormal{config}_i[i]\cap\textnormal{FD}_i[i].part) = \emptyset))$.
%
\end{itemize}
\end{definition}


\begin{claim}[Eventually there is no type-1 stale information]
\label{thm:noStale}
Eventually the system reaches a state $c \in R$ in which the invariant of no type-1 stale information holds thereafter.
\end{claim}
\begin{proof}
Let $c \in R$ be a system state in which processor $p_i$ has an applicable step $a_i$ that includes the execution of the do forever loop (line~\ref{ln:doForever} to~\ref{ln:send}).
We note that immediately after $a_i$, processor $p_i$ has no type-1 stale information (line~\ref{ln:stale} removes it). Moreover, that removal always occurs before $a_i$ sends any message $m$ (line~\ref{ln:send}). Therefore, eventually for every active processor $p_j$ that receives from $p_i$ message $m$ (line~\ref{ln:receive}), it holds that $(m.phase = 0) \iff (m.\textnormal{\notif} = (0, \bot))$ as well as for every item $\textnormal{\notif}_i[k]: p_k \in P$. Once this invariant holds for every pair of active processors $p_i$ and $p_j$, the system reaches state $c$. We conclude the proof by noting that Algorithm~\ref{alg:disCongif} never assigns to $\textnormal{\notif}_i[j]$ a values that violates this claim invariant.    
\end{proof}

\subsubsection*{Replacement state and message; explicit and spontaneous replacements}
We say that a processor $p_i$'s state encodes a \emph{(delicate) replacement} when $\textnormal{\notif}_i[j] \neq \langle 0, \bot \rangle$ and we say that a message $m_{i,j}$ in the channel from $p_i$ to $p_j$ encodes a \emph{(delicate) replacement} when its $\textnormal{\notif} \neq\langle 0, \bot  \rangle$. 
Given a system execution $R$, we say that $R$ \emph{does not include an explicit (delicate) replacement} when throughout $R$ no node $p_i$ calls $\configEstab()$. Suppose that execution $R$ does not include an explicit (delicate) replacement and yet there is a system state $c \in R$ in which a processor state or a message in the communication channels encodes a (delicate) replacement. In this case, we say that $R$ includes a \emph{spontaneous (delicate) replacement}.



%

\begin{lemma}[Eventually there is no type-2 stale information]
\label{thm:noConflict}
Let $R$ be an admissible execution that does not include explicit (delicate) replacements nor spontaneous ones. The system eventually reaches  a state $c \in R$ in which the invariant of no type-2 stale information holds thereafter. 
\end{lemma}
\begin{proof}
Note that any of $R$'s steps that includes the do forever loop (line~\ref{ln:doForever} to~\ref{ln:send}) does not run lines~\ref{ln:adoptAll} to~\ref{ln:automatonStep} (since $R$ does not include an explicit nor spontaneous replacement). If $R$'s starting system state does not include any configuration conflicts, we are done. Suppose that $R$'s starting system state does include a conflict, i.e., $\exists p_i, p_j \in P: (\textnormal{config}_i[i] = \bot)  \lor  (\textnormal{config}_i[i] \neq \textnormal{config}_i[j]) \lor (\textnormal{config}_i[i] \neq \textnormal{config}_j[j])$ or there is a message, $m_{i,j}$, in the communication channel from $p_i$ to $p_j$, such that the field $(m_{i,j}.\textnormal{config}[k] = \bot) : p_k \in \textnormal{FD}_i[i] \lor (m_{i,j}.\textnormal{config}\neq\textnormal{config}_i[i])$, where both $p_i$ and $p_j$ are active processors. 
In Claims~\ref{thm:thereBotSim},~\ref{thm:thereBot} and~\ref{thm:onceBot} we show that in all of these cases, eventually $\forall p_i \in P: \textnormal{config}_i[i] \in \{ \bot, \textnormal{FD}_i[i] \}$ before showing that eventually there are no configuration conflicts  (Claim~\ref{thm:2Config}).


Claims~\ref{thm:thereBotSim},~\ref{thm:thereBot} and~\ref{thm:onceBot} consider the values in the field $\textnormal{config}$ that are either held by an active processor $p_i \in P$ or in its outgoing communication channel to another active processor $p_j \in P$. We define the set $S= \{ S_i \cup S\_out_{i} \}_{p_i \in P}$ to be the set of all these values, where $S_i = \{ \textnormal{config}_i[j] \}_{p_j \in \textnormal{FD}_i[i]}$ and $S\_out_{i} = \{ m_{i,j}.\textnormal{config} \}_{p_j \in \textnormal{FD}_i[i]}$. 

%

\begin{claim}
\label{thm:thereBotSim}
Suppose that in $R$'s starting system state, there are  processors $p_i, p_j \in P$ that are active in $R$, for which $|S  \setminus \{ \bot, \sharp \} |>1$, where $\exists S' \subseteq S : S' \in \{ \{ \textnormal{config}_i[i], \textnormal{config}_i[j] \}, \{ \textnormal{config}_i[i], m_{i,j}.\textnormal{config} \} \}$. Eventually the system reaches a state in which $\textnormal{config}_i[i] \in \{ \bot, \textnormal{FD}_i[i] \}$ holds. 
\end{claim}

\begin{proof}
Suppose that $S' = \{ \textnormal{config}_i[i], \textnormal{config}_i[j] \}$ holds. Immediately after $R$'s starting state, processor $p_i$ has an applicable step that includes the execution of the do forever loop (line~\ref{ln:doForever} to~\ref{ln:send}). In that step, the if-statement condition 
%
%
$(|  \{ \textnormal{config}_i[k] : p_k \in \textnormal{FD}_i[i] \}  \setminus \{ \bot, \sharp \} | > 1)$
(line~\ref{ln:nullConfig}'s if-statement) holds, $p_i$ assigns $\bot$ to $\textnormal{config}_i[i]$ and the proof is done. Suppose that $S' = \{ \textnormal{config}_i[i], m_{i,j}.\textnormal{config} \}$ holds. Upon $m_{i,j}$'s arrival, processor $p_i$ assigns $m_{i,j}.\textnormal{config}$ to $\textnormal{config}_i[j]$ (line~\ref{ln:receive}) and the case of $S' = \{ \textnormal{config}_i[i], \textnormal{config}_i[j] \}$ holds.
\end{proof}

\begin{claim}
\label{thm:thereBot}
Suppose that in $R$'s starting system state, there are processors $p_i, p_j \in P$ that are active in $R$, for which $|S  \setminus \{ \bot, \sharp \} |>1$, where $\exists S' \subseteq S : S' \in \{ \{ \textnormal{config}_i[i], \textnormal{config}_j[j] \} \}$. Eventually the system reaches a state in which  $\textnormal{config}_i[i] \in \{ \bot, \textnormal{FD}_i[i] \}$ or $\textnormal{config}_j[j] \in \{ \bot, \textnormal{FD}_i[i] \}$ holds. 
\end{claim}

\begin{proof}
%
%
Suppose, towards a contradiction, for any system state $c \in R$ that neither $\textnormal{config}_i[i] \in \{ \bot, \textnormal{FD}_i[i] \}$ nor $\textnormal{config}_j[j] \in \{ \bot, \textnormal{FD}_i[i] \}$. Note that $p_i$ and $p_j$ exchange messages eventually, because whenever processor $p_i$ repeatedly sends the same message to processor $p_j$, it holds that $p_j$ receives that message eventually (the fair communication assumption, Section~\ref{s:sys}) and vice versa. Such message exchange implies that the case of $|S  \setminus \{ \bot, \sharp \} |>1$ (Claim~\ref{thm:thereBotSim}) holds eventually, where $\exists S' \subseteq S : S' \in \{ \{ \textnormal{config}_i[i], m_{i,j}.\textnormal{config} \}, \{ \textnormal{config}_j[j], m_{i,j}.\textnormal{config} \}\}$. Thus, we reach a contradiction and therefore eventually the system reaches a state in which 
 $\textnormal{config}_i[i] \in \{ \bot, \textnormal{FD}_i[i] \}$ or $\textnormal{config}_j[j] \in \{ \bot, \textnormal{FD}_i[i] \}$ hold.
\end{proof}


\begin{claim}
\label{thm:onceBot}
Suppose that in $R$'s starting system state, there is a processor $p_i \in P$ that is active in $R$, for which $\textnormal{config}_i[i]  \in \{ \bot, \textnormal{FD}_i[i] \}$. (1) For any system state $c \in R$, it holds that $\textnormal{config}_i[i] \in \{ \bot, \textnormal{FD}_i[i] \}$. Moreover, (2) $R=R'\circ R''$ has a suffix, $R''$, for which $\forall c'' \in R'':\forall p_i, p_j : (\{ m_{i,j}.\textnormal{config}, \textnormal{config}_j[i], \textnormal{config}_j[j] \} \setminus \{ \bot, \textnormal{FD}_i[i] \}) = \emptyset$.
\end{claim}

\begin{proof}
%
We prove each part of the statement separately. 

\noindent {\bf Part (1).~}
We start the proof by noting that $\forall p_i, p_j \in P$, it holds that, throughout $R$, we have that $\textnormal{FD}_i[i]$'s value does not change and that $\textnormal{FD}_i[i] = \textnormal{FD}_j[j]$, because this lemma assumes that $R$ is admissible.
To show that $\textnormal{config}_i[i] \in \{ \bot, \textnormal{FD}_i[i] \}$ holds in any $c \in R$, we argue that any step $a_i \in R$ in which $p_i$ changes $\textnormal{config}_i[i]$'s value includes the execution of line~\ref{ln:nullConfig} or line~\ref{ln:restartConfig} (see the remark at the beginning of this lemma's proof about $a_i \in R$ not including the execution of lines~\ref{ln:adoptAll} to~\ref{ln:automatonStep}), which assign to $\textnormal{config}_i[i]$ the values $\bot$, and respectively, $\textnormal{FD}_i[i]$. 

\noindent {\bf Part (2).~} In this part of the proof, we first consider the values in $m_{i,j}.\textnormal{config}$ and $\textnormal{config}_j[i]$ before considering the one in $\textnormal{config}_j[j]$.

\noindent {\bf Part (2.1).~} 
To show that in $c'' \in R''$ it holds that $\forall p_i, p_j : \{ m_{i,j}.\textnormal{config}, \textnormal{config}_j[i]  \} \setminus \{ \bot, \textnormal{FD}_i[i] \} = \emptyset$, we note that when $p_i$ loads a message $m_{i,j}$ (line~\ref{ln:send}) before sending to processor $p_j$, it uses $\textnormal{config}_i[i]$'s value for the field $\textnormal{config}$. Thus, eventually $m_{i,j}.\textnormal{config}  \in \{ \bot, \textnormal{FD}_i[i] \}$ and therefore $\textnormal{config}_j[i] \in \{ \bot, \textnormal{FD}_i[i] \}$ records correctly in $c'' \in R''$ the most recent $m_{i,j}$'s value that $p_j$ receives from $p_i$ (line~\ref{ln:receive}).

\noindent {\bf Part (2.2).~}
To show that eventually $\textnormal{config}_j[j] \in \{ \bot, \textnormal{FD}_i[i] \}$, we note that once $p_j$ changes the value of $\textnormal{config}_j[j]$, it holds that  $\textnormal{config}_j[j] \in \{ \bot, \textnormal{FD}_i[i] \}$ thereafter (due to the remark in the beginning of this lemma, which implies that only lines~\ref{ln:nullConfig} and~\ref{ln:restartConfig} can change $\textnormal{config}_j[j]$, and by the part (1) of this claim's proof while replacing the index $i$ with $j$). Suppose, towards a contradiction, that $p_j$ does not change that value of $\textnormal{config}_j[j]$ throughout $R$ and yet $\textnormal{config}_j[j] \notin \{ \bot, \textnormal{FD}_i[i] \}$. 
%
%
Note that $(| \{  \textnormal{FD}[j]: p_j \in \textnormal{FD}[i] \} | = 1)$ (see the second clause of the if-statement condition in line~\ref{ln:restartConfig}) holds throughout $R$, because $R$ is admissible. Therefore, whenever $p_i$ takes a step that includes the execution of the do forever loop (line~\ref{ln:doForever} to~\ref{ln:send}), processor $p_i$ sends to $p_j$ the message $m_{i,j}$, such that $m_{i,j}.\textnormal{config}=\textnormal{config}_i[i]$ (line~\ref{ln:send}) and $\textnormal{config}_i[i]=\textnormal{FD}_i[i]$ (see part (2.1) of this proof). Since $p_i$ sends $m_{i,j}$ repeatedly, processor $p_j$ receives eventually $m_{i,j}$ (the fair communication assumption, Section~\ref{s:sys}) and stores in $\textnormal{config}_j[i] = m_{i,j}.\textnormal{config} = \textnormal{config}_i[i] = \textnormal{FD}_i[i] \neq \bot$. Immediately after that step, the system state allows $p_j$ to take a step in which the condition 
%
%
$(|  \{ \textnormal{config}_j[k] : p_j' \in \textnormal{FD}_j[k] \}  \setminus \{ \bot, \sharp \} | > 1)$
(line~\ref{ln:nullConfig}'s if-statement) holds and $p_j$ changes $\textnormal{config}_j[j]$'s value to $\bot$. Thus, this part of the proof end with a contradiction, which implies that the system reaches a state in which $\textnormal{config}_j[j] \in \{ \bot, \textnormal{FD}_i[i] \}$.
\end{proof}

\begin{claim}
\label{thm:2Config}
Suppose that in $R$'s starting system state, it holds that for every two  processors $p_i, p_j \in P$ that are active in $R$, we have that $(\{ \textnormal{config}_i[i], \textnormal{config}_j[i], m_{i,j}.\textnormal{config} \} \setminus \{ \bot, \textnormal{FD}_i[i] \}) = \emptyset$, where $m_{i,j}$ is a message in the channel from $p_i$ to $p_j$. Eventually the system reaches a state in which $\textnormal{config}_i[i]=\textnormal{FD}_i[i]$.  
\end{claim} 

\begin{proof}
%
\sloppy{By this claim assumptions, we have that in $R$'s starting system state, the if-statement condition 
%
%
$(|  \{ \textnormal{config}_i[k] : p_k \in \textnormal{FD}_i[k] \}  \setminus \{ \bot, \sharp \} | > 1)$
(line~\ref{ln:nullConfig}) does not hold.} Moreover, $|\{  \textnormal{FD}
_i[j]: p_j \in \textnormal{FD}_i[i] \}|=1$ (line~\ref{ln:restartConfig}) holds throughout $R$, because $R$ is admissible. Therefore, this claim assumptions with respect to $R$'s starting states actually hold for any system state $c \in R$, because only lines~\ref{ln:nullConfig} and~\ref{ln:restartConfig} can change the value of $\textnormal{config}_i[i] \in \{ \bot, \textnormal{FD}_i[i] \}$, which later $p_i$ uses for sending the message $m_{i,j}$ (line~\ref{ln:send}), and thus also $m_{i,j}.\textnormal{config}  \in \{ \bot, \textnormal{FD}_i[i] \}$ as well as $\textnormal{config}_j[i] \in \{ \bot, \textnormal{FD}_i[i] \}$ records correctly the most recent $m_{i,j}$'s that $p_j$ receives from $p_i$ (line~\ref{ln:receive}).    

To conclude this proof, we note that immediately after any system state $c \in R$, processor $p_i$ has an applicable step $a_i \in R$ that includes the execution of line~\ref{ln:restartConfig} (by similar arguments to the ones used by the first part of this proof). Moreover, $a_i$ does not include the execution of line~\ref{ln:nullConfig}, because by the first part of this proof the condition of the if-statement of line~\ref{ln:nullConfig} does not hold. In the system state that immediately follow $a_i$, the invariant $\textnormal{config}_i[i] = \textnormal{FD}_i[i]$ holds. 
\end{proof}

By this lemma's assumption, there is no configuration $c \in R$ replacement state nor replacement message. Claim~\ref{thm:2Config} implies that the system reaches a state $c_{noConf} \in R$ that has no configuration conflict eventually. Thus, $c_{noConf}$ is safe.   
\end{proof}

\begin{claim}[Eventually there is no type-4 stale information]
\label{thm:noStale4}
Let $R$ be an admissible execution of Algorithm~\ref{alg:disCongif}.
Eventually the system reaches a state $c \in R$ in which the invariant of no type-4 stale information holds thereafter.
\end{claim}
\begin{proof}
Without the loss of generality, suppose that there is no system state in $R$ that encodes a configuration conflict. (We can make this assumption without losing generality because Lemma~\ref{thm:noConflict} implies that this claim is true whenever this assumption is false.) Moreover, let $c \in R$ be a system state in which processor $p_i$ has an applicable step $a_i$ that includes the execution of the do forever loop (line~\ref{ln:doForever} to~\ref{ln:send}). 

Since $R$ is admissible, 
$(\{(\textnormal{FD}_i[i],\textnormal{FD}_i[i].part)\}=\{(\textnormal{FD}_i[k],\textnormal{FD}_i[k].part)\}_{p_k \in \textnormal{FD}_i[i].part})$ holds in $c$. Therefore, the case in which $((\textnormal{config}_i[i]\cap\textnormal{FD}_i[i].part) = \emptyset)$ in $c$ implies a call  to the function $configSet(\bot)$ (line~\ref{ln:stale}) in the step that immediately follows. By using Lemma~\ref{thm:noConflict}, we have that this lemma is true.
\end{proof}

\subsubsection*{Phase and degree progressions}
Let $R$ be an execution of Algorithm~\ref{alg:disCongif} that is admissible with respect to the participant sets. Suppose that $p_i$ is a processor that is active in $R$ and that $p_i \in \textnormal{FD}[i].part$. We say that processor $p_i$ is an active participant in $R$.
%
%
Let $p_i,p_j,p_k \in P$ be processors that are active participants in $R$ and $c \in R$ a system state. 
\sloppy{We denote by $\textnormal{NA}(c) = \{ (\textnormal{\notif}_j[k],\textnormal{all}_j[k]) \}_{ p_j, p_k \in \textnormal{FD}_i[i]} \cup \{ (\textnormal{\notif}_j[k], \textnormal{all}_j[k])  : m_{j,k}= \langle \bullet, \textnormal{\notif}_j[k], \textnormal{all}_j[k], (\bullet) \rangle \in channel_{k,j} \}_{ p_j, p_k \in \textnormal{FD}_i[i]} \cup \{ (\textnormal{echo}_{j}[k].\textnormal{\notif}, \textnormal{echo}_{j}[k].\textnormal{all})  : m_{j,k}=  \langle \bullet, (\bullet, \textnormal{echo}_{j}[k].\textnormal{\notif}, \textnormal{echo}_{j}[k].\textnormal{all}) \rangle \in channel_{k,j} \}_{ p_j, p_k \in \textnormal{FD}_i[i]} \setminus \{ (\langle 0, \bot \rangle, \bullet) \}$ the set of all pairs that includes the notification and all fields that appear in $c$ (after excluding the default notification, $\langle 0, \bot \rangle$, while including all the information in the processors' states and communication channels as well as their replications, e.g., the echo field).}
We denote by $N(c) = \{ \textnormal{n} : 
(\textnormal{n}, \textnormal{a}) \in \textnormal{NA}(c) \}$ the set of notifications that appear in $c$.
We denote the \emph{degree set} of notification $\textnormal{n} \in N(c)$ by $D(c,\textnormal{n}) = \{ 2 \cdot \textnormal{n}.phase + |\{1 : \textnormal{a} \}| :
 (\textnormal{n}, \textnormal{a}) \in \textnormal{NA}(c) \}$. For a given system state $c \in R$, notification $\textnormal{n} \in N(c)$ that appears in $c$, we denote by $S(c,\textnormal{n})=\{ \textnormal{n}' : ((\textnormal{n}' \in N(c)) \land (\textnormal{n}.set = \textnormal{n}'.set)) \}$ the set of all notifications $\textnormal{n}' \in N(c)$ that have the same set filed as the one of $\textnormal{n}$ and $S(c)=\{ \textnormal{n}.set : \textnormal{n} \in N(c)\}$ is the set of all notification sets in $c$.

%

\begin{lemma}
\label{thm:convDeg}
Let $R$ be an execution of Algorithm~\ref{alg:disCongif} that is admissible with respect to the participant sets and that does not include an explicit (delicate) replacement. 
%
%
Suppose that in $R$'s starting system state, $c$, there are notifications, i.e., $N(c) \neq \emptyset$. Let $\textnormal{n} \in N(c)$ be a notification for which it holds that $\forall \textnormal{n}' \in N(c) : \textnormal{n}' \leq_{\textnormal{lex}} \textnormal{n}$ in $c$ (recall the definition of $\leq_{\textnormal{lex}}$ in the description of the $maxNtf()$ macro).
%
%
The system reaches eventually a state $c_{\nexists \textnormal{n}} \in R$ in which $S(c_{\nexists \textnormal{n}},\textnormal{n})=\emptyset$.
%
%
\end{lemma}
\begin{proof}
%
%
%
%
Let $a \in R$ be a step that immediately precedes the system state $c \in R$ and $c' \in R$ be the system state that immediately follows $a$.
Note that when the system reaches a system state $c' \in R$ in which $\exists \textnormal{n}: (S(c,\textnormal{n}) \cap S(c',\textnormal{n})) \neq \emptyset$, the proof is done.
Suppose, in the way of a proof by contradiction, that this lemma is false, i.e., $\nexists c_{\nexists \textnormal{n}} \in R:  (S(c,\textnormal{n}) \cap S(c_{\nexists \textnormal{n}},\textnormal{n})) = \emptyset$, i.e., the notification $\textnormal{n}$ (and additional  notifications that have the same proposed set but different phase number) does not ``disappear'' from the system after any step and ``become'' the default notification $\langle 0, \bot\rangle$. 
%
%
Claims~\ref{thm:general},~\ref{thm:notDec} and~\ref{thm:invar} are needed for the proof of Claim~\ref{thm:contra}, which completes this proof by contradiction.

\begin{claim}
\label{thm:general}
Suppose that $\nexists c_{\nexists \textnormal{n}} \in R:  (S(c,\textnormal{n}) \cap S(c_{\nexists \textnormal{n}},\textnormal{n})) = \emptyset$.
Without the loss of generality, we can assume that (1) no step in $R$ includes a call to $configSet(\bot)$, (2) $\textnormal{n}$'s lexical value in $R$'s starting system state, $c$, is the greater or equal to the lexical value of any notification in any system state $c'' \in  R$, (3) $\textnormal{\notif}_{i}[i] = \textnormal{n}$ in $R$'s starting system state, where $p_i \in P$ is an active participant in $R$, and (4) step $a \in R$ that immediately follows $R$'s starting system state includes the execution of the do forever loop (line~\ref{ln:doForever} to~\ref{ln:send}). 
\end{claim}
\begin{proof}
	We prove each part of the claim separetly. 
	
\noindent {\bf Part (1).~}
This claim's assumption that $\nexists c_{\nexists \textnormal{n}} \in R:  (S(c,\textnormal{n}) \cap S(c_{\nexists \textnormal{n}},\textnormal{n})) = \emptyset$ implies that no step in $R$ includes a call to $configSet(\bot)$, say, due to the if-statement condition of lines~\ref{ln:stale} or~\ref{ln:nullConfig} hold in any system state $c'' \in R$. The reason is that Claim~\ref{thm:thereBot} implies that a call to $configSet(\bot)$ brings the system to a state  $c''' \in R$ in which $S(c''',\textnormal{n})=\emptyset$. Namely, $c''' = c_{\nexists \textnormal{n}}$ and the proof of this  lemma is done.

\noindent {\bf Part (2).~}
We can choose a suffix of $R$, such that $\textnormal{n}$'s  lexical value in $R$'s starting system state, $c$, is the greater or equal to the lexical value of any notification in any system state $c'' \in  R$. We can do that because there is a bounded number of possible lexical values and our assumption that $\nexists c_{\nexists \textnormal{n}} \!\in\! R\!:\!  (S(c,\textnormal{n}) \cap S(c_{\nexists \textnormal{n}},\textnormal{n})) = \emptyset$. 

\noindent {\bf Part (3).~}
We note that when considering the case in which $c$ encodes a message $m_{i,j}$ that has $\textnormal{n}$ in one of its fields, message $m_{i,j}$ reaches eventually from $p_j$ to $p_i$ (the fair communication assumption, Section~\ref{s:sys}). Therefore, we suppose that notification $\textnormal{n}$ is encoded in $p_i$'s state, i.e., $\textnormal{\notif}_{i}[j] = \textnormal{n}$ in $c$, where $p_i, p_j \in P$ are active participants in $R$. We can make this assumption, without the loss of generality, because we can take a suffix of $R$ (that system reaches eventually) in which this assumption holds for its starting configuration and then we take that suffix to be the execution that this lemma considers, i.e., $R$.
Using similar arguments about generality, we also assume that step $a$ includes the execution of the do forever loop (line~\ref{ln:doForever} to~\ref{ln:send}), and thus $\textnormal{\notif}_{i}[i] = \textnormal{n}$ in $c'$. 

\noindent {\bf Part (4).~}
Use the same arguments as for part (3) of this proof.
\end{proof}

\begin{claim}
\label{thm:notDec}
Suppose that step $a$ includes the execution of the do forever loop (line~\ref{ln:doForever} to~\ref{ln:send}) immediately after and before the system states $c$, and respectively, $c'$.  When $\textnormal{\notif}_{i}[i].set = \textnormal{n}.set$ holds in $c$, we have that $\textnormal{\notif}_{i}[i].phase \geq \textnormal{n}.phase$ does not decrease in $c'$. 
\end{claim}
\begin{proof}
Since $\textnormal{\notif}_{i}[i].phase$ changes only in lines~\ref{ln:readyToReplace} and~\ref{ln:automatonStep} (case 1), the assumption that $\nexists c_{\nexists \textnormal{n}} \in R:  (S(c,\textnormal{n}) \cap S(c_{\nexists \textnormal{n}},\textnormal{n})) = \emptyset$ in any $c'' \in R$ implies that $\textnormal{\notif}_{i}[i] = \textnormal{n}$ in any $c'' \in R$ that is after $c'$. The reason is that line~\ref{ln:automatonStep} indeed does not decrease $\textnormal{\notif}_{i}[i].phase$ (part (2) of Claim~\ref{thm:general}) and line~\ref{ln:readyToReplace} only decreases $\textnormal{\notif}_{i}[i].phase$ when assigning $0$ (cf. case $0$ of the function $increment()$). However, the latter cannot occur for any $p_i \in R$ that is an active participant in $R$ and for which $\textnormal{\notif}_{i}[i] = \textnormal{n}$ holds in $c'$ (part (3) of Claim~\ref{thm:general}).~\end{proof}

\begin{claim}
\label{thm:alwaysN}
For every system state $c'' \in R$, it holds that $\textnormal{\notif}_i[i] = \textnormal{n}$.
\end{claim}
\begin{proof}
Since $\textnormal{\notif}_{i}[i] = \textnormal{n}$ holds in $c$ (Claim~\ref{thm:general}), since $\textnormal{\notif}_{i}[i].phase$ is non-decreasing during step $a$ (Claim~\ref{thm:notDec}), since $\textnormal{n}$ is lexically greater or equal than any other notification in any system state of $R$ (Claim~\ref{thm:general}) and by this proof's assumption that $\nexists c_{\nexists \textnormal{n}} \in R:  (S(c,\textnormal{n}) \cap S(c_{\nexists \textnormal{n}},\textnormal{n})) = \emptyset$ holds, it is true that $\textnormal{\notif}_i[i] = \textnormal{n}$ holds in every system state $c'' \in R$.
\end{proof}

\begin{claim}
\label{thm:invar}
%
The following sequence of invariants is true.
\begin{itemize}[topsep=2pt,itemsep=-.5ex,partopsep=.5ex,parsep=1ex,leftmargin=2mm]
\item []{\bf (1)} Suppose that $\textnormal{\notif}_i[i] = \textnormal{n}$ holds in  every system state $c' \in R$. Eventually the system reaches a state $c'' \in R$, such that for any $p_j \in P$ that is an active participant in $R$, it holds that $\textnormal{\notif}_j[i] = \textnormal{n}$ and $\textnormal{FD}_j[i] = \textnormal{FD}_i$. Moreover, 
%
%
$\textnormal{\notif}_j[j]=\textnormal{n}$ and $\textnormal{FD}_j[j] = \textnormal{FD}_i$ in $c''$ eventually. 

\item[]{\bf (2)} Suppose that invariant (1) holds in  every system state $c' \in R$. Eventually the system reaches a state $c'' \in R$, such that for any $p_i \in P$ that is an active participant in $R$, it holds that $\textnormal{echo}_i[j].\textnormal{\notif} = \textnormal{n}$, $\textnormal{echo}_i[j].\textnormal{part} = \textnormal{FD}_i[i].part$ and $\textnormal{\notif}_i[j]=\textnormal{n}$  in $c''$.

\item[]{\bf (3)} Suppose that invariants (1) and (2) hold in every system state $c \in R$. Eventually the system reaches a state $c'' \in R$, such that for any $p_i \in P$ that is an active participant in $R$, it holds that $\textnormal{all}_i[i] = true$  in $c''$.

\item[]{\bf (4)} Suppose that invariants (1), (2) and (3) hold in every system state $c' \in R$. Eventually the system reaches a state $c'' \in R$, such that for any $p_j \in P$ that is an active participant in $R$, it holds that $\textnormal{all}_j[i] = true$  in $c''$. 

\item[]{\bf (5)} Suppose that invariants (1) to (4) hold in every system state $c' \in R$. Eventually the system reaches a state $c'' \in R$, such that for any $p_j \in P$ that is an active participant in $R$, it holds that $\textnormal{echo}_i[j] = (\textnormal{FD}_i[i].part, \textnormal{\notif}_i[i], \textnormal{all}_i[i])$ in $c''$. 

\item[]{\bf (6)} Suppose that invariants (1) to (5) hold in every system state $c' \in R$. Eventually the system reaches a state $c'' \in R$, such that for any $p_j \in P$ that is an active participant in $R$, it holds that $p_i \in allSeen_j$  in $c''$. 
%



\item[]{\bf (7)} Suppose that invariants (1) to (6) hold in every system state $c' \in R$. Eventually the system reaches a state $c'' \in R$, such that for all $p_i \in P$ that is an active participant in $R$, it holds that the if-statement condition of line~\ref{ln:readyToReplace} holds in $c''$.
\end{itemize}
\end{claim}
\begin{proof} We prove items (1) to (7). 
	
\noindent {\bf (1)} Since $p_i$ repeatedly sends message $m_{i,j}$ to every active processor $p_j \in \textnormal{FD}_i[j]$ (line~\ref{ln:send}), where $m_{i,j} = \langle \bullet, \textnormal{\notif}=\textnormal{n}, \bullet \rangle$, message $m_{i,j}$ arrives eventually to $p_j$ (line~\ref{ln:receive} and the fair communication assumption, Section~\ref{s:sys}). This causes $p_j$ to store $\textnormal{n}$ in $\textnormal{\notif}_j[i]$ (at least) as long as $\textnormal{\notif}_i[j] = \textnormal{n}$, i.e., in every system state $c''$ that follows $m_{i,j}$ arrival to $p_j$. For the case of $\textnormal{n}.phase=1$, we have that $\textnormal{\notif}_j[j]=\textnormal{n}$ (Claim~\ref{thm:notDec} and line~\ref{ln:automatonStep} (case 1) as well as the fact that $\textnormal{n}$ has the greatest lexicographic value in every $c' \in R$, cf. part (2) Claim~\ref{thm:general}).
%
%
The same arguments as above imply that $\textnormal{FD}_j[i] = \textnormal{FD}_i$. Moreover, $\textnormal{FD}_j[j] = \textnormal{FD}_i$ is implied by the assumption that $R$ is admissible with respect to the participant sets.

We show that $\textnormal{\notif}_j[j]=\textnormal{n}$ for any $\textnormal{n}.phase \in \{0,1,2\}$. By the assumptions made in the beginning of the proof of this lemma, we do not need to consider the case of $\textnormal{n}.phase=0$. Thus, we start by considering the case of $\textnormal{n}.phase=1$, we have that $\textnormal{\notif}_j[j]=\textnormal{n}$ (Claim~\ref{thm:notDec} and line~\ref{ln:automatonStep} (case 1) as well as the fact that $\textnormal{n}$ has the greatest lexicographic value in every $c' \in R$, cf. part (2) Claim~\ref{thm:general}).
We now turn to consider the case of $\textnormal{n}.phase=2$. I.e.,  while keeping in mind that $\textnormal{\notif}_j[i]=\textnormal{n}$, we need to show that $\textnormal{\notif}_j[j]=\textnormal{n}$.
We do that by showing that whenever $\textnormal{\notif}_j[j] \neq \textnormal{\notif}_j[i]$, the if-statement condition of line~\ref{ln:stale} holds for any value of $\textnormal{\notif}_j[j].phase\in \{0,1,2\}$ and we focus on the case of $((\textnormal{\notif}_j[j].set \neq \textnormal{\notif}_j[i].set) \land (\textnormal{\notif}_j[j].phase = \textnormal{\notif}_j[i].phase =2))$. Note that by showing that the if-statement condition of line~\ref{ln:stale} does hold in $c''$, we have reached  a contradiction (part (1) of Claim~\ref{thm:general}'s statement). 
%
%
For the case of $\textnormal{\notif}_j[j].phase\in \{0,1\}$, the if-statement condition of line~\ref{ln:stale} holds in $c''$ because $\textnormal{\notif}_j[j] \neq \textnormal{\notif}_j[i]$ implies $\textnormal{all}_j[j]=false$ (line~\ref{ln:adoptAll}) and the condition $|degree_j(j)-degree_j(i)|>1$ (line~\ref{ln:stale} does hold due to the fact that  $degree_j(j)\leq 2$ and $degree_j(i)\geq 4$). Moreover, for the case of $((\textnormal{\notif}_j[j].set \neq \textnormal{\notif}_j[i].set) \land (\textnormal{\notif}_j[j].phase = \textnormal{\notif}_j[i].phase =2))$, the condition $(|\notifSet_i|>1)$  (line~\ref{ln:stale}) does not hold, because $\{ \textnormal{\notif}_j[j].set, \textnormal{\notif}_j[i].set \} \subseteq  \notifSet$ and $|\{ \textnormal{\notif}_j[i].set, \textnormal{\notif}_j[j].set \} |=2$, where $\notifSet=\{ \textnormal{\notif}_j[k].set : \exists p_k \in \textnormal{FD}_j[j]:\textnormal{\notif}_j[i] = \langle 2, \bullet \rangle \}_{p_k \in \textnormal{FD}_j[j]}$.

\noindent {\bf (2)} By similar arguments to the ones that show the arrival of $m_{i,j}$ in part (1) of this proof, processor $p_j$  repeatedly sends the message $m_{j,i} = \langle \bullet, \textnormal{\notif}=\textnormal{n}, \bullet, \textnormal{echo}=(\bullet, \textnormal{n}, \bullet) \rangle$ to $p_i$, which indefinitely stores $\textnormal{n}$ in $\textnormal{\notif}_i[j]$ and $\textnormal{echo}_i[j]$ (line~\ref{ln:receive}). The same arguments as above imply that $\textnormal{echo}_i[j].\textnormal{part} = \textnormal{FD}_i[i].part$.

\noindent {\bf (3)} The system reaches a state in which the invariants of parts (1) and (2) hold for every $p_k \in \textnormal{FD}_i[i]$. In that system state, it holds that $\textnormal{\notif}_i[i]=\textnormal{n}$ (this claim assumption), as well as the fact that  $\textnormal{echo}_i[j].\textnormal{\notif} = \textnormal{n}$ and $\textnormal{\notif}_i[j]=\textnormal{n}$ (part (2) of this proof for showing that $\textnormal{\notif}_j[j]=\textnormal{n}$ eventually and then using part (1) for showing that 
$\textnormal{FD}_i[j] = \textnormal{FD}_j[j]$ and $\textnormal{\notif}_i[j] = \textnormal{n}$). In other words, $(\forall {p_k \in \textnormal{FD}[i].part}:(echoNoAll(k) \land same(k))$
%
%
(cf. line~\ref{ln:addSaw}). 
Let us look at the first step after that state in which processor $p_i$ executes the do forever loop (lines~\ref{ln:doForever} to~\ref{ln:send}). Immediately after that step, the system reaches a state in which $\textnormal{all}_i[i]=true$ holds (line~\ref{ln:adoptAll}).

\noindent {\bf (4)} By similar arguments that appear in parts (1) and (2) of this proof, processor $p_i$ sends repeatedly the message $m_{i,j} = \langle \bullet, \textnormal{\notif}=\textnormal{n}, \textnormal{all}=true, \bullet) \rangle$ to processor $p_j$. Once message $m_{i,j}$ arrives eventually to $p_j$ (line~\ref{ln:receive} and the fair communication assumption, Section~\ref{s:sys}), processor $p_j$ stores $true$ in $\textnormal{all}_j[i]$ (as well as the notification $\textnormal{n}$ in $\textnormal{\notif}_j[i]$). This holds for every system state $c''$ that follows $m_{i,j}$'s arrival to $p_j$. 

\noindent {\bf (5)} By similar arguments that appear in parts (1) and (2) of this proof, processor $p_j$ sends repeatedly the message $m_{i,j} = \langle \bullet, \textnormal{\notif}=\textnormal{n}, \textnormal{all}=true, \textnormal{echo}=(\textnormal{FD}_i[i].part,\textnormal{n},true) \rangle$ to processor $p_i$. Once message $m_{j,i}$ arrives eventually to $p_i$ (line~\ref{ln:receive} and the fair communication assumption, Section~\ref{s:sys}), processor $p_i$ stores $(\textnormal{n},true)$ in $\textnormal{echo}_i[j]$. This holds for every system state $c''$ that follows $m_{j,i}$'s arrival to $p_j$. 

\noindent {\bf (6)} We show that, once the invariants of parts (1) to (5) of this proof hold, the for-each condition of line~\ref{ln:addSaw} holds as well. Moreover, we show that assuming, towards a contradiction, that $p_i \notin \textnormal{allSeen}_j$ implies that there is a step in $R$ that includes a call to  $configSet(\bot)$ (line~\ref{ln:stale}), which is false according to part (1) of Claim~\ref{thm:general}'s statement. 
By showing both the necessity and sufficiency conditions, we can conclude that the assertion of this part in the claim even thought in Algorithm~\ref{alg:disCongif} processors $p_i$ and $p_j$ do not exchange  information about $\textnormal{allSeen}_j$ in order to validate that the set $\textnormal{allSeen}_j$ includes $p_i$. 

The for-each condition of line~\ref{ln:addSaw} condition requires that for any $p_j \in \textnormal{FD}_i[i].part$ to have $(echoNoAll(j) \land same(j))$, 
where the index $k$ is substituted here with $j$. The following is true: $\textnormal{\notif}_i[j] = \textnormal{\notif}_i[i] \land \textnormal{FD}_i[j].part = \textnormal{FD}_i[i].part \land \textnormal{echo}_i[j].part = \textnormal{FD}_i[i].part$  (part (1) of this proof and the assumption that $R$ is admissible with respect to participant sets), $\textnormal{echo}_i[j].\textnormal{\notif} = \textnormal{\notif}_i[i]$ (part (5) of this proof), and $\textnormal{all}_i[j] = true$ (part (1) of this proof and then apply (4) for the case of $\textnormal{\notif}_j[i]=\textnormal{n}$).

We now turn to show that $p_j \notin \textnormal{allSeen}_i$ implies that there is a step in $R$ that includes a call to  $configSet(\bot)$ (line~\ref{ln:stale}).
By the assumptions made in the beginning of the proof of this lemma, we do not need to consider the case of $\textnormal{n}.phase=0$. Note that for $\textnormal{n}.phase=x \in \{ 1, 2 \}$ and $p_i \notin allSeen_j$, when the following condition holds $(\exists p_i \in \textnormal{FD}_j[j] : ((\textnormal{\notif}_j[j].phase =x) \land (\textnormal{\notif}[k].phase =(x+1)(\bmod ~3)) \land (p_i \notin allSeen_j)))$, the condition $(p_i \notin allSeen_j)$ implies that the if-statement condition of line~\ref{ln:stale} holds in $c''$, and this contradicts part (1) of Claim~\ref{thm:general}'s statement. Thus, $p_i \in allSeen_j$ in $c''$. 

\noindent {\bf (7)} Recall that part (5) of this proof says that $\textnormal{echo}_i[j] = (\bullet, \textnormal{n}, true)$. By taking $\textnormal{\notif}_j[j]=\textnormal{n}$  from part (1) and applying it to the results of parts (1) to (6), we get that $p_j \in allSeen_i$ holds eventually, because this is true for every pair of processors $p_i,p_j \in P$ that are active participants in $R$. Let us consider a system state in which invariants (1) to (6) hold for any such pair of active participants $p_i$ and $p_j$.
This implies that $allSeen_i()$ holds (by the fact that part (4) eventually implies that $\forall p_k \in \textnormal{FD}_i[i]: p_k \in \textnormal{allSeen}_i$). Therefore, the if-statement condition of line~\ref{ln:readyToReplace} holds (by this lemma's assumption that $R$ is admissible with respect to the participant sets and part (4) of this proof that shows $\forall p_k \in \textnormal{FD}_j[i]: \textnormal{all}_j[k]=true$).
\end{proof}

\medskip

Claim~\ref{thm:contra} shows a contradiction to the assumption made in the beginning of the proof of Lemma~\ref{thm:convDeg}. 
This contradiction completes the proof of this lemma.

\begin{claim}
\label{thm:contra}
%
$\exists ~c_{\exists \textnormal{n}} \in R$ in which $S(c_{\nexists \textnormal{n}},\textnormal{n})=\emptyset$.
\end{claim}
\begin{proof}
Part (7) of Claim~\ref{thm:invar} says that the if-statement condition of line~\ref{ln:readyToReplace} holds in $c'' \in R$ eventually. Let $a_i$ a step that immediately follows $c''$ and in which $p_i$ executes the do forever loop (line~\ref{ln:doForever} to~\ref{ln:send}).
Note that the case of $\textnormal{n}.phase\in \{ 0\}$ contradicts this lemma's assumption that $c_{\nexists \textnormal{n}} \in R$ in which $S(c_{\nexists \textnormal{n}},\textnormal{n})=\emptyset$.
Suppose that $\textnormal{n}.phase\in \{ 1 \}$ in $c''$. The step $a_i$ includes an execution of line~\ref{ln:readyToReplace} that increases the lexical value of $\textnormal{n}$ (contradicting part (2) of Claim~\ref{thm:general}).
Suppose that $\textnormal{n}.phase\in \{ 2 \}$ in $c''$. The step $a_i$ includes an execution of line~\ref{ln:readyToReplace} that changes $\textnormal{n}$'s value to $\langle 0, \bot \rangle$. Claim~\ref{thm:invar} implies that this holds for any $p_i \in P$ that is an active participant in $R$. This is a contradiction with the assumption made in the beginning of the proof of this lemma, i.e.,  
$\exists ~c_{\nexists \textnormal{n}} \in R$ in which $S(c_{\nexists \textnormal{n}},\textnormal{n})=\emptyset$.
\end{proof}

\medskip

\noindent This completes the proof of Lemma~\ref{thm:convDeg}.\hfill\end{proof}


\subsubsection*{Algorithm~\ref{alg:disCongif} is self-stabilizing}
Let $a_i \in R$ be a step in which processor $p_i$ calls the function $\configEstab(\textnormal{set})$ (line~\ref{ln:configEstab}), and in which the if-statement condition $(\noReconfig() \land (set \notin \{ \textnormal{config}[i], \emptyset \} ))$ does hold in the system state that immediately precedes $a_i$. We say that $a_i$ is an \emph{effective  (configuration establishment)} step in $R$. Similarly, we consider $a_i \in R$ to be a step in which processor $p_i$ calls the function $participate()$ (line~\ref{ln:participate}), and in which the if-statement condition $\noReconfig()$ does hold in the system state that immediately precedes $a_i$.
Let $R = R' \circ R_{\textnormal{VNER}} \circ R'''$ be an execution that does include explicit (delicate) replacements, where $R'$ and $R'''$ are a prefix, and respectively, a suffix of $R$. Let us consider $R_{\textnormal{VNER}}$, which is a part of execution $R$. We say that $R_{\textnormal{VNER}}$ \emph{virtually does not include explicit (delicate) replacements} (VNER) when for any step $a \in R_{\textnormal{VNER}}$ that includes a call the function $\configEstab(\textnormal{set})$ (line~\ref{ln:configEstab}) or $participate()$ (line~\ref{ln:participate}) is ineffective. Given a system state $c \in R$, we say that $c$ includes no notification if none of its active processors stores a notification and there are no notifications in transit between any two active processors.

\begin{lemma}[Eventually there is an VNER part]
\label{thm:virtuallyNotExplicit}
Let $R$ be an execution of Algorithm~\ref{alg:disCongif} (which may include explicit delicate replacements). Eventually the system reaches a state (1) $c \in R$ after which an VNER part, $R_{\textnormal{VNER}}$, starts. Moreover, after that, the system reaches eventually a state (2) $c_{\textnormal{goodNtf}} \in R_{\textnormal{VNER}} :\textnormal{NA}(c_{\textnormal{goodNtf}})=\emptyset$ that has either (2.1) no notification or (2.2) at most one notification in the system and this notification becomes the system quorum eventually (after reaching a state in which invariants (1) to (7) of Claim~\ref{thm:invar} hold in $c_{\textnormal{goodNtf}}$). 
\end{lemma}
\begin{proof}
Note that the proof of part (1) is done whenever $R$ has a suffix $R'$ in which no processor takes a step that includes a call to the function $\configEstab(\textnormal{set})$ (line~\ref{ln:configEstab}). Let us suppose that suffix $R'$ exists and show that part (2.1) is true. The reason is that $R'$ satisfies the conditions of Lemma~\ref{thm:convDeg}, which implies that eventually $R'$ reaches $c_{\textnormal{goodNtf}}$. This is because Lemma~\ref{thm:convDeg} says that the lexical largest notification is removed eventually from the system. Once that happen, either the system reaches $c_{\textnormal{goodNtf}}$, or we can look at the suffix $R''$ that starts after that removal (of that lexical largest notification) and then we apply Lemma~\ref{thm:convDeg} again. Since $R''$ does not include steps with a call to the function $\configEstab(\textnormal{set})$, no new notifications are ``added'' to the system, and we can continue to apply Lemma~\ref{thm:convDeg} to the suffix of $R''$ until the system includes no notifications, i.e., it reaches $c_{\textnormal{goodNtf}}$.

For the complementary case, we show that $R$ includes an $R_{\textnormal{VNER}}$ part, i.e., part (1) of this proof. We consider the different cases of the phase of the notification with the maximal lexical value. We show that notification must reach the phase value of $2$ and then leaves the system eventually. For this case, we show that parts (1) and (2.1) of this lemma hold.    

%
Let us consider the case in which $R$ includes an unbounded number of steps in which active processors call the function $\configEstab(\textnormal{set})$ (line~\ref{ln:configEstab}).
%
%
(Note that this is the  complementary case to the first part of this proof because here $R'$ does not exists.) Let $S = \{ set :  a_i \in R \textnormal{ includes a step that calls }  \configEstab(\textnormal{set})\}$.
(Note that the case of $S = \emptyset$ is just another way to say that there are no effective steps in $R$ and thus $R'$ exists.)
Suppose that $S \neq \emptyset$ and let $\textnormal{n} = \langle \bullet, set \rangle$ be a notification with the largest lexical value in $R$, where $\textnormal{set} \in S$. (Note that the proof needs not consider the case of $\textnormal{n} = \langle 0, set \rangle$, see line~\ref{ln:configEstab}.)

In order to complete the proof of parts (1) and (2) for the complementary case, we show that, for any choice of phase, it cannot be that $\textnormal{n}=\langle phase, set \rangle$ and execution $R$ does not have an VNER part, $R_{\textnormal{VNER}}$. We do that by showing that the system reaches eventually a state $c$ in which $noReco()$ (line~\ref{ln:noReco}) returns false for any active processor, $p_i \in P$. Moreover, $c$ marks the beginning of $R_{\textnormal{VNER}}$ for which we show in part (2) of the proof that $c_{\textnormal{goodNtf}} \in R_{\textnormal{VNER}}$ follows.

\noindent \textbf{Suppose that $\textnormal{n} = \langle 1, set \rangle$.~~} This part of the proof considers two cases. One in which $\textnormal{n}$ is a proposal that never leaves the system and the another case in which it does leave eventually. 

\noindent $\bullet$ \textit{Suppose that $\textnormal{n}$ never leaves the system.}
Suppose, towards a contradiction, that (i) $\textnormal{n}$ never leaves the system and yet (ii) the system never reaches $R_{\textnormal{VNER}}$. By the same arguments that appear in the proof of part (1) in Claim~\ref{thm:invar} we get that invariant   
$\textnormal{\notif}_i[j] = \textnormal{n}$ holds, where $p_i, p_j \in P$ are any two active processors. Thus, $noReco()$ (line~\ref{ln:noReco}) returns false for any active processor, $p_i \in P$, and the system reaches $R_{\textnormal{VNER}}$'s first configuration, $c$. Thus, a contradiction because the assumption that (i) and (ii) can hold together is false.

\noindent$\bullet$ \textit{Suppose that $\textnormal{n}$ does leave the system.}
We note that the case in which $\textnormal{n}$ leaves the system eventually contradicts our assumption that $\textnormal{n}$ has the highest lexical values in $S$. The reason is that only active processors can make a notification leave the system and they do so by changing their values in lines~\ref{ln:readyToReplace} to~\ref{ln:automatonStep} (Claim~\ref{thm:notDec}). But then, for the case of $\textnormal{n} = \langle 1, set \rangle$, the lexical value of $\textnormal{n}$ increases when Algorithm~\ref{alg:disCongif} 
takes such a step.

We note that there is no need to show part (2) of this lemma for the above two cases, because none of them is possible.

\noindent \textbf{Suppose that $\textnormal{n} = \langle 2, set \rangle$.~~}
Similar to the previous case, we consider both the case in which $\textnormal{n}$ is a proposal that never leaves the system and the other case in which it does leave eventually. 

\noindent$\bullet$ \textit{Suppose that $\textnormal{n}$ never leaves the system.}
As in the proof of the first case of $\textnormal{n} = \langle 1, set \rangle$  and in which $\textnormal{n}$ never leaves the system, we can use the same arguments as in the proof of part (1) of Claim~\ref{thm:invar} for showing that invariant   
$\textnormal{\notif}_i[j] = \textnormal{n}$ holds, where $p_i, p_j \in P$ are any two active processors. Thus, $noReco()$ (line~\ref{ln:noReco}) returns false for any active processor, $p_i \in P$, and the system reaches $R_{\textnormal{VNER}}$'s first configuration, $c$. 
Note that once the system enters the VNER part of the execution, $R_{\textnormal{VNER}}$, we can use similar arguments to the ones used in the proof of Lemma~\ref{thm:convDeg} to show that the system removes $\textnormal{n}$. Thus, a contradiction with our assumption (of this case) that $\textnormal{n}$ never leaves the system. Moreover, there is no need to prove part (2) of this lemma (for this impossible case). 

\noindent$\bullet$ \textit{Suppose that $\textnormal{n}$ does leave the system.}
We show that the only possible case is that $\textnormal{n} = \langle 2, set \rangle$ does leave the system. Let us  consider the set $S'=\textnormal{NA}(c_{\textnormal{start}})$ of notifications that were present in $R$'s starting state, $c_{\textnormal{start}}$. We note that $|S'| \in {\cal O}(cap \cdot n^2)$, where $cap$ is the bound on the link capacity (Section~\ref{sec:settings}). Thus, even though the number of notifications in $S$ is unbounded, the number of the notifications in $S'$ is bounded. Since this case assumes that $\textnormal{n}$ eventually leaves the system, and all other cases above are impossible, it must be the case that $R$ has a suffix $R''$ that starts in a system state $c'' :\textnormal{NA}(c'') \cap S' \neq \emptyset$ the encodes no notification that appeared in $R$'s starting state. This can be shown by induction on $|S'|$. Moreover, the proof can consider the case in which a given notification $\textnormal{n}=\langle \bullet, \textnormal{set} \rangle: \textnormal{set} \in S'$ leaves the system and then returns via a call to $estab()$ (line~\ref{ln:configEstab}). This is done by letting the proof to decorate the notification sets in $R$'s starting system state, say, using the color `red'. Moreover, each set in a notification that is the result of a call to $estab(\textnormal{set})$ (line~\ref{ln:configEstab}) also has  decoration, but this time with another color, say `blue'.

Let $\textnormal{n}',\textnormal{n}'' \in S \setminus S'$ be two notifications (with decorated sets) that become present in $R$ but they are not present in $R$'s starting state. The only way in which  $\textnormal{n}'=\langle 1, \bullet \rangle$ can become present in the system is via a step $a_1 \in R$ that calls $\configEstab(\textnormal{set})$ (line~\ref{ln:configEstab}). Moreover, the only way that $\textnormal{n}''=\langle 2, \bullet \rangle$ can become present in $R$ is via step $a_2 \in R$ that changes $\textnormal{n}'$ to $\textnormal{n}'' = \langle 2, set \rangle$ (lines~\ref{ln:readyToReplace} to~\ref{ln:automatonStep}), such that $\exists \textnormal{n}' \in S \setminus S' : \textnormal{n}'.set=\textnormal{n}''.set$. 
Note that after step $a_1$, other steps $a_0 \in R$ occur in which $\textnormal{n}''$ disappears from the system (lines~\ref{ln:readyToReplace} to~\ref{ln:automatonStep}). This sequence of steps between $a_2$ and $a_0$ and the system state that immediately precedes them are depicted by invariants (1) to (7) of Claim~\ref{thm:invar} since these invariants hold in these states. Thus, $c \in R$ holds immediately before $a_2$ and $R$ must include an $R_{\textnormal{VNER}}$ part, because invariant (1) holds eventually in $R$. Moreover, invariants (1) to (7) of Claim~\ref{thm:invar} imply that $\textnormal{n}''$ is the only notification in the system, which eventually becomes the quorum configuration after all steps $a_0$ are taken. Thus, part (2) of this lemma is shown.  
\end{proof}

Theorem~\ref{thm:staleFreeExecution} demonstrates the eventual absence of stale information, which implies that Algorithm~\ref{alg:disCongif} convergences eventually, i.e., it is practically-stabilizing. 

\begin{theorem}[Convergence] 
\label{thm:staleFreeExecution}
Let $R$ be an admissible execution of Algorithm~\ref{alg:disCongif}.
Eventually the system reaches a state $c \in R$ in which the invariants of no type-1, type-2, type-3 and type-4 stale information hold thereafter.  
\end{theorem}
\begin{proof}
Claim~\ref{thm:noStale} shows that there is no type-1 stale information  eventually. 
Lemma~\ref{thm:noConflict} shows that there is no type-2 stale information  eventually. 
Lemmas~\ref{thm:convDeg} and~\ref{thm:virtuallyNotExplicit} say that the system either reaches a state in which there are no notifications or there is at most one notification (for which invariants (1) to (7) of Claim~\ref{thm:invar} hold) that later  becomes the system configuration. Note that both cases imply that there is no type-3 stale information  eventually. 
Claims~\ref{thm:noStale4} shows that there is no type-4 stale information  eventually. 
\end{proof}

\begin{theorem}[Closure]
%
%
\label{thm:closureThm}
%
%
Let $R$ be an execution of Algorithm~\ref{alg:disCongif}.
Suppose that execution $R$ starts from a system state that includes no stale information. (1) for any system state $c \in R$, it holds that 
$c$ includes no stale information. Suppose that the step that immediately follows $c$ includes a call to $\configEstab()$. (2) The only way that $\textnormal{set}$ becomes a notification is via a call to $\configEstab(\textnormal{set})$ (line~\ref{ln:configEstab}) and the only way that a processor becomes a participant is via a call to $participate()$ (line~\ref{ln:participate}). (3) The presence of notifications in $R$ implies that one of them replaces the system configuration eventually. 
%
\end{theorem}

\begin{proof} The proof essentially follows from established results above.
	
\noindent \textbf{Part (1).~}~%
Since there is no stale information in the system state that immediately proceeds $c$, there is no stale information in $c$. 
This follows from a close investigation of the 
lines that can change the system state in a way that might introduce stale information; the most  
relevant lines are the ones that deal with notifications (lines~\ref{ln:adoptAll} and~\ref{ln:automatonStep}) and new participants (line~\ref{ln:join}). 
Thus, the proof is completed via parts (2) and (3). 

\noindent \textbf{Part (2).~}~%
This is immediate from lines~\ref{ln:configEstab} and~\ref{ln:participate}.

\noindent \textbf{Part (3).~}~%
%
It is not difficult to see that $R$ includes an $R_{\textnormal{VNER}}$ part (Lemma~\ref{thm:virtuallyNotExplicit}, Part (1)). 
Then, the proof completes by applying part (7) of Claim~\ref{thm:invar} twice: Once for showing the selection of a single notification that has phase $2$, 
and the second time for showing that the selected notification replaces the 
configuration.~\end{proof}


\subsection{Reconfiguration Management}
\label{sec:reconMan}
The Reconfiguration Management $recMA$ layer shown in Algorithm~\ref{alg:SSQR} bears the weight of initiating or \emph{triggering} a reconfiguration when (i) the configuration majority has been lost, or (ii) when the prediction function $evalConf()$ indicates to a majority of members that a reconfiguration is needed to preserve the configuration.
To trigger a reconfiguration, Algorithm~\ref{alg:SSQR} uses the $\configEstab(set)$ interface with the $recSA$ layer.
In this perspective, the two algorithms display their modularity as to their workings.
Namely, $recMA$ controls \emph{when} a reconfiguration should take place, but the reconfiguration replacement process is left to $recSA$, which will install a new configuration also trying to satisfy $recMA$'s proposal of the new configuration's set.
Several processors may trigger reconfiguration simultaneously, but by the correctness of Algorithm~\ref{alg:disCongif} this does not affect the delicate reconfiguration, and by the correctness of Algorithm~\ref{alg:SSQR}, each processor can only trigger once when this is needed.

In spite of using majorities, the algorithm is generalizable to other (more complex) quorum systems, while the prediction function $evalConf()$ can be either very simple, e.g., asking for reconfiguration once $1/4^{th}$ of the members are not trusted, or more complex, based on application criteria or network considerations.
More elaborate methods may also be used to define the set of processors that Algorithm~\ref{alg:SSQR} proposes as the new configuration. Our current implementation, aiming at simplicity of presentation, defines the set of trusted participants of the proposer as the proposed set for the new configuration.

\begin{algorithm*}[t]

   \caption{Self-stabilizing Reconfiguration Management; code for processor $p_i$}
\label{alg:SSQR}
\begin{footnotesize}


\noindent {\bf Interfaces:}
$evalConf()$ returns $\sf True/False$ on whether a reconfiguration is required or not by using some (possibly application-based) prediction function.
The rest of the interfaces are specified in Algorithm~\ref{alg:disCongif}.
$\noReconfig()$ returns $\sf True$ if a reconfiguration is not taking place, or $\sf False$ otherwise.
$\configEstab(set)$ initiates the creation of a new configuration based on the $set$.
$getConfig()$ returns the current local configuration.
\label{SSQR:interfaces}

\noindent {\bf Variables:} \label{SSQR:var}
$needReconf[\,]$ is an array of size at most $N$, composed of booleans $\{\sf True, False \}$, where $needReconf_i[j]$  holds the last value of $needReconf_j[j]$ that $p_i$ received from $p_j$ as a result of exchange (lines~\ref{SSQR:send} and \ref{SSQR:receive}) and $needReconf$ is an alias to $needReconf_i[i]$ i.e., of $p_i$'s last reading of $evalConf()$.
Similarly, $noMaj_i[\,]$ is an array of booleans of size at most $N$ on whether some trusted processor of $p_i$ detects a majority of members that are active per the reading of line~\ref{SSQR:testMaj}.
$noMaj_i[j]$ (for $i \neq j$) holds the last value of $noMaj_j[j]$ that $p_i$ received from $p_j$. 
Finally $prevConfig$ holds $p_i$'s believed previous $config$.
\\

\noindent {\bf Macros:} \\
$core() =$ $\bigcap_{p_j \in FD_i[i].part}FD[j].part$\;\label{SSQR:defCore}
$flushFlags():$ \lForEach{$p_j\in FD[i]$}{$needReconf[j]  \gets (noMaj[j] \gets {\sf False})$}
\noindent {\bf Do forever} \Begin{
\If{$p_i \in FD[i].part$\label{SSQR:isPartpnt}}{
$curConf = getConfig()$\; \label{SSQR:readConfig}

$needReconf[i]  \gets (noMaj[i] \gets {\sf False})$\;\label{SSQR:ownFlagReset}
\lIf{$prevConfig$ $\not \in$ $\{curConf, \bot\}$}{\label{SSQR:configChanged}
$flushFlags()$}\label{SSQR:flagsResetAll}

\If{$\noReconfig() = {\sf True}$}{ \label{SSQR:configChanging}

$prevConfig$ $\gets$ $curConf$\; \label{SSQR:updatePrev}

\lIf{$|\{p_j \in  curConf \cap FD[i]\}| < (\frac{|curConf|}{2}+1)$}{$noMaj[i] \gets {\sf True}$} \label{SSQR:testMaj}
\uIf{$(noMaj[i]$ $=\sf True)$ $\land$ $(|core()|>1)$ $\land$ $({\forall p_k \in core(): noMaj[k] = {\sf True}})$ \label{SSQR:checkNoMajCore}}{
$\configEstab(FD[i].part)$\; \label{SSQR:noMajTrigger}
$flushFlags()$\;\label{SSQR:noMajResetAll}
}

\ElseIf{$(needReconf[i] \gets evalConf(curConf))$ $\land$ $|\{p_j \in curConf \cap FD[i]: needReconf[j] = {\sf True}\}| >\frac{|curConf|}{2}$\label{SSQR:gracefulQreconf}}{
$\configEstab(FD[i].part)$\; \label{SSQR:needReconfTrigger}
$flushFlags()$\;\label{SSQR:needReconfResetAll}
}  

}

\lForEach{$p_j$ $\in$ $FD[i].part$}{$send(\langle noMaj[i], needReconf[i] \rangle)$} \label{SSQR:send}
}
}

\noindent {\bf Upon receive} $m$ 
{\bf from} $p_j$ \textbf{do} \label{SSQR:receive} 
\lIf{$p_i \in FD[i].part$}{$\langle noMaj[j], needReconf[j] \rangle \gets m$}\label{SSQR:receiveStore}


\end{footnotesize}
\end{algorithm*}
\setlength{\textfloatsep}{5pt}

\subsubsection{Algorithm Description}
\paragraph{Preserving a majority.}
The algorithm strives to ensure that a majority of the configuration is active. 
Although majority is a special case of a quorum, the solution in extensible to host other quorum systems that can be built on top of the $config$ set, in which case, the algorithm aims at keeping a robust quorum system where robustness criteria are subject to the system's dynamics and application requirements. 
In this vein, the presented algorithm employs a configuration evaluation function $evalConf()$ used as a black box, which predicts the quality of the current $config$ and advises any participant whether a reconfiguration of $config$ needs to take place.
Given that local information is possibly inaccurate, we prevent unilateral reconfiguration requests --that may be the result of inaccurate failure detection-- by demanding that a processor must first be informed of a majority of processors in the current $config$ that also require a reconfiguration (lines~\ref{SSQR:gracefulQreconf}--\ref{SSQR:needReconfResetAll}). 

\paragraph{Majority failure.}
On the other hand, we ensure liveness by handling the case were either the prediction function does not manage to sustain a majority, or an initial arbitrary state lacks a majority but there are no $config$ inconsistencies that can trigger a delicate reconfiguration (via the $\configEstab()$ interface).
Lines~\ref{SSQR:checkNoMajCore}--\ref{SSQR:noMajResetAll} tackle this case by defining the \emph{core} of a processor $p_i$ to be the intersection of the failure detector readings that $p_i$ has for the processors in its own failure detector, i.e., $\cap_{p_j \in FD_i[i]}FD_i[j]$.
If this local core agrees that there is no majority, i.e. that $noMaj= {\sf True}$, then $p_i$ can request a new $config$.
As a liveness condition to avoid triggering a new $config$ due to FD inconsistencies when there actually exists a majority of active configuration members, we place the \emph{majority-supportive core} assumption on the failure detectors, as seen in Definition~\ref{def:majSupCore} below.
Simply put, the assumption requires that if a majority of the current configuration is active, then the core of every processor $p_t$ that is a participant, contains at least one processor $p_s$ with a failure detector supporting that a majority of $config$ is trusted.
Furthermore, $p_t$ has knowledge that $p_s$ can detect a majority of trusted members.

\paragraph{Detailed description.}
The algorithm is essentially executed only by participants as the condition of line~\ref{SSQR:isPartpnt} suggests.
Line~\ref{SSQR:readConfig} reads the current configuration, while line~\ref{SSQR:ownFlagReset} initiates the local $noMaj_i[i]$ and $needReconf_i[i]$ variables to $\sf False$. 
If a change from the previous configuration has taken place, the arrays $noMaj[\,]$ and $needReconf[\,]$ are reset to $\sf False$ (line~\ref{SSQR:flagsResetAll}).
The algorithm proceeds to evaluate whether a reconfiguration is required by first checking whether a reconfiguration is already taking place (line~\ref{SSQR:configChanging}) through the $\noReconfig()$ interface of $recSA$.
If this is not the case, then it checks whether it can see a trusted majority of configuration members, and updates the local $noMaj_i[i]$ boolean accordingly (line~\ref{SSQR:testMaj}).
If $noMaj_i[i] = \sf True$, i.e., no majority of members is active, and line~\ref{SSQR:checkNoMajCore} finds that all the processors in its core also fail to find a majority of members, then $p_i$ can trigger reconfiguration using $\configEstab()$ with the current local set of participants as the proposed new $config$ set (lines~\ref{SSQR:noMajTrigger}--\ref{SSQR:noMajResetAll}).
The $needReconf_i[\,]$ and $noMaj_i[\,]$ arrays are again reset to $\sf False$ to prevent other processors that will receive these to trigger.
Line~\ref{SSQR:gracefulQreconf} checks whether the prediction function $evalConfig()$ suggests a reconfiguration, and if a majority of members appears to agree on this, then the triggering proceeds as above.
Participants continuously exchange their $noMaj$ and $needReconf$ variables (lines~\ref{SSQR:send}--\ref{SSQR:receiveStore}).


\subsubsection{Correctness}
The Reconfiguration Management algorithm is responsible for triggering a reconfiguration when either a majority of the members crash or whenever the (application-based) $config$ evaluation mechanism $evalConfig()$ suggests to a members' majority that a reconfiguration is required. 
The correctness proof ensures that, given the assumption of majority-supportive core holds, Algorithm~\ref{alg:SSQR} can converge from a transient initial state to a safe state, namely, that after $recMA$ has triggered a reconfiguration, it will never trigger a new one before the previous one is completed and only if a new event makes it necessary.

\paragraph{Terminology.} We use the term \emph{steady $config$ state} to indicate a system state in an execution were a $config$ has been installed by Algorithm~\ref{alg:disCongif} at least once, and 
the system state is conflict-free.
A \emph{legal execution} $R$ for Algorithm~\ref{alg:SSQR}, refers to an execution that converges to a steady $config$ state.
Moreover, a reconfiguration in $R$ takes place only when a majority of the configuration members fails, or when a majority of the members require a reconfiguration.
The system remains conflict-free and moves to a new steady $config$ state with a new configuration.

\begin{definition} [Majority-supportive core]
\label{def:majSupCore}
Consider a steady $config$ state in an execution $R$ where the majority of members of the established $config$ never crashes.
The \emph{majority-supportive core} assumption requires that every participant $p_i$ with a local core  $\cap_{p_j \in FD_i[i]}FD_i[j]$ containing more than one processor, must have a core with at least one active participant $p_r$ whose failure detector trusts a majority of the $config$, and for such a processor $noMaj_i[r] = \sf False$ throughout $R$. 
\end{definition}

\begin{remark}
\label{rem:reconfThroughInterface}
We say that Algorithm~\ref{alg:disCongif} is \emph{triggered} when a reconfiguration is initialized. 
By Algorithm~\ref{alg:SSQR}, the only way that Algorithm~\ref{alg:SSQR} can cause a triggering of Algorithm~\ref{alg:disCongif} is through a call to the $\configEstab()$ interface with Algorithm~\ref{alg:disCongif} on lines~\ref{SSQR:noMajTrigger} and~\ref{SSQR:needReconfTrigger}.
\end{remark}

\begin{lemma}
\label{thQ:reachSteady}
Starting from an arbitrary initial state in an execution $R$, where stale information exists, Algorithm~\ref{alg:SSQR} converges to a steady $config$ state where local stale information are removed.
\end{lemma}

\begin{proof}
Consider a processor $p_i$ with an arbitrary initial local state where stale information exists  (1) in the program counter, (2) in $noMaj_i[\bullet]$ and $needReconf_i[\bullet]$ and (3) in $prevConf$. \\
\textbf{Case 1} -- Stale information may initiate the algorithm in a line other than the first of the pseudocode.
If $p_i$'s program counter starts after line~\ref{SSQR:configChanging} and if a reconfiguration is taking place, then Algorithm~\ref{alg:SSQR} may force a second reconfiguration while Algorithm~\ref{alg:disCongif} is already reconfiguring.
The counter for example could start on lines~\ref{SSQR:noMajTrigger} and~\ref{SSQR:needReconfTrigger}.
This would force a brute reconfiguration.
This triggering cannot be prevented in such a transient state, but we note that any subsequent iteration of the algorithm is prevented from accessing $\configEstab()$ lines (as in  Remark~\ref{rem:reconfThroughInterface}) before the reconfiguration is finished.

%

\noindent\textbf{Case 2} -- We note that after a triggering as the one described above, the fields of arrays $noMaj_i[\bullet]$ and $needReconf_i[\bullet]$ are set to $\sf False$. 
Moreover, in every iteration $noMaj_i[i]$ and $needReconf_i[i]$ are set to $\sf False$.
During a reconfiguration these values are propagated to other processors and $p_i$ receives their corresponding values.
Therefore, $p_i$ must eventually receive $noMaj_j[j]$ ($needReconf_j[j]$) from some participant $p_j$, and overwrite any transient values.
Lemma~\ref{thQ:noAbruptConfigTriggered} bounds the number of reconfigurations that may be triggered by corruption that is not local, i.e., that emerges from corrupt $noMaj$ ($needReconf$) values that arrive from the communication links.

\noindent\textbf{Case 3} -- We anticipate that any reconfiguration returns a different configuration than the previous one.
In a transient state though, the previous configuration ($prevConfig$) and the current configuration $curConf$ may coincide.
This ignores the check of line~\ref{SSQR:flagsResetAll} that sets $noMaj_i[\bullet]$ and $needReconf_i[\bullet]$ to $\sf False$ upon the detection of a reconfiguration change. 
This forms a source of a potential unneeded reconfiguration.
Nevertheless, $prepConfig$ receives the most recent configuration value on every iteration of line~\ref{SSQR:updatePrev} and thus the above may only take place once throughout $R$ per processor.

\noindent Eliminating these sources of corruption, we reach a steady $config$ state without local stale information.
\end{proof}
\vspace{.5em}

\begin{lemma}
\label{thQ:noAbruptConfigTriggered}
Consider a steady $config$ state $c$ in an execution $R$ where the majority-supportive core assumption holds throughout, the majority of $config$ processors never crash and there is never a majority of members supporting $evalConf() = \sf True$. 
There is a bounded number of triggerings of the Algorithm~\ref{alg:disCongif} that are a result of stale information, namely  $O(N^2 cap)$.
\end{lemma}

\begin{proof}
By Remark~\ref{rem:reconfThroughInterface}, the only way that the algorithm may interrupt a steady $config$ state, is by reaching lines~\ref{SSQR:noMajTrigger} and~\ref{SSQR:needReconfTrigger} that have a call to $\configEstab()$.
We assume that some member $p_t \in config$ has triggered Algorithm~\ref{alg:disCongif} at some system state $c_t \in R$, and we examine whether and when this state is attainable from $c$.
We note that \emph{in a complete iteration} of Algorithm~\ref{alg:SSQR}, $p_t$ must have no reconfiguration taking place while triggering, since this is a condition to reach the above mentioned lines imposed by line~\ref{SSQR:configChanging}.
We first prove that if there is a triggering it must be due to initial corrupt information and then argue that this can take place a bounded number of times.

\noindent \textbf{Case 1 -- The reconfiguration was initiated by line~\ref{SSQR:noMajTrigger}}. 
\sloppy{This implies that the condition of line~\ref{SSQR:checkNoMajCore} is satisfied, i.e.,
at some system state $c_t \in R$, $p_t$ has local information that satisfies $(noMaj_t[t]$ $=\sf True)$ $\land$ $(|core_t()|>1)$ $\land$ $({\forall p_k \in core_t(): noMaj_t[k] = {\sf True}})$.}
Condition $(noMaj_t[t]$ $=\sf True)$ may be true locally for $p_t$, only due to failure detector inaccuracy, because, by the \assert, the majority of processors in the $config$ never fails throughout $R$. 
Condition $|core_t()|>1$ suggests that $p_t$ has at least two participant processors in its core (without requiring $p_t \in core_t()$). 
By the majority-supportive core assumption and the above, we are guaranteed that $\exists p_s \in core_t(): |FD_s[s] \cap config_s[s]| > \frac{|config_s[s]|}{2}$ throughout $R$ and $noMaj_t[s] = {\sf True} \iff noMaj_s[s] = {\sf True}$.
But in this case, $noMaj_s[s] = \sf False$ and $noMaj_t[s] = \sf True$ which contradicts the majority supportive assumption. 
We thus reach to the result.
 
Note that $noMaj_t[s] = \sf True$ can reside in $p_t$'s local state or in the communication links that may carry stale information. 
Because of the boundedness of our system, we can have one instance of corrupt $noMaj_t[s] = \sf True$ in $p_t$'s local state, and $cap$ instances in the communication link.
I.e., such information may cause a maximum of $1+cap\cdot N$ triggerings per processor.
Any processor that enters the system cannot introduce corrupt information to the system due to the data-links protocols and the joining mechanism.
Thus majority supportive assumption is also attainable even when starting from arbitrary states.

\noindent \textbf{Case 2 -- The reconfiguration procedure was triggered by line~\ref{SSQR:gracefulQreconf}}.
This implies that for $p_t$, both conditions were true, i.e., (a) $(needReconf_t \gets evalConf_t(config_t))$ and (b) $|\{p_j \in config_t\cap FD_t: needReconf_t[j] = {\sf True}\}| >\frac{|config_t|}{2}$.
We note that the $needReconf_t[t]$ variable is always set to $\sf False$ upon the beginning of every iteration.
Thus the local function $evalConf_t()$ due to $p_t$'s failure detector and other application criteria explicitly suggested a reconfiguration \emph{in the specific iteration} in which $p_t$ triggered the reconfiguration.
From the \assert, there is no majority of processors in the $config$ that supports a reconfiguration, even at the time when $p_t$ triggered the reconfiguration.

Thus $needReconf_t[s] = \sf True$ must reside in $p_t$'s local state and in the communication links. 
We can have one instance of corrupt $needReconf_t[s] = \sf True$ in $p_t$'s local state, and $cap$ instances in the communication links.
I.e. such information may cause a maximum of $1+cap\cdot N$ triggerings per processor.
Note that after every such triggering, the source of triggering is eliminated by reseting $needReconfig[\,]$ to $\sf False$ (lines~\ref{SSQR:flagsResetAll}, \ref{SSQR:noMajResetAll} and~\ref{SSQR:needReconfResetAll}).
From this point onwards any processor that enters the system cannot by the data-links and the joining mechanism introduce corrupt information to the system.

So the possible triggerings in the system attributed to stale information are confined to $O(N^2cap)$ and by Algorithm~\ref{alg:disCongif} guarantees we always reach a steady $config$ state.
\end{proof}

\vspace{.5em}
Let $c_{safe}$ denote a \emph{safe system state} where all possible sources of triggerings  attributed to the arbitrary initial state have been eliminated.
We denote an execution starting from $c_{safe}$ as $R_{safe}$.

\begin{lemma}
\label{thQ:steadyRemainsSteady}
Consider an execution $R_{safe}$ where the majority-supportive core assumption holds throughout, the majority of $config$ processors never crash and there is never a majority of the $config$ with local $evalConf() = \sf True$. 
This execution is composed of only steady $config$ states. 
\end{lemma}

\begin{proof}
By Lemma~\ref{thQ:noAbruptConfigTriggered} there is a bounded number of triggerings due to initial arbitrary information.
Given that we have reached a safe system state, these triggerings do not occur.
The last $config$ change, has by line~\ref{SSQR:flagsResetAll} reset all the fields in $noMaj[\,]$ and $needReconf[\,]$ to $\sf False$ and this holds for all participants (even if they are not members of the $config$).
By our assumption a majority of processors does not crash. 
The majority-supportive core assumption states that throughout $R_{safe}$ there exists at least one processor $p_i$ in the core of $p_t$ that always has $noMaj_i[i] = {\sf False}$ and $p_t$ has $noMaj_t[i] = {\sf False}$ .
Thus the condition of line~\ref{SSQR:checkNoMajCore} can never be true, and thus there is no iteration of the algorithm that can reach line~\ref{SSQR:noMajTrigger}.
Similarly, since no majority of processors in the $config$ change to $needReconf = {\sf True}$ in this execution, and the local states are exchanged continuously over the token-based data-link, line~\ref{SSQR:needReconfTrigger} cannot be true.
Thus any system state in $R$ is a steady $config$ state.
\end{proof}

\begin{lemma}
\label{thQ:triggeredWhenNeeded}
Starting from an $R_{safe}$ execution, Algorithm~\ref{alg:SSQR} guarantees that if (1) a majority of $config$ members collapse or if (2) a majority of members require a reconfiguration as per the prediction function, a reconfiguration takes place.
\end{lemma}

\begin{proof} We consider the two cases separately.

\noindent \textbf{Case 1 --} 
If a majority of the members collapses, then based on the failure detector's correctness, a non-crashed participant $p_t$ will eventually stop including a majority of $config$ members in its failure detector and participants ($FD.part$) set.
We remind that rejoins are not permitted.
Since the majority-supporting core assumption does not apply in this case, any processor in $p_t$'s core must eventually reach to the same conclusion as $p_t$.
Every such participating processor $p_s \in core_t()$ propagates $noMaj_s[s] = \sf True$ in every iteration. 
By the assumption that a packet sent infinitely often arrives infinitely often (the fair communication assumption, Section~\ref{s:sys}), any processor such as $p_t$ must eventually collect a $noMaj = \sf True$ from every member like $p_s$ core and thus enable a reconfiguration.

\noindent \textbf{Case 2 --} 
The arguments are similar to Case 1. 
The difference lies in that the processor $p_t$ must eventually receive $needReconfig = {\sf True}$ from a majority of $config$ members (rather than the local core processors) before it moves to  trigger a reconfiguration.
\end{proof}

\begin{lemma}
\label{thQ:triggersControlled}
Starting from an $R_{safe}$ execution, any triggering of Algorithm~\ref{alg:disCongif} (lines~\ref{SSQR:noMajTrigger} and~\ref{SSQR:needReconfTrigger}) related to a specific event (majority collapse or agreement of majority to change $config$), can only cause a one per participant concurrent trigger. After the $config$ has been established, no triggerings that relate to this event take place.
\end{lemma}

\begin{proof}
We consider the two cases that can trigger a reconfiguration (Remark~\ref{rem:reconfThroughInterface}), and assume that $p_t$ is the first processor to trigger $\configEstab()$.
Assume first that $p_t$ has called Algorithm~\ref{alg:disCongif} two consecutive times, without a $config$ being \emph{completely} established between the two calls. 
Note that a processor can access $\configEstab()$ in either of lines~\ref{SSQR:noMajTrigger} or~\ref{SSQR:needReconfTrigger} but not both in a single iteration.
A call to $\configEstab()$ initiates a reconfiguration and thus any subsequent check of $p_t$ in line~\ref{SSQR:configChanging} returns ${\sf False}$ from $p_t$'s $recSA$ layer. 
Thus $p_t$ cannot access lines~\ref{SSQR:noMajTrigger} or~\ref{SSQR:needReconfTrigger} until the reconfiguration has been completed. 
This implies that  $p_t$ can never trigger for a second time unless the new $config$ has been established.
Note that if $p_t$ triggers, another processor satisfying the conditions of line~\ref{SSQR:configChanging} may trigger concurrently, but is also subject to the trigger-once limitation.
On the other hand, due to the exchange of information in Algorithm~\ref{alg:disCongif}, when one processor triggers other processors eventually find their proposals and join the reconfiguration. 
So not every single processor's Reconfiguration Management module needs to trigger.
Convergence to a single $config$ is guaranteed by Algorithm~\ref{alg:disCongif}.

We conclude by indicating that lines~\ref{SSQR:noMajResetAll} and~\ref{SSQR:needReconfResetAll} reset both arrays $noMaj_t[\,]$ and $needReconf_t[\,]$ immediately after $\configEstab()$.
Thus the triggering data used for this event are not used again.
Moreover, upon configuration change, the same arrays are again set to $\sf False$ for the processors that have not triggered Algorithm~\ref{alg:disCongif}  themselves through Algorithm~\ref{alg:SSQR}.
We thus reach a new steady $config$ state, and no more triggerings can take place due to the same event that had caused the reconfiguration.
\end{proof}

\begin{theorem}
\label{thQ:corrUpperApp}
Let $R$ be an execution starting from an arbitrary system state. Algorithm~\ref{alg:SSQR} guarantees that $R$ eventually reaches an execution suffix which is a legal execution.
\end{theorem}

\begin{proof}
By Lemmas~\ref{thQ:reachSteady} and~\ref{thQ:noAbruptConfigTriggered}, we are guaranteed that we reach a safe system state $c_{safe}$ where stale information from the arbitrary initial state cannot force a triggering of new $config$.
This is the suffix $R_{safe}$.
Lemma~\ref{thQ:steadyRemainsSteady} ensures that after we have reached $c_{safe}$, and until a new triggering takes place that is caused by a loss of majority or a majority of the $config$ deciding to reconfigure, the current $config$ will not be changed for any other reason.
Lemma~\ref{thQ:triggersControlled} guarantees that after a change, we return to a steady $config$ state.
Hence, $R_{safe}$ is a legal execution.
\end{proof}

\subsection{Joining Mechanism}
\label{sec:join}
Every processor that wants to become a participant, uses the snap stabilizing data-link protocol (see Section~\ref{s:sys}) so as to avoid introducing stale information after it establishes a connection with the system's processors. 
Algorithm~\ref{alg:disCongif} enables a joiner to obtain the agreed $config$ when no reconfiguration is taking place.
Note that, in spite of having knowledge of this $config$, a processor should only be able to participate in the computation if the application allows it.
In order to sustain the self-stabilization property, it is also important that a new processor initializes its application-related local variables to either default values or to the latest values that a majority of the configuration members suggest.
The joining protocol, Algorithm~\ref{alg:join}, illustrates the above and introduces joiners to the system, but only as \emph{participants} and not as $config$ \emph{members}.

The critical difference between a participant and a joiner is that the first is allowed to send configuration information via the $recSA$ layer, whereas the latter may only receive.

\setlength{\intextsep}{0pt}
%
\begin{algorithm}[t]

\caption{Self-stabilizing Joining Mechanism; code for processor $p_i$}
\label{alg:join} 
\begin{footnotesize}

\noindent {\bf Interfaces.}
The algorithm uses following interfaces from Algorithm~\ref{alg:disCongif}.
$\noReconfig()$ returns $\sf True$ if a reconfiguration is not taking place.
$participate()$ makes $p_i$ a participant. 
$getConfig()$ returns the agreed configuration from Algorithm~\ref{alg:disCongif} or $\bot$ if reconfiguration is taking place.
The $passQuery()$ interface to the application, returns a $\sf True/False$ in response to granting a permission to a joining processor.

{\bf Variables.}
$FD[]$ as defined in Algorithm~\ref{alg:disCongif}.
$state[]$ is array of containing application states, where $state[i]$ represents $p_i$'s local variables and $state[j]$ the state that $p_i$ most recently received by $p_j$. 
$pass[]$ collects all the passes that $p_i$ receives from configuration members.

{\bf Functions.} $resetVars()$ initializes all variables related to the application based on default values.
$initVars()$ initializes all variables related to the application based on states exchanged with the configuration members.\\

{\bf procedure} $join()$ \Begin{\label{JOIN:start}
\lForEach{$p_j \in FD$}{$pass[j] \gets {\sf False}$} \label{JOIN:resetPass}
{\bf do forever} \Begin{\label{JOIN:doForever}
\If{$p_i \not \in FD[i].part$\label{JOIN:checkPart}}{
$resetVars()$\; \label{JOIN:resetVars}
\Repeat{$p_i \in FD[i].part$}{
\label{JOIN:repeatStart}
\textbf{let} $comConf = getConfig()$\; \label{JOIN:readConfig}
\If{$\noReconfig()$ $\land$ $(|\{p_j: p_j\in comConf \cap FD[i] \land pass[j] = {\sf True}\}| > \frac{|comConf|}{2})$ \label{JOIN:checkConfig}}{
$initVars()$\; \label{JOIN:initVars}
$participate()$\; \label{JOIN:bePartpnt}
} 

\lForEach{$p_j \in FD[i]$}{\textbf{send}$({\sf ``Join"})$} \label{JOIN:send}
}
\label{JOIN:repeatEnd}
}
}
}

{\bf upon receive} $({\sf ``Join"})$ \textbf{from} $p_j \in FD \setminus FD[i].part$ {\bf do} \Begin{
\lIf{$p_i \in config$ $\land$ $\noReconfig() = {\sf True}$\label{JOIN:assessPass}}
{\textbf{send}$(\langle passQuery(), state_i\rangle)$\label{JOIN:sendPass}}
}
{\bf upon receive} $m =\langle pass, state \rangle$ \textbf{from} $p_j \in FD$ {\bf do} \Begin{
\lIf{$p_i \not \in FD[i].part$}{
$\langle pass[j], state[j] \rangle \gets m$} \label{JOIN:receiveJoiner} 
}

\end{footnotesize}
\end{algorithm}



\subsubsection{Algorithm description}
The algorithm is executed by non-participants and participants alike.

\paragraph{The joiner's side.}
Upon a call to the $join()$ procedure, a joiner sets all the entries of its $pass[\,]$ array to $\sf False$ (line~\ref{JOIN:resetPass}) and resets application-related variables to default values, (lines~\ref{JOIN:resetVars}).
The processor then enters a do-forever loop, the contents of which it executes only while it is not a participant (line~\ref{JOIN:checkPart}).
A joiner then enters a loop in which it tries to gather enough support from a majority of configuration members.
In every iteration, the joiner sends a $\sf ``Join"$ request (line~\ref{JOIN:send}) and stores the $\sf True/False$ responses by any configuration member $p_j$ in $pass[j]$, along with the latest application $state$ that $p_j$ has send. 
If a majority of active members has granted a $pass = \sf True$ and there is no reconfiguration taking place, then $participate()$ is called to allow the joining processor to become a participant.

\paragraph{The participant's side.}
A participant only executes the do--forever loop (line~\ref{JOIN:doForever}), but none of its contents since it always fails the condition of line~\ref{JOIN:checkPart}. 
Participants however respond to join requests (line~\ref{JOIN:assessPass}) by checking whether a joining processor has the correct configuration, and whether a reconfiguration is not taking place, as well as if the application can accept a new processor.
If the above are satisfied then the participant sends a $pass = \sf True$ and its applications' $state$, otherwise it responds with $\sf False$. 

\subsubsection{Correctness}
\label{app:join}
The term \emph{legal join initiation} indicates a processor's attempt to become a participant by initiating Algorithm~\ref{alg:join} on line~\ref{JOIN:start}, and not on any other line of the $join()$ procedure. If the latter case occurred it would indicate a corruption to the program counter.

\begin{lemma}
\label{thJ:boundedCorruptJoins}
Consider an arbitrary initial state in an execution $R$. 
There are up to $N$ possible instances of processors introducing corruption to the system.
\end{lemma}

\begin{proof}
Processors may be found with an uninitialized or falsely initialized local state due to a transient fault in their program counter which allowed them to reach line~\ref{JOIN:initVars} without a legal join initiation.
In an arbitrary initial state, any processor with stale information may manage to become a participant.
There can be up to $N$ such processors, i.e., the maximal number of active processors.
Nevertheless, a processor trying to access the system after this, is forced to start the execution of $join()$ from line~\ref{JOIN:start}.

\end{proof}

\begin{claim}
\label{thJ:noReconfJoin}
Consider any processor $p_i$ performing a legal join initiation. 
In the existence of other participants in the system, this processor never becomes a participant through the $join()$ procedure during reconfiguration.
\end{claim}

\begin{proof}
We consider the situation where participants exist and reconfiguration is taking place, thus $\noReconfig()$ is $\sf False$.
In order for $p_i$ to become a participant, it needs to gather a pass from at least a majority of the configuration members.
This can only happen if a configuration is in place, and if each of these members is not reconfiguring.
Thus if a pass is granted, it must be that during the execution more than a majority of $\sf True$ passes have arrived at $p_i$.
Note that since the propagation of passes is continuous if a reconfiguration starts, then passes can also be retracted.
Finally, since getting a majority of passes can coincide with the initialization of a reconfiguration, we note that due to asynchrony this processor is considered a participant of the previous configuration, since it has full knowledge of the system's state and is also known by the previous configuration members.
\end{proof}


\begin{lemma}
\label{thQ:noReconfByJoiner}
Consider an execution $R$ where Lemma~\ref{thJ:noReconfJoin} holds, such that during $R$, a processor $p$ becomes a participant. 
Then $p$ cannot cause a reconfiguration, unless there exists a majority of the configuration set, or if there is no majority of the $config$ that requires a reconfiguration.
\end{lemma}

\begin{proof}
We assume that $p$ enters the computation with a legal join initiation. 
If $p$ triggers a reconfiguration in the absence of the above two cases, then this implies that $p$ has managed to become a participant while carrying corrupt information which have triggered a reconfiguration either directly or indirectly (through Algorithm~\ref{alg:SSQR}).
Corruption can either be local or in the communication links.
Since the snap-stabilizing data-link protocol runs before the processor calls $join()$, this removes data-link corruption for newly joining participants. 
We turn to the case of a corrupt local state.
By the legal join initiation assumption, $p$ must have reset its state on line~\ref{JOIN:resetVars}.
Before joining, the majority of members must acknowledge the latest state of $p$ and $p$ initiates its variables to legal values.
It is therefore impossible that $p$ can become a participant while it carries a corrupt state.
Therefore, $p$ cannot cause a reconfiguration.
\end{proof}

\begin{theorem}
\label{thJ:finalApp}
Consider an arbitrary initial state of an execution $R$ of Algorithm~\ref{alg:join}. 
We eventually reach an execution suffix in which every joining processor $p$ will continue trying to join a participant if the application allows it. 
Additionally, this new processor cannot trigger a delicate reconfiguration before becoming a participant and cannot trigger a delicate reconfiguration without majority loss or majority agreement after it becomes a participant.
\end{theorem}

\begin{proof}
By Lemma~\ref{thJ:boundedCorruptJoins}, we eventually reach an execution suffix where all joining processors enter the computation with a legal join initiation.
We assume that a reconfiguration is not taking place, that messages sent infinitely often are eventually received infinitely often, and that the application interface invoked by the participating processors allows $p$ to join. 
Then $p$ will eventually have a failure detector including a majority of member processors and will send its $\sf ``Join"$ request to a majority (line~\ref{JOIN:send}).
Since there is no reconfiguration taking place, $p$ must learn the current configuration from Algorithm~\ref{alg:disCongif}, which should agree with the $config$ held by other processors.
Thus each member must grant a pass to $p$ by sending $\sf True$ through line~\ref{JOIN:sendPass}.
Therefore, $p$ will gather a majority supporting its entrance and will eventually satisfy line~\ref{JOIN:checkConfig}.
This allows it to reach line~\ref{JOIN:bePartpnt} and thus $p$ becomes a participant.
Notice that if the application does not give permission of entry via $passQuery()$, then $p$ cannot become a participant unless this changes, but $p$ will continue sending requests.
Finally, Lemma~\ref{thQ:noReconfByJoiner} ensures that the new participant does not cause perturbations to the current configuration, and hence the result. 
\end{proof}


\section{Applications of the Reconfiguration Scheme}
\label{sec:labelCounter}
Using our self-stabilizing reconfiguration scheme, we present a collection of applications designed for more static settings (with a known fixed set of crash-prone processors) and adapt them to be able to run on the more dynamic setting that we describe here.
We first present a general purpose labeling and counter scheme (Sections~\ref{sec:label} and~\ref{sec:counter}) and then proceed to show how to build a  \emph{self-stabilizing} reconfigurable virtually synchronous replicated state machine (Section~\ref{sec:VS}).


\subsection{Labeling Scheme and Algorithm}
\label{sec:label}
Many distributed applications assume access to an unbounded counter, e.g., to provide ballot numbers for consensus in Paxos, tag numbers for distributed shared memory emulation or view identifiers for virtually synchronous reliable multicast~\cite{SSVS}.
An unbounded counter implemented as a 64-bit integer, for example, is practically inexhaustible when initiated at 0 for the lifetime of most known systems.
Transient failures, nevertheless, can immediately drive the counter (\emph{sequence number} or $seqn$) to its maximal value (e.g., $2^{64}$) causing it to wrap. 
Recently, we extended an existing labeling and counter increment scheme 
to enable any processor of a fixed processor set to increment a counter integer attached to an epoch $label$~\cite{SSVS}.
When the counter is exhausted, a new maximal label is used with a non-exhausted $seqn$. 
We now adjust that solution to benefit from our reconfiguration mechanism. 
In this scheme, configuration members are the ones that run the labeling algorithm and maintain a globally maximal label and counter.
The labeling and counter increment algorithms consider every new configuration as a new instance of the corresponding algorithms of~\cite{SSVS}.
To this end, Algorithm~\ref{alg:configLabeling} is a wrapper of the self-stabilizing labeling algorithm of~\cite{SSVS} that retains the initial algorithm as a module (Algorithm~\ref{alg:receiveLabels}) allowing it to cope with reconfiguration. 
We first provide the labeling algorithm, and then extend labels to counters, presenting how counter increments take place.

\begin{algorithm*}[t!]
  \caption{{Self-Stabilizing Labeling Algorithm for Reconfiguration; code for $p_i \in config$}}
%
%
\label{alg:configLabeling}

\begin{footnotesize}

{\bf Variables:} Let $v$ be the size of the configuration $config$ as returned by $getConfig()$.\\
$max[v]$ of label pairs $\langle ml$, $cl \rangle$: $max[i]$ is $p_i$'s largest label pair, $max[j]$ refers to $p_j$'s $max_j[j]$ label pair that was last sent to $p_i$ (canceled when $max[\bullet].cl \neq \bot$).
$storedLabels[v]$: an array of queues of label pairs, where $storedLabels[j]$ holds the labels created by $p_j \in config$. For $p_j \in (config \setminus \{ p_i \})$, $storedLabels[j]$'s queue size is limited to $(v+m)$ w.r.t. label pairs, where $m$ is the maximum number of label pairs that can be in transit in the system. The $storedLabels[i]$'s queue size is limited to $(v(v^2+m))+v$ pairs.

{\bf Interfaces:}
$\noReconfig()$ returns $\sf False$ (from the reconfiguration module) when a reconfiguration is taking place, and $\sf True$ otherwise.
$getConfig()$ returns the current configuration if one exists.
$labelReceiptAction()$ maintains the label arrays by calling the receive function of Algorithm~\ref{alg:receiveLabels}.\\
{\bf Operators:}
$rebuild(v)$ rebuilds the $storedLabels[]$ array of queues and $max[]$ to have $v$ entries. It also adjusts the queue size for the new $v$.
$emptyAllQueues()$ clears all $storedLabels[]$ queues.
$confChange()$ returns $\sf True$ if a reconfiguration has taken place, and $\sf False$ otherwise, by comparing the current label structures with the result of $getConfig()$.

{\bf Macros:}\\ $cleanLP(x) =$ {\bf if} $(\exists \ell \in x.\langle ml,cl \rangle,$ $p_j \not \in curConf):$ $\ell.lCreator = p_j)$ {\bf then return} $\langle \bot, \bot \rangle$ {\bf else return $x$}; 
\label{LAB:cleanLP}


{\bf function} $cleanMax()$ \lForEach{$p_j \in curConf,  \ell \in max[j].\langle ml,cl \rangle :  \ell.lCreator = p_k  \not \in  curConf$}{$max[j] \gets \langle \bot, \bot \rangle$} \label{LAB:cleanMax}

{\bf do forever} \label{LAB:doForever}\Begin{

\If{$\noReconfig() = {\sf True} \land confChange() = {\sf True}$ \label{LAB:uponConfChange}}{
$curConf = getConfig()$\; \label{LAB:newSize}
$rebuild(|curConf|)$\; \label{LAB:rebuild}
$emptyAllQueues()$\;\label{LAB:emptyQs}
$cleanMax()$\;\label{LAB:cleanUponReconf}
$labelReceiptAction(\langle \bot, max[i], p_i \rangle)$\; \label{LAB:findNewMax}
}
}

{\bf upon} $transmitReady(p_k \in curConf \setminus \{ p_i \})$\label{LAB:beginTransmit} \Begin{
\If{$ \noReconfig() = {\sf True} \land confChange() = \sf False$}{{
\bf transmit}$(\langle max[i], max[k] \rangle \gets \langle cleanLP(max[i]) , cleanLP(max[k]) \rangle)$\label{LAB:transmit}}
}

{\bf upon} $receive(\langle sentMax, lastSent \rangle)$ {\bf from} $p_k \in curConf$ \Begin{  \label{LAB:uponReceive}
\If{$ \noReconfig() = {\sf True} \land confChange() = \sf False$}{
$cleanMax()$\; \label{LAB:receiveCleanMax}
$\langle  sentMax , lastSent  \rangle = \langle cleanLP(sentMax) , cleanLP(lastSent) \rangle$\;\label{LAB:cleanReceived}
$labelReceiptAction(\langle sentMax, lastSent, p_k \rangle)$\; \label{LAB:receiveAction}
} 
}

\end{footnotesize}
\end{algorithm*}

\begin{algorithm*}[t!]
%
  \caption{{Self-Stabilizing Labeling Algorithm receipt action; code for $p_i \in config$}}
%
%
\label{alg:receiveLabels}

\begin{small}

{\bf Variables:} For a configuration $config$ with $v = |config|$\\
$max[v]$ of $\langle ml$, $cl \rangle$: $max[i]$ is $p_i$'s largest label pair, $max[j]$ refers to $p_j$'s label pair (canceled when $max[j].cl \neq \bot$).\\

$storedLabels[v]$: an array of queues of the most-recently-used label pairs, where $storedLabels[j]$ holds the labels created by $p_j \in config$. For $p_j \in (config \setminus \{ p_i \})$, $storedLabels[j]$'s queue size is limited to $(v+m)$ w.r.t. label pairs, where $m$ is the maximum number of label pairs that can be in transit in the system. The $storedLabels[i]$'s queue size is limited to $(v(v^2+m))+v$ pairs. The operator $add(\ell)$ adds $lp$ to the front of the queue, and $emptyAllQueues()$ clears all $storedLabels[]$ queues. We use $lp.remove()$ for removing the record $lp \in storedLabels[]$. Note that an element is brought to the queue front every time this element is accessed in the queue.\\
$nextLabel()$ creates a label that is greater than any other label in $storedLabes_i[i]$.\\

%
%



{\bf Notation:} Let $y$ and $y'$ be two records that include the field $x$. We denote  $y$ $=_{x}$ $y'$ $\equiv$ $(y.x$ $=$ $y'.x)$\\

{\bf Macros:}\\
$legit(lp)$ $=$ $(lp$ $=$ $\langle \bullet, \bot \rangle)$~~~\\
$labels(lp)$ $:$ {\Return{$(storedLabels[lp.ml.lCreator])$}}\\
$double(j, lp) = (\exists lp' \in storedLabels[j] : ((lp \neq lp') \land ((lp =_{ml} lp') \lor (legit(lp) \land legit(lp')))))$~~~\\ \label{ln:double}
$staleInfo() = (\exists p_j \in P, lp \in storedLabels[j] : (lp \neq_{lCreator} j)
\lor double(j, lp))$~~~\\ \label{ln:staleInfo}
$recordDoesntExist(j) = (\langle max[j].ml, \bullet \rangle \notin labels(max[j]))$~~~\\
$notgeq(j, lp) = \mathbf{if~}(\exists lp' \in storedLabels[j]$ $:$ 
$(lp'.ml \not \preceq_{lb} lp.ml))$ $\mathbf{then~return}(lp'.ml)$ $\mathbf{else~return}(\bot)$~~~\\
$canceled(lp) = \mathbf{if~}(\exists lp' \in labels(lp)$ $:$ $((lp' =_{ml} lp)$ $\land$ $\neg legit(lp')))$ $\mathbf{then~return}(lp')$ $\mathbf{else~return}(\bot)$~~~\\
$needsUpdate(j)=(\neg legit(max[j]) \land \langle max[j].ml, \bot \rangle \in labels(max[j]))$~~\\
$legitLabels() = \{ max[j].ml : \exists p_j \in P \land legit(max[j]) \}$~\label{ln:legitLabels}~\\
$useOwnLabel()=\mathbf{if~}(\exists lp \in storedLabels[i] : legit(lp))$ $\mathbf{then~}max[i]$ $\gets$ $lp$ $\mathbf{else~}storedLabels[i].add(max[i]$ $\gets$ $\langle nextLabel(), \bot \rangle)$~\label{ln:useOwnLabelDef}
\tcp{For every $lp \in storedLabels[i]$, we pass in $nextLabel()$ both $lp.ml$ and $lp.cl$.}

{\bf function} $labelReceiptAction(\langle sentMax, lastSent, p_k \rangle)$
\Begin{
$max[k]$ $\gets$ $sentMax$\; \label{ln:exposeStore}
\lIf{$\neg legit(lastSent)$ $\land$ $max[i] =_{ml} lastSent$}{$max[i] \gets lastSent$} \label{ln:lastSentCancel}

    \lIf{$staleInfo()$}{$storedLabels.emptyAllQueues()$} \label{LBLln:clean}

     \lForEach{$p_j \in P : recordDoesntExist(j)$}{$labels(max[j]).add(max[j])$} \label{ln:add}

        \lForEach{$p_j \in P, lp \in storedLabels[j] : (legit(lp) \land (notgeq(j,lp)\neq \bot))$}{$lp.cl \gets notgeq(j,lp)$} \label{ln:cancelLabels}

        \lForEach{$p_j \in P, lp \in labels(max[j]) : (\neg legit(max[j]) \land (max[j] =_{ml} lp) \land legit(lp))$}{$lp \gets max[j]$} \label{ln:receivedCanceled}

        \lForEach{$p_j \in P, lp \in storedLabels[j] : double(j, lp)$}{$lp.remove()$} \label{ln:remove}

        \lForEach{$p_j \in P : (legit(max[j]) \land (canceled(max[j])\neq \bot))$}{$max[j] \gets canceled(max[j])$} \label{ln:cancelMax}

    \lIf{$legitLabels() \neq \emptyset$}{$max[i] \gets \langle \max_{\prec_{lb}}(legitLabels()), \bot \rangle$} \label{ln:adopt}

    \lElse{$useOwnLabel()$}  \label{ln:useOwnLabel}

}

\end{small}

\end{algorithm*}

\subsubsection{Description of Algorithm~\ref{alg:configLabeling}}

\noindent{\bf Outline.} 
The algorithm is run only by configuration members.
Each label is marked by its creator's identifier and any two labels are compared first as to their creator identifier and then as to a set of integers using the operator $\prec_{lb}$. 
Labels by the same processor can be \emph{incomparable}.
A processor that is aware of a set of labels with its own identifier, can always create a greater label.
The aim is for members to learn of any valid label in the system and finally result to the globally greatest one, and to this end, members exchange labels.
We refer the reader to~\cite{SSVS} for more details on the label structure. 
The algorithm ignores labels by non-member creators by setting them to $\bot$.
If Algorithm~\ref{alg:disCongif} reports that a reconfiguration is taking place (via $\noReconfig$), no actions are taken.
Upon the completion of a reconfiguration, every member's local label storage is rebuilt to reflect the new configuration set, and all the label queues are emptied. 
Newly joined processors are assumed to join with initialized links and empty label structures and thus cannot introduce corrupt information. 
If a reconfiguration is not reported, member $p$ of the configuration periodically sends and receives its locally maximal labels with all the other members.
Whenever it sends or receives a new label pair, it checks whether this has the identifier of one of the current members. 
The received label pairs are passed to the receive function (Algorithm~\ref{alg:receiveLabels}), which is exactly the same as the one in~\cite{SSVS}.
This always returns a local maximal label either by some other member or by the caller itself. 
%
We underline that for every configuration we can find a greatest label, but it cannot be guaranteed that the label of a configuration will continue being the greatest in a following configuration.

%

\paragraph{Variables.} Processor $p_i$ that belongs to a configuration $config$ with $v=|config|$, has an array $max_i[v]$, where $max_i[i]$ contains the local maximal label $p_i$ knows, and $max_i[j]$ the last value that $config$ member $p_j$ has send.
The array of label queues $storedLabels_i[]$ holds a queue of size $(v(v^2+m))+v$ in $storedLabels_i[i]$ for labels concerning processor $p_i$, and queues of size $v+m$ for all other configuration members in $storedLabels_i[j]$.

\paragraph{Interfaces and Operators.}
The $\noReconfig()$ interface of Algorithm~\ref{alg:disCongif}  returns $\sf False$ when a reconfiguration is taking place, and $\sf True$ otherwise. It also returns the current configuration if one exists via $getConfig()$.
$labelReceiptAction()$ maintains the label arrays by calling the receive function of Algorithm~\ref{alg:receiveLabels}.
$emptyAllQueues()$ clears all $storedLabels[]$ queues.
$confChange()$ returns $\sf True$ if a reconfiguration has taken place, and $\sf False$ otherwise, by comparing the current label structures with the result of $getConfig()$.

\paragraph{During reconfiguration.}  The conditions $ \noReconfig() = {\sf True}$ and  $confChange() = \sf False$ of lines~\ref{LAB:uponReceive} and~\ref{LAB:transmit} prevent transition and reception of labels during reconfiguration and before the label structures have been reset after reconfiguration has taken place.  

\paragraph{After reconfiguration.} Lines~\ref{LAB:uponConfChange}--\ref{LAB:findNewMax} are only executed upon a completed reconfiguration.
Line~\ref{LAB:newSize} gets the new configuration from Algorithm~\ref{alg:disCongif}, and line~\ref{LAB:rebuild} uses $rebuild(config)$ to adjust $max[]$ array so that it holds the entries of any processor that also belonged to the previous configuration, removing the ones by old removed members and adding new fields for the labels of new $config$ members.
Line~\ref{LAB:cleanUponReconf} removes labels from the new $max[]$ that were not created by configuration members. 
The effect is analogous for $storedLabels[]$, where $v+m$-sized label queues are added and removed to reflect the changes in the $config$ composition, but noting that $v$ is now the cardinality of the new configuration set.
These queues are emptied by line~\ref{LAB:emptyQs}.
Finally, the processor finds a new local maximal label either from the ones that it has in $max[]$, or by creating a new one with its one creator identifier (line~\ref{LAB:findNewMax}).

\paragraph{Label exchange and maintenance.} If a reconfiguration is not taking place, then a member periodically sends to every other configuration member $p_j$ its own local maximal label and the last label that $p_j$ sent for the local maximal label.
Note that messages from non-members are discarded before being sent, so they are not propagated (lines~\ref{LAB:beginTransmit}--\ref{LAB:transmit}).

Similarly when $p_i$ receives a message as the one described above, it cleans its $max[]$ array and the received two labels from non-member labels and passes them to the $labelReceptionAction()$ (Algorithm~\ref{alg:receiveLabels}).
The description of the inner workings of the $labelReceptionAction()$ is given in great detail in~\cite{SSVS}. 
As an overview, it first stores $p_j$'s value in $max_i[j]$ and processes it along with the other label that was received which reflected $p_i$'s maximal label that was received most recently by $p_j$.
In general terms, it performs housekeeping of labels in $storedLabels[]$ and $max[]$ such that only the greatest label per processor is considered before a local maximal label is chosen by $p_i$.

\subsubsection{Correctness}
\textbf{Outline.} We first establish that after a full iteration, of Algorithm~\ref{alg:receiveLabels} every configuration member does not sustain and does not introduce  any label that has a creator identifier by a non-member.
It is then possible to map every configuration of an execution to an instance of the fixed-set labeling algorithm of~\cite{SSVS} and thus induce the correctness proof therein.
This ensures that eventually a maximal label is found.
Since the algorithm is aware of configuration changes (that exist after reconfigurations), it can use this events to achieve more efficient convergence for a configuration that follows a reconfiguration (in contrast to one that is in place in an initial arbitrary state). 
We thus reach to the above bounds where $N$ is an upper bound on the system and thus a possible upper bound for the configuration size.

{\em Note about participants that are not members.} Algorithms~\ref{alg:configLabeling} and~\ref{alg:receiveLabels} are run by processors that are strictly \emph{members} of the current configuration and only if a reconfiguration is not taking place.
Members send, receive and take into account only labels that come from and concern processors that are members.
In this perspective, it should not be possible for a non-member processor to add labels to the system in a way that will affect the system.
If a processor for any reason stops being a member after a configuration but remains an active participant, then once reconfiguration takes place, the member processors stop considering labels from this processor and \textit{void} (set to $\langle \bot, \bot \rangle$) any label pair from this processor. 

\begin{lemma}
\label{thL:staleInfo}
Consider an arbitrary starting system state of an execution $R$ where a configuration does not change throughout $R$, a processor $p_i \in config$ reaches a local state where there are no labels from any processor $p_s \not \in config$, immediately after the first iteration of Algorithm~\ref{alg:configLabeling} that includes a receive event. 
Moreover, a label by $p_s \not \in config$ can never be introduced to the local state of $p_i$ at any point after the configuration is established.
\end{lemma}

\begin{proof}
Notice that while the configuration does not change because the reconfiguration module finds that this configuration can serve the system, there could be corruption relating to the local variables of the labeling algorithm.
We first establish that at any point after a complete execution of lines~\ref{LAB:uponReceive}--\ref{LAB:receiveAction}, $p_i$'s local state never contains a label by a non-member $p_s$.
Assume for contradiction that such a label exists at some system state after the execution of these lines.
This label can exist either in (i) the $max[]$ array or (ii) the $storedLabels[]$ array.

Member $p_i \in config$ acknowledges that $p_s \not \in config$, since the configuration is agreed and any inconsistency will cause a reconfiguration in Algorithm~\ref{alg:disCongif}. 
We have already assumed for the purposes of the proof that a reconfiguration does not happen throughout $R$.
Labels arriving from $p_s$ cannot start a receive event by the conditions line~\ref{LAB:uponReceive}. 
If a label created by $p_s$ is received from some processor $p_j \in config$ then this is set to $\bot$ by line~\ref{LAB:cleanReceived}. 
So no incoming labels can enter the local state.

We now consider the system state immediately after line~\ref{LAB:receiveAction} returns from executing Algorithm~\ref{alg:receiveLabels}.
Line~\ref{LAB:receiveCleanMax} of Algorithm~\ref{alg:configLabeling} guarantees that $max[]$ is cleaned of labels that appear as created by $p_s$.
Assume that a label by $p_s \not \in config$ exists in the $storedLabels[]$ structure.
There are two cases depending on the structure of the $storedLabels[]$ array of queues.\\
\noindent {\bf Case 1 --} \textbf{Processor $p_i$ does not have an entry for $p_s$ in $storedLabels[]$.}
This must be true for a legal state, since a member processor should only keep queues in $storedLabels_i[]$ that relate to member processors.
So $p_s$'s label must reside in a queue of $storedLabels_i[]$ that is not intended for $p_s$'s labels. 
In such a case since $p_s \not \in config$, line~\ref{LBLln:clean} of Algorithm~\ref{alg:receiveLabels} will cause the flushing of the queues because $staleInfo() = \sf True$.

\noindent {\bf Case 2 --} \textbf{Processor $p_i$ has an entry for $p_s \not \in config$ in $storedLabels[]$.}
This is the result of transient fault and implies that the label structures queues where not prepared for the new configuration.
Nevertheless, in a complete iteration of lines~\ref{LAB:doForever}--\ref{LAB:receiveCleanMax}, $confChange()$ of line~\ref{LAB:uponConfChange} will return $\sf True$ because of this discrepancy between $config$ and the composition of the label structures.
But this causes the algorithm to move to line~\ref{LAB:uponConfChange} and so execute lines~\ref{LAB:newSize}--\ref{LAB:cleanUponReconf} thus emptying the $storedLabels[]$ and cleaning $max[]$ of $p_s$-created labels.

Therefore, immediately after a receive is completed for Case~1 or after the execution of line~\ref{LAB:cleanUponReconf} in Case~2, there cannot be a label created by $p_s \not \in config$ inside $p_i$'s local state.
Furthermore, we have established that such incoming labels cannot enter $p_i$'s state.
Hence we reach to the result.
It is evident that once $p_i$'s state is cleaned, it cannot transmit any such labels via line~\ref{LAB:transmit}.
\end{proof}

\begin{lemma} 
\label{thL:map2SSVS}
Consider an execution $R$ of Algorithm~\ref{alg:configLabeling} where Lemma~\ref{thL:staleInfo} holds and no reconfiguration takes place throughout $R$. 
It holds that this instance of the problem of providing a self-stabilizing labeling scheme in the dynamic setting can be reduced to the one of the problem of a fixed processor set self-stabilizing labeling scheme problem, where the fixed set is the common configuration. 
\end{lemma}

\begin{proof}
Note that the fixed-set version allows for processors from a specific non-changing set to crash but not rejoin.
We identify the agreed configuration set to this fixed set of possibly active processors.
Processors of the configuration may crash but may not rejoin.
Algorithms~\ref{alg:configLabeling} and~\ref{alg:receiveLabels} are run only by member processors.
As established by Lemma~\ref{thL:staleInfo}, labels can have the identifier of any of the processors in this set but of no other processor.
The communication links between configuration members are of bounded capacity and no labels from a non-member can be added to the local state of the member processors (since they are not considered by line~\ref{LAB:uponReceive}) nor can be added to the communication links between the processors (line~\ref{LAB:transmit}).
Hence, the algorithm reduces every execution $R$ to an instance of providing self-stabilizing labels. 
\end{proof}
\vspace{.5em}

The corollary follows since we can use the solution to the fixed processor set problem to solve each instance of the problem whenever a reconfiguration is not taking place.
In particular Algorithm~\ref{alg:receiveLabels} and line~\ref{LAB:transmit} comprise of the solution given in~\cite{SSVS} to solve the fixed-set problem.

%
%
%

\begin{corollary}
\label{thL:maxReached}
Consider an execution $R$ of Algorithm~\ref{alg:configLabeling} starting in an arbitrary state.  
While a reconfiguration is not taking place, the solution provided by~\cite{SSVS} can be used to guarantee that a maximal label created by a member of the configuration will eventually be adopted by all active members.
\end{corollary}

\noindent The theorem follows. 

\begin{theorem}
\label{thL:finalApp}
\label{thL:uniqueLabel}
Starting in an arbitrary state, Algorithm~\ref{alg:configLabeling} provides a maximal label. If a reconfiguration does not take place then there can be up to $O(N(N^2+m))$ label creations before a maximal label is established. If a reconfiguration takes place then there can be up to $O(N^2)$ label creations.
\end{theorem}

\begin{proof}
By Corollary~\ref{thL:maxReached} Algorithm~\ref{alg:configLabeling} reaches a maximal label.
By the results of Dolev et al~\cite{SSVS} the worse case takes place when starting in an arbitrary state which requires at most $O(N(N^2+m))$ label creations, where $m$ is the system's communication link capacity in labels and in our case $N \geq v$.
A processor may create a label with its own identifier the label it considers as maximal is canceled and it knows of no other labels.
Labels that are possibly unknown to some processor and may cancel its maximal label, are found either in the local state ($|max[\bullet]| \leq N$  and $|storedLabels[N]| \leq 2N(N^2 + m-1)$), or in transit between configuration members.
Note that in an arbitrary state, the configuration may happen to be valid and not require a change, although the labels in the system may be corrupt.
This is a worse case scenario that requires possibly $O(N(N^2+m))$ label creations until the maximal label is reached.

Nevertheless, if a reconfiguration takes place, then the snap stabilizing link will clear the $m$ labels found in the communication links and any active participant will empty their local queues in $storedLabels[]$. 
The only source of labels that could possibly force a processor to create new labels in order to find a maximal, are the $max[]$ arrays.
There are up to $N$ such arrays of size at most $N$, and hence $O(N^2)$ possible label creations.
\end{proof}

\subsection{Counter Increment Algorithms}
\label{sec:counter} 
Using the labeling algorithm, we implement a practically infinite self-stabilizing counter, inexhaustible for the lifetime of most known systems (e.g., with upper bound of $2^{64}$).
We first show how to move from labels to counters and then provide a description of the algorithm along with a correctness proof.

\subsubsection{Description}

\begin{algorithm}
%
  \caption{{Self-stabilizing Counter Management Algorithm for Reconfiguration; code for $p_i \in config$}}
%
%
\label{alg:configCounting}

\begin{footnotesize}
(Let $v$ be the size of the configuration $config$ as returned by $getConfig()$).\\
{\bf Variables:}
A counter is a triple $\langle lbl, seqn, wid \rangle$ where $lbl$ is a label of the labeling scheme, $seqn$ is the sequence number related to $lbl$, and $wid$ is the identifier of the creator of this $seqn$. 
A counter pair $\langle mct, cct\rangle$ extends a label pair. $cct$ is a canceling counter for $mct$, such that $cct.lbl \not \prec_{lb} mct.lbl$ or $cct.lbl = \bot$. 
We rename structures $max[]$ and $storedLabels[]$ of Alg.~\ref{alg:receiveLabels} to  $maxC[]$ and $storedCnts[]$ that hold counter pairs instead of label pairs.
Array $maxC[v]$ of counter pairs $\langle mct$, $cct \rangle$: $maxC[i]$ is $p_i$'s largest counter pair, $maxC[j]$ refers to $p_j$'s $maxC_j[j]$ counter pair that was last sent to $p_i$ (canceled when $maxC[\bullet].cct.lbl \neq \bot$).
$storedCnts[v]$: an array of queues of the counter pairs, where $storedCnts[j]$ holds the counters with labels created by $p_j \in config$. 
For $p_j \in (config \setminus \{ p_i \})$, $storedCnts[j]$'s queue size is limited to $(v+m)$ w.r.t. counter pairs, where $m$ is the maximum number of counter pairs that can be in transit in the system. 
The $storedCnts[i]$'s queue size is limited to $(v(v^2+m))+v$ pairs.
\label{CCT:var} \\

{\bf Interfaces:}
$\noReconfig()$ returns $\sf False$ (from the reconfiguration module) when a reconfiguration is taking place, and $\sf True$ otherwise.
$getConfig()$ returns the current configuration if one exists.
$counterReceiptAction((\langle cnt, cnt, id \rangle)$ - executes the function $counterReceiptAction((\langle lbl, lbl, id \rangle)$ of Algorithm~\ref{alg:receiveLabels} adjusted for counter structures and handling counters.
For counter pairs with the same $mct$ label, only the instance with the greatest counter w.r.t. $\prec_{ct}$ is retained. 
In case where one counter is canceled, we keep the canceled.
For ease of presentation we assume that a counter with a label created by $p_i$ in line~\ref{ln:useOwnLabel} of Algorithm~\ref{alg:receiveLabels}, is initiated with a $seqn = 0$ and $wid=i$. 
A call of $counterReceiptAction()$ (without arguments) essentially ignores lines~\ref{ln:exposeStore} and~\ref{ln:lastSentCancel} of Alg.~\ref{alg:receiveLabels}.

{\bf Operators:} 
$rebuild(v)$ rebuilds the $storedCnts[]$ array of queues and $maxC[]$ to have $v$ entries. It also adjusts the queue size for the new $v$.
$emptyAllQueues()$ clears all $storedCnts[]$ queues.
$confChange()$ returns $\sf True$ if a reconfiguration has taken place, and $\sf False$ otherwise, by comparing the current counter structures with the result of $getConfig()$.
$enqueue(ctp)$ - places a counter pair $ctp$ at the front of a queue. 
If $ctp.mct.lbl$ already exists in the queue, it only maintains the instance with the greatest counter w.r.t. $\prec_{ct}$, placing it at the front of the queue. 
If one counter pair is canceled then the canceled copy is the one retained. \label{CCT:operations}\\ 
{\bf Notation:} Let $y$ and $y'$ be two records that include the field $x$. We denote  $y$ $=_{x}$ $y'$ $\equiv$ $(y.x$ $=$ $y'.x)$.\\

{\bf Macros:}\\ $cleanCP(x) =$ {\bf if} $(\exists \ell \in x.\langle mct,cct \rangle.lbl,$ $p_j \not \in curConf):$ $\ell.lCreator = p_j)$ {\bf then return} $\langle \bot, \bot \rangle$ {\bf else return $x$}; \\
	$exhausted(ctp)$ $=$ $(ctp.mct.seqn$ $\geq$ $2^{64})$~~~\\
	$legit(ctp)=(ctp.cct = \bot \rangle )$\\	
	$cancelExhausted(ctp) :$ {$ctp.cct \gets ctp.mct$}\\  
	$cancelExhaustedMaxC() :$ 	\lForEach{$p_j\in config,\ c \in maxC[j]: exhausted(c)$}{$cancelExhausted(maxC[j])$}  
	 \label{CCT:cancExh}
	

{\bf function} $cleanMax()$ \lForEach{$p_j \in curConf,  \ell \in maxC[j].\langle mct,cct \rangle.lbl:  \ell.lCreator = p_k  \not \in  curConf$}{$maxC[j] \gets \langle \bot, \bot \rangle$} 

{\bf do forever} \Begin{

\If{$\noReconfig() = {\sf True} \land confChange() = {\sf True}$}{
$curConf = getConfig()$\; 
$rebuild(|curConf|)$\; 
$emptyAllQueues()$\;
$cleanMax()$\;
$counterReceiptAction(\langle \bot, maxC[i], p_i \rangle)$\; 
}
}

{\bf upon} $transmitReady(p_k \in curConf \setminus \{ p_i \})$ \Begin{
\If{$ \noReconfig() = {\sf True} \land confChange() = \sf False$}{{
\bf transmit}$(\langle maxC[i], maxC[k] \rangle \gets \langle cleanCP(maxC[i]) , cleanCP(maxC[k]) \rangle)$}
}

{\bf upon} $receive(\langle sentMax, lastSent \rangle)$ {\bf from} $p_k \in curConf$ \Begin{  
\If{$ \noReconfig() = {\sf True} \land confChange() = \sf False$}{
$cleanMax()$\; 
$\langle  sentMax , lastSent  \rangle = \langle cleanCP(sentMax) , cleanCP(lastSent) \rangle$\;
$counterReceiptAction(\langle sentMax, lastSent, p_k \rangle)$\; 
} 
}

\textbf{upon request to increment counter $inc()$} \Begin{ 
\lIf{$ \noReconfig() = {\sf True} \land confChange() = \sf False$}{\textbf{return} $incrementCounter(getConfig())$} \label{CCT:incr}}

\end{footnotesize}
\end{algorithm}

\paragraph{Counters.}
The labeling scheme used above, can be used to implement counters.  
The idea is to extend the labeling scheme to handle {\em counters}, where a counter is a triple $\langle label, seqn, wid \rangle$, where $seqn$ is an integer {\em sequence number}, ranging from $0$ to $2^{b}$, where $b$ is large enough, say $b=64$; and $wid$ the processor identifier of the $seqn$ creator (not necessarily the same as the $label$'s creator). 
Specifically, we say that counter $ct_1=\langle \ell_1, seqn_1, wid_1\rangle$ is {\em smaller} than counter $ct_2=\langle \ell_2, seqn_2, wid_2\rangle$, and write that $ct_1 \prec_{ct} ct_2$, if ($\ell_1 \prec \ell_2$), or (($\ell_1 = \ell_2$) and ($seqn_1<seqn_2$)), or $((\ell_1 = \ell_2$) and ($seqn_1 = seqn_2$) and ($wid_1<wid_2))$. 
Note that when processors have the same label, the above relation forms a total ordering and processors can increment a shared counter also when attempting to do so concurrently. 
Also, when the labels of the two counters from the same processor are incomparable, the counters are also incomparable.

\paragraph{Outline.} 
Algorithm~\ref{alg:configCounting} maintains counters as Algorithm~\ref{alg:configLabeling} maintains labels. 
Counter increment for a participant that is not a configuration member is seen in Algorithm~\ref{alg:cntrIncrNonMember}, and for a member, which also bears the responsibility to maintain the maximal counter, in Algorithm~\ref{alg:cntrIncrMember}. 
A participant that wants to increment the counter, first queries the configuration for the maximal counter.
It only needs to consider the responses of the majority, because the intersection property of majorities guarantees at least one member that holds the most recent value of a completed counter increment.
Having this maximal counter, $seqn$ is incremented and written back to the configuration, awaiting for acknowledgments from a majority.
This is, in spirit, similar to the two-phase write operation of MWMR register implementations, focusing on the sequence number rather than on an associated value. 

Configuration members have the extra task of retaining the maximal value and ensuring the convergence property.
To that end, they exchange their maximal counter and update their counter structures in the same way as in the labeling algorithm. 
The maximal counter needs to have the maximal label held by the configuration members and the highest sequence number, breaking symmetry with the writer identifier. 
If this maximal sequence number is {\em exhausted}, members proceed to find a new maximal label, using the maximal sequence number known for this epoch label (possibly 0 if it is a newly created or unused label). 

\paragraph{Variables.} The counter management algorithm (Algorithm~\ref{alg:configCounting}) uses the same structures and procedures as the labeling algorithm, but now with counters instead of labels. 
Specifically we name the $max[\bullet]$ array to $maxC[\bullet]$ and the $storedLabels\bullet]$ array of label queues to $storedCnts[\bullet]$ array of counter queues.
Processors only hold one copy of a counter with the same label, namely the one with the highest $seqn$ (breaking ties with $wid$). 
Processors send and receive pairs of counters $\langle mct, cct \rangle$ where the first is the believed \emph{maximal} counter, and the second a \emph{cancelling} counter.
While $cct=\bot$ then $mct$ is not-cancelled and can be considered as a valid candidate for the local maximal counter.
When a counter becomes exhausted, i.e., its $seqn$ exceeds $2^{b}$, $mct$ becomes canceled by assigning $cct \gets mct$.

\paragraph{Interfaces and Operators.}
$counterReceiptAction(\langle lbl, lbl, id \rangle)$ is an interface to Algorithm~\ref{alg:receiveLabels} which now acts on and maintains counter pairs rather than label pairs.
It concludes by naming a local maximal label.
We do not present this algorithm again as it is essentially the same, less naming ``counter'' in place of ``label".
The rest of the interfaces and operators are detailed at the beginning the algorithms in which they are being used for the first time.

\paragraph{Maintaining a maximal counter.} Algorithm~\ref{alg:configCounting} presents how Algorithm~\ref{alg:configLabeling} is modified to maintain a largest counter (rather than just the largest label), by exchanging the local maximal counters with other configuration members.
The conditions that prevent send/receive during reconfiguration are the same.
Similarly, the above actions do not take place after reconfiguration, until the structures are rebuild and reset in exactly the same way as described by the labeling algorithm.
In addition, the same conditions that prevent send/receive also prevent incrementing the counter (line~\ref{CCT:incr}).

\begin{algorithm*}[t!]
   \caption{Counter Increment for configuration member $p_i \in config$; code for $p_i$}

\label{alg:cntrIncrMember}
\begin{footnotesize}
{\bf Variables, Interfaces and Macros} are found in Algorithm~\ref{alg:configCounting}. \\
\textbf{Operator:} $abort()$ aborts the procedure and returns $\bot.$

{\bf Notation:} Let $y$ and $y'$ be two records that include the field $x$. We denote  $y$ $=_{x}$ $y'$ $\equiv$ $(y.x$ $=$ $y'.x)$.\\
{\bf Macros:}\\ 
	$retCntrQ(ct) :$ {\bf return} $(storedCnts[ct.lbl.lCreator])$\\
	$legitCounters()$ $=$ $\{maxC[j].mct: \exists p_j \in config \land legit(maxC[j])\} $\\
	$getMaxSeq():$ {\bf return} $max_{wid} (\{max_{seqn}(\{ctp:ctp.mct \in legitCounters() \land maxC[i] =_{mct.lbl} ctp\})\})$
	
\BlankLine


{\bf procedure} $incrementCounter(config)$ \Begin{
	$majRead()$\;
	\Repeat{$legit(maxC[i]) \land \neg exhausted(maxC[i])$}{
		$findMaxCounter()$; 
	}
	{\bf let} $newCntr = \langle maxC[i].mct.lbl, maxC[i].mct.seqn + 1, i\rangle$\label{CIM:incrCntr}\; 
	\lIf{$majWrite(newCntr)$} 
		{$maxC[i] \gets newCntr$\; $retCntrQ(maxC[i].mct).enqueue(maxC[i])$}
	\textbf{return} $maxC[i]$\;
}

{\bf procedure} $majRead()$ \Begin{ \label{CIM:majRead}
	\lForEach {$p_j \in config$}{{\bf send} $majMaxRead()$}
	\While{waiting for responses from majority of $config$}{
		{\bf upon receipt of} $m$ {\bf from} $p_j \in config$ {\bf do}\\
			\lIf{$m = \sf Abort$}{$abort()$ \textbf{else} $maxC[j] \gets cleanCP(m)$ \nllabel{CIM:readCntrs}} 
	}
}

{\bf upon request for} $majMaxRead()$ {\bf from} $p_j$ \Begin{
	\If{$\noReconfig()$}{$findMaxCounter()$\; 
	{\bf send} $maxC_i[i]$ {\bf to} $p_j$;}
	\lElse{\textbf{send} $\sf Abort$}
}
\vspace{.1em}
{\bf procedure} $findMaxCounter()$ \Begin{		
		$cancelExhaustedMaxC()$\;
		$counterReceiptAction()$\; 
		$maxC[i] \gets getMaxSeq()$;
}

{\bf procedure} $majWrite(maxC_i[i])$ 
\Begin{
	\lForEach{$p_j \in config$}{{\bf send} $majMaxWrite(maxC_i[i])$} 	
	{{\bf wait for $ACK$ from majority of $config$}\; 
	\lIf{$\sf Abort$ {\bf received}}{$abort()$}}
}

{\bf upon request for} $majMaxWrite(max^j)$ {\bf from} $p_j$ \Begin{ 
	\If{$\noReconfig() = \sf True$}{$maxC_i[j] \gets cleanCT(max_{ct}(max^j, maxC_i[j]))$\nllabel{CIM:getGreatest}\;
	\lIf{$max^j=_{lbl.lCreator}i$}{$storedCnts_i[i].enqueue(maxC[i])$}
	\lIf{$exhausted(maxC_i[j])$}{$cancelExhausted(maxC_i[j])$}
	{\bf send $ACK$ to} $p_j$\;
	}
	\lElse{\textbf{send} $\sf Abort$}
}

\end{footnotesize}
\end{algorithm*} 

\paragraph{Counter increment for configuration members.} Algorithm~\ref{alg:cntrIncrMember} shows how configuration members increment the counter. 
A member $p_i$ first sends a query to all other members requesting the counter that they consider as the global maximum and awaits for responses from a majority.
These counters are gathered and passed to the counter structures (line~\ref{CIM:readCntrs}). 
Using the $counterReceiptAction()$ the algorithm eventually finds the maximal epoch label and the maximal sequence number it knows for this label. 
In other words, it collects counters and finds the counter(s) with the largest global label; there can be more than one such counter, in which case it returns the one with the highest sequence number, breaking symmetry with the sequence number processor ids. 
Then it checks whether this maximal sequence number is {\em exhausted}.
When this is the case, it proceeds to find a new maximal label  until it finds one that is not exhausted and uses the maximal sequence number it knows for this epoch label. 
The processor then increments the sequence number by one, sets its identifier as the writer of the sequence number (line~\ref{CIM:incrCntr}) and sends the new counter to all members, awaiting for acknowledgments from a majority.
Note that read and write requests during reconfigurations are answered with $\sf Abort$.
When the processor requesting the read/write receives an $\sf Abort$ it immediately terminates the increment procedure returning $\bot$.


\paragraph{Counter increment for non-member participants.} 
As per Algorithm~\ref{alg:cntrIncrNonMember}, participants that do not  belong to the configuration, request counters from a majority of configuration members.
They request the greatest with respect to $\prec_{ct}$ counter that is non-exhausted and legit. They then increment and write this to a majority of the configuration.
If at any point during read or write they receive an $\sf Abort$, they stop the procedure and return $\bot$.
The same happens if they do not manage to find a maximal counter after the read.
They can expect though that because of the counter propagation and the correctness of the labeling algorithm, such a counter will eventually appear.

\begin{algorithm*}[t!]
   \caption{Counter Increment for non-member participant $p_i \not \in config$; code for $p_i$}

\label{alg:cntrIncrNonMember}
\begin{footnotesize}

{\bf procedure} $incrementCounter(config)$\label{CtNM:beginIncrement} \Begin{
	\textbf{let} $counters = majRead()$\label{CtNM:qRead}\;
	\If{ $\forall ct \in counter, \exists ct' \in counter: (\neg exhausted(ct')) \land (legit(ct')) \land (ct \preceq_{ct} ct')$}{\textbf{let} $maxCounter = ct'$ \textbf{else let} $maxCounter =  \bot$}
	\If{$maxCntr \neq \bot$}{
	{\bf let} $newCntr = \langle maxCntr.lbl, maxCntr.seqn + 1, i\rangle$\label{CtNM:cntIncr}\; 
	\lIf{$majWrite(newCntr)$\label{algCtNM:qWrite}}{\textbf{return} $maxCntr$ \textbf{else return} $\bot$}
	}
	\lElse{\textbf{return} $\bot$}
}\label{CtNM:end}

{\bf procedure} $majRead()$ \label{CtNM:qReadDef}\Begin{
	\textbf{let} $counters = \emptyset$\;
	\lForEach {$p_j \in config$}{{\bf send} $majMaxRead()$}
	\While{$|counters| \leq \frac{|config|}{2}$ }{\label{CtNM:qDataBookkeep}
		{\bf upon receipt of} $m$ {\bf from} $p_j \in config$ {\bf do}
			\lIf{$m = \sf Abort$}{$abort()$ \textbf{else} $counters.add(m)$} 
	}
\label{CtNM:majReadEnd}
	\textbf{return} $counters$;
}

{\bf procedure} $majWrite(maxCntr)$ \label{CtNM:qWriteSend}
\Begin{
	\lForEach{$p_j \in config$}{{\bf send} $majMaxWrite(maxCntr)$} 
	{\bf if $ACK$ from majority of $config$ received then wait\; else if $\sf Abort$ received then $abort()$\;}
\label{algCt:waitQWrite}
}
\end{footnotesize}
\end{algorithm*} 

\subsubsection{Correctness}
\begin{lemma}
\label{thCNT:convLab}
Starting from an arbitrary state, in an execution $R$ of Algorithm~\ref{alg:configCounting}, configuration members eventually converge to a global maximal label.
\end{lemma}

\begin{proof}
We note that the result asks that configuration members reach to a global maximal label, and thus it is not assume that they all hold the same $seqn$ and $wid$ corresponding to this label.
We prove in a step-by-step fashion, following the labeling algorithm lemmas and explaining how the new algorithms allow do not affect the correctness.\\
\noindent (i) Initially we note that labels by non-members are in an identical way as in the labeling algorithm excluded and cannot be reintroduced to the system (Lemma~\ref{thL:staleInfo}).
This is immediate for Algorithm~\ref{alg:configCounting}.
Also lines~\ref{CIM:readCntrs} and~\ref{CIM:getGreatest} of Algorithm~\ref{alg:cntrIncrMember} guarantee that counters read during reads or writes are set to $\bot$ if they have labels by processors not in the current configuration.\\
\noindent (ii) In the same way as the labeling algorithm here is reduced the algorithm of \cite{SSVS} we perform the same reduction on the counter algorithm, by running the algorithm on the configuration member set and having each instance of the counter algorithm correspond to the fixed-set case.
\noindent (iii) By this we conclude that we reach to a global maximal label for all active members.
The bounds given by Theorem~\ref{thL:uniqueLabel} are still relevant, since as we have explained, counters with a corrupted $seqn$ and counters with a corrupt label are bounded by the same numbers.
Thus a counter may be adopted but then be exhausted very quickly, because it was initialized near its maximum.
Thus eventually any corrupt counters will be removed.
\end{proof}

\begin{theorem}
\label{thCNT:finalApp}
Starting from an arbitrary state, Algorithms~\ref{alg:configCounting}, \ref{alg:cntrIncrMember} and~\ref{alg:cntrIncrNonMember} eventually establish a monotonically increasing counter within the configuration that they are being executed.
\end{theorem}

\begin{proof}
We first note that Lemma~\ref{thCNT:convLab} ensures that processors reach to a maximal label.
We establish that any two calls to $incrementCounter()$ return a strictly ordered counter value. 
Consider a processor $p_i$ performing a counter increment.
Given that no reconfiguration takes place and thus no aborts take place during counter reads or writes, then if a majority has received the global counter $cnt$ then by the intersection property of majorities, $p_i$ must receive at least one copy of $cnt$ when it reads (line~\ref{CIM:majRead} for Alg.~\ref{alg:cntrIncrMember} and line~\ref{CtNM:qReadDef} for Alg.~\ref{alg:cntrIncrNonMember}).
For Algorithm~\ref{alg:cntrIncrNonMember} if any of the counters collected is legit and not exhausted it is incremented, and written back to a majority of configuration members with the writer's $wid$.
In the members' algorithm, a member can always find a maximal counter (with Algorithm~\ref{alg:configCounting}, even if the majority did not manage to return a legit, non-exhausted counter.
Note that members always hold the greatest counter and discard of the smaller one. 
This ensures that after the write completes, any subsequent call to $incrementCounter()$ will return at least one copy of this new greatest value. 
Concurrent calls to $incrementCounter()$ may return to two processors $p_i$ and $p_j$ the same global maximal counter.
This will be incremented by both to the same $seqn$, but will will be written with different $wid$'s, so assuming $i<j$, any subsequent comparison will give the counter by $p_j$ as the maximal.
Hence the counter algorithm establishes a monotonically increasing counter.
Also note, that for non-members, the algorithm will return a greater counter than the last completed call at the time of the call, or will return a $\bot$, in cases where reconfiguration is taking place, or where it is impossible to find a maximal counter, since the labels have yet to converge (and therefore incomparable counters exist).
\end{proof}


\subsection{Reconfigurable Virtually Synchronous State Machine Replication} 
\label{sec:VS}

The self-stabilizing reconfiguration service together with the self-stabilizing labeling and counter scheme, can extend the capabilities of various applications to run on more dynamic settings.
Virtual synchrony is an established technique for achieving state machine replication (SMR).
We have recently presented a self-stabilizing virtually synchronous SMR algorithm for a fixed set of processors~\cite{SSVS}.
We now discuss how this virtually synchronous SMR solution can benefit from the reconfiguration service to run on more dynamically changing environments.
We note that applications can use the reconfiguration service to guarantee continuous service, only when a delicate reconfigurations take place between periods of steady configuration and not when brute reconfiguration takes place. 
Nevertheless, brute reconfiguration guarantees that, after a transient fault, the service will eventually return to the desired behavior. 

A \emph{view} is a the set of processors with a unique identifier, that allows to the view members to achieve reliable multicast within the view.
The self-stabilizing reconfigurable virtually synchronous SMR task is specified such that any two processors that are together in any two consecutive views will deliver the same messages in their respective views and preserve the same state.
Moreover, the two views can be belong to two different consecutive configurations when the second configuration was the result of delicate reconfiguration. 
A pseudocode is found in Algorithm~\ref{alg:rvs} and a description and correctness proof follows.

\subsubsection{Description} 

%
\paragraph{Obtaining a coordinator.}
The virtual synchrony algorithm of~\cite{SSVS} is coordinator-based and works on the primary component given the \emph{supportive majority} assumption on the failure detectors.
This assumption states that a majority of processors of the (fixed) processor set mutually never suspect some processor on their failure detectors throughout an infinite execution, given that this processor does not crash.
The proof of~\cite{SSVS} establishes that such a supported processor eventually becomes the coordinator throughout the execution.

We modify the definition of supportive majority so that a majority of the configuration members' set ($config$) needs to provide such support to the coordinator.
Out of the current $config$, any $config$ member 
with supporting majority can obtain a counter from the labeling and counter algorithms that are run on $config$. 
The processor with the greatest counter, becomes the coordinator and establishes a \emph{view} reflecting its own local failure detector, and with the counter as its view identifier{\footnote{We may occasionally refer to view members and this should not be confused with configuration members, although a majority of configuration members must belong to the view}}.
In doing so it also collects the states from view members as well as undelivered messages and synchronizes to create the most up-to-date state for the view members to replicate.
The coordinator changes the view when the view set does not reflect its failure detector set.
View members follow the view composition and state of the coordinator.
Upon coordinator collapse, the same process provides a new coordinator and preserves the state. 
Note that the correctness for this is immediate from the correctness proofs of~\cite{SSVS}.

\paragraph{Coordinator-controlled joins.}
By the joining protocol, before a joiner tries to become a participant, it has any application-related local data set to $\bot$. 
The coordinator is the one that controls whether the application may or may not allow joining processors to become participants.
Several approaches may be used to give permission to joiners.
The coordinator may allow joining whenever it detects that more participants are required, by raising a flag and warning $config$ members that they can allow access to the application. 
This can be easily applied to our current joining protocol (Algorithm~\ref{alg:join}) where configuration members implement the $passQuery()$ interface, simply by returning their most recent $\sf True/False$ value they have for the coordinator's flag.
Another approach would be for a configuration member to provide the coordinator's details to the joiner so that joiners may, by directly communicating with the coordinator, gain permission to become participants.

Note that in both cases, joiners may become participants but they have yet to gain access to the application.
This takes place if they are part of the coordinator's FD and are subsequently included in the next view.
In becoming part of the view they also acquire a copy of the most recent state and begin state replication.
This gives leverage to the application when controlling the number of processors that are allowed to access the application.

\begingroup
\LinesNumberedHidden
\begin{algorithm}[t!]

\caption{Coordinator-led delicate reconfiguration; code for processor $p_i$}
\label{alg:coordUpper}

\begin{footnotesize}

\nlset{1} \noindent {\bf Interfaces:}
$needDelicateReconf()$ returns $\sf True/False$ on whether $p_i$ is the coordinator and whether application criteria requires and is ready for a  a reconfiguration.
\label{coordUpper:interfaces}

%
%
%
\tcc{Replaces line~\ref{SSQR:gracefulQreconf} of Reconfiguration Management layer}
\vspace{.3em}

\nlset{17} \uElseIf{$needDelicateReconf()$\label{coordUpper:delicate}}{
}

\end{footnotesize}
\end{algorithm}
\endgroup

\paragraph{Coordinator-initiated reconfiguration.}
The coordinator can be given the authority to initiate delicate reconfigurations.
This implies that there is no need for the prediction functions of a majority of processors to support a reconfiguration before it can take place, but it only suffices for a coordinator to be in place, in order to take the decision for reconfiguration based on application-specific criteria and suspend changes to the state during reconfiguration. 
To this end, line~\ref{SSQR:gracefulQreconf} of Algorithm~\ref{alg:SSQR} that triggers the delicate reconfiguration in Algorithm~\ref{alg:disCongif} is replaced with line~\ref{coordUpper:delicate} of Algorithm~\ref{alg:coordUpper}.
The $needReconf$ flag is rendered unnecessary for this cause.

\begin{algorithm*}[t!]

   \caption{Self-stabilizing reconfigurable VS SMR; code \textbf{for participant} $p_i$}
\label{alg:rvs}
\begin{footnotesize}


\noindent {\bf Interfaces:}
$fetch()$ next multicast message, 
$apply(state, msg)$ applies the step $msg$ to $state$ (while producing side effects), 
$synchState(replica)$ returns a replica consolidated state, 
$synchMsgs(replica)$ returns a consolidated array of last delivered messages, 
$inc()$ returns a counter from the increment counter algorithm,
$getConfig()$ returns the latest configuration and $\noReconfig()$ returns $\sf True/False$ on whether a reconfiguration is \emph{not} taking place (Algorithm~\ref{alg:disCongif}),
$evalConf()$ is the reconfiguration prediction function returning $\sf True/False$.
\nllabel{VSln:inter} 

\noindent {\bf Variables:} \nllabel{VSln:var}
The following arrays consider both $p_i$'s own value (the $i$-th entry) and $p_j$'s most recently received value (the $j$-th entry).  
$\rep[\bullet]=\langle view$ $=$ $\langle ID$, $set \rangle$, 
$status$ $\in$ $\{{\sf Multicast}$, ${\sf Propose}$, ${\sf Install}\}$, 
$(multicast$ $round$ $number)$ $rnd$, 
$(replica)$ $state$, $(last$ $delivered$ $messages)$ $msg[n]$ $(to$ $the$ $state$ $machine)$, $(last$ $fetched)$ $input$ $(to$ $the$ $state$ $machine)$, 
$propV$ $=$ $\langle ID$, $set \rangle$, $(no$ $coordinator$ $alive)$ $noCrd$, 
$(message$ $delivery)$ $suspend \rangle$ : an array of state variables. 
$FD[j].crd$ returns the last reported identifier of $p_j$'s local coordinator which is continuously propagated using the token exchange (see Section~\ref{s:sys}). It is $\bot$ if $p_j$ has no coordinator. 

%
\textbf{Function} $needDelicateReconf() = $ $(reconfReady = {\sf True} \land valCrd = p_i \land evalConfig()=\sf True)$\; 
 \nllabel{VSln:go4reconf}

\noindent {\bf Do forever} \Begin{

\textbf{let} $curConf = getConfig()$; \tcp{Gets latest \config ~set}
{\bf let} $seemC‏rd$ $=$ $\{ p_\ell$ $=$ $\rep[\ell].propV.ID.wid$ $\in$ $FD.part \cap curConfig$ 
$:$ $(|\rep[\ell].propV.set|$  $>$ $\lfloor |curConf|/2\rfloor)$ 
$(|\rep[\ell].FD.part|$ $>$ $\lfloor config/2\rfloor)$ 
$\land$ 
$( p_\ell \in$ $\rep[\ell].propV.set)$ $\land$ 
$(p_k \in$ $\rep[\ell].propV.set$ $\leftrightarrow$ $p_\ell$ $\in$ $FD[k].part)$ 
$\land$ $((\rep[\ell].status$ $=$ ${\sf Multicast})$ $\rightarrow$ $(\rep[\ell].(view$ $=$ $propV) \land crd(\ell)=\ell))$  $\land$ $((\rep[\ell].status = {\sf Install})$ $\rightarrow$ $FD[\ell].crd$ $=\ell)\}$\; \nllabel{VSln:seemCrd}

{\bf let} $valCrd$ $=$  $\{ p_\ell$ $\in$ $seemC‏rd$ $:$ $(\forall p_k$ $\in$ $seemC‏rd$ $:$ $\rep[k].propV.ID$ $\preceq_{ct}$ $\rep[\ell].propV.ID) \}$\; \nllabel{VSln:valCrd}

$noCrd$ $\gets$ $(|valCrd|$ $\neq$ $1)$; $FD[i].crd \gets valCrd$\; \nllabel{VSln:noCrd}

\lIf{$(valCrd$ $=$ $\{p_i\} \land status={\sf Multicast} \land reconfReady = {\sf True})$\nllabel{VSln:updSuspMult}}{$suspend \gets (reconfReady \gets evalConfig())$\nllabel{VSln:noCrdSusp}} 
\lElseIf{$(valCrd \neq \{p_i\}) \land (valCrd=\{p_\ell\} \land state[\ell].status\in \{ \sf Propose, Install\})$}{$suspend \gets (reconfReady \gets {\sf False})$\nllabel{VSln:updSuspNoCrd}}
\lIf{$\noReconfig() = \sf False$}{$suspend \gets \sf True$ \nllabel{VSln:suspOnRecon}} 

\lIf{
$(|FD.part\cap curConf|>\lfloor |curConf|/2\rfloor)$ 
$\land$ 
$(((|valCrd|$ $\neq$ $1)$ $\land$ $(|\{ p_k \in FD.part $ $:$ $p_i$ $\in$ $FD[k].part$ 
$\land$ $\rep[k].noCrd \}|$ $>$ $\lfloor |curConf|/2\rfloor))$ 
$\lor$ $((valCrd = \{p_i\})$ $\land$ $(FD.part \neq propV.set)$ $\land $
$(|\{p_k\in FD.part:$ $\rep[k].propV = propV\}|$ $> \lfloor |curConf|/2 \rfloor)))$
$\land$ $\noReconfig()$
}
{$(status, propV)$ $\gets$ $({\sf Propose}$, $\langle inc()$, $FD.part\rangle)$} \nllabel{VSln:incrCntr}

\ElseIf{$(valCrd$ $=$ $\{p_i\})$ $\land$ $(\forall$ $p_j$ $\in$ $view.set$ $:$ $\rep[j].(view$, $status$, $rnd)$ $=$ $(view$, $status$, $rnd))$ $\lor$ $((status$ $\neq$ ${\sf Multicast})$ $\land$ $(\forall$ $p_j$ $\in$ $propV.set$ $:$ $\rep[j].(propV,status)=(propV,{\sf Propose}))$ \nllabel{ln:switch}}{
\If{$status={\sf Multicast} \land reconfReady = {\sf False}$\nllabel{VSln:mulC}}{
$apply(state, msg)$; 
$suspend \gets evalConf()$\; \nllabel{VSln:evalSuspend}
\If{$(reconfReady \gets$ $(\forall p_k \in view.set:rep[k].suspend = \sf True)) = \sf False) \lor (\noReconfig()=\sf True)$ \nllabel{VSln:setReconfReady}} 
{
$input\gets fetch()$\; \nllabel{VSln:fetchCrd}
\lForEach{$p_j \in P$}{{\bf if} $p_j \in view.set$~{\bf then} $msg[j]\gets \rep[j].input$ {\bf else} $msg[j]\gets \bot$\nllabel{ln:collect}}
$rnd \gets rnd+1$\; \nllabel{VSln:rndIncr}
}
}
\lElseIf{$status={\sf Propose}$\nllabel{VSln:proC}}{$(state,status, msg)\gets (synchState(\rep), {\sf Install}, synchMsgs(\rep))$}
\lElseIf{$status={\sf Install}$\nllabel{VSln:insC}}{$(view, status, rnd, suspend, reconfReady)\gets (propV,{\sf Multicast}$, 0, $\sf False, False$)}}

\ElseIf{$valCrd=\{p_\ell\} \land \ell\neq i \land ((\rep[\ell].rnd = 0 \lor rnd < \rep[\ell].rnd \lor \rep[\ell].(view\neq propV))$\nllabel{VSln:repF}}{
\If{$\rep[\ell].status={\sf Multicast} \land suspend = {\sf False}$\nllabel{VSln:optCond}}{
$\rep[i] \gets \rep[\ell]$\tcc*{also adopts suspend flag}\nllabel{VSln:replicate} 
$apply(state,\rep[\ell].msg)$\tcc*{for the sake of side-effects} \nllabel{VSln:applyF}
\lIf{$state[\ell].suspend = \sf False$}{$input \gets fetch()$ \nllabel{ln:fetchFol}}
}
\lElseIf{$\rep[\ell].status = {\sf Install} $\nllabel{ln:adoptRep}}{$\rep[i] \gets \rep[\ell]$}
\lElseIf{$\rep[\ell].status = {\sf Propose}$}{$(status, propV) \gets \rep[\ell].(status, propV)$}
\nllabel{VSln:adoptProp}
}



{\bf let} $sendSet$ $=$ $(seemC‏rd$ $\cup$ $\{ p_k$ $\in$ $propV.set$ $:$ $valCrd$ $=$ $\{ p_i \} \}$ $\cup$ $\{ p_k$ $\in$ $FD$ $:$ $noCrd$ $\lor$ $(status$ $=$ ${\sf Propose}) \})$ \nllabel{VSln:sendSet}

\lForEach{$p_j$ $\in$ $sendSet$}{$send(\rep[i])$} \nllabel{VSln:send}
} 

\noindent {\bf Upon message arrival} $m$ {\bf from} $p_j$ {\bf do} $\rep[j] \gets m$\; \nllabel{VSln:receive}

\end{footnotesize}
\end{algorithm*}

\paragraph{Before reconfiguration.}
In order to initiate a reconfiguration that will not result in loss of the state of the replicas, the coordinator must ensure that all view processors have the most recent state and that the exchange of messages is suspended.
We require that for the state to survive to the first view of the next configuration, at least one replica 
of the last view before reconfiguration must not crash. 

When the coordinator is informed by its prediction function that reconfiguration is required, it performs a multicast round to gather the most recent state and received messages, but also raises a $suspend$ flag that it propagates with the current state. 
It waits until all view members return their current state and messages and that they have suspended new messages.
Upon receiving this information it performs a final multicast round to ensure that the last state with the messages have been received and applied and then calls the reconfiguration through Algorithm~\ref{alg:SSQR}. 
Note that in case of coordinator crash, a new view coordinator needs to be established before the reconfiguration can take place, which adds an extra delay.


\paragraph{After reconfiguration.}
Once the labeling and increment counter algorithm have stabilized to a new max counter any member of the new configuration has access to a counter value in order to become the coordinator.
Processors propagate their state which can be either the last before reconfiguration or $\sharp$ in case they are newly joining processors.
The new coordinator will synchronize the states and messages and establish the new view so that the service can continue, and will also resume the application for new messages to be fetched.
This is in essence a mere view installation with nothing additional to the algorithm of~\cite{SSVS}.

\paragraph{Providing service.} 
Inside the view, the basic functionality of the algorithm of~\cite{SSVS}, i.e., reliable multicast and state replication, are not in any way obstructed by the underlying reconfiguration service when a reconfiguration is not taking place. 
Additionally, view changes are not affected when reconfiguration is not taking place.



\subsubsection{Correctness}
We establish the correctness of the algorithm, by extending the correctness results of~\cite{SSVS}. 
We first prove that the proposed service (Algorithm~\ref{alg:rvs}), when working in an established configuration, is an execution of the algorithm of~\cite{SSVS} and thus the correctness for that case follows.
We then consider the stabilization of the application after a reconfiguration, and conclude by establishing that state replication is unaffected as to its correctness when a delicate reconfiguration that was initiated by the coordinator takes place.

\begin{lemma}
\label{thVS:noStaleRecon}
Starting in an arbitrary state in an execution $R$, stale information can only cause a single reconfiguration throughout $R$, unless the prediction mechanism $evalConfig()$ allows it. 
\end{lemma}

\begin{proof}
We note that stale information relating to reconfiguration consist of $suspend$ and the copies of $suspend$ exchanged with every other processor, and $reconfReady$, as well as $valCrd$ that may result from stale information as per lines~\ref{VSln:seemCrd} and~\ref{VSln:valCrd}. 
We pay attention to when the conditions are satisfied for the $needDelicateReconf()$ interface (line~\ref{VSln:go4reconf}). 
This is called by a processor that believes itself to be the coordinator ($valCrd = p_i$), in order to initiate a delicate reconfiguration in the modified Algorithm~\ref{alg:coordUpper}.
The three conditions are:
\emph{(i)} $evalConfig()$ returns $\sf True$,
\emph{(ii)} All processors of the view have $suspend = \sf True$,
\emph{(iii)} The caller believes itself to be the coordinator, i.e., $valCrd = p_i$.
We will refer to a coordinator as one who has established its view in the primary component of the $config$, and other $config$ members follow this processor in performing state replication.
We study the following three complementary cases.\\
{\bf Case 1 -- A non-coordinator $p_i$ that does not believe itself to be the coordinator.} (I.e., $valCrd \neq p_i$) In this case, $p_i$ cannot initiate a reconfiguration since condition (iii) will always fail.\\
{\bf Case 2 -- A non-coordinator $p_i$ that believes it is the coordinator.} (I.e., $valCrd = p_i$ locally for $p_i$, but $valCrd \neq p_i$ for a set of other processors greater than the majority of the $config$.) In~\cite{SSVS} processors such as $p_i$ are proved to eventually stop believing to be the coordinator due to propagation of information assumption.
Nevertheless, before information may be propagated, such a processor's stale information, can cause a reconfiguration if the conditions are satisfied. 
Note that after reconfiguration takes place, this processor needs to again establish its (alleged) coordinatorship before it reconfigures again and in the process it needs to reset its $suspend$ variables (line~\ref{VSln:updSuspNoCrd}.\\
{\bf Case 3 -- A coordinator $p_i$ with corrupt initial state}, cannot cause a reconfiguration in case $evalConfig() = \sf False$. 
It may cause a reconfiguration if also satisfies condition \emph{(ii)} due to corruption and $evalConfig() = \sf True$.
This cannot be distinguished from the non-transient case.
As already noted, after reconfiguration takes place, this coordinator needs to again establish its coordinatorship and in the process it needs to reset its $suspend$ variables (line~\ref{VSln:updSuspNoCrd}).
\end{proof}

\begin{lemma}
\label{thVS:reduction}
Consider an execution $R$ of Algorithm~\ref{alg:rvs} in which the following hold throughout: (i) the supportive majority assumption, (ii) no reconfiguration takes place (and therefore a valid configuration $config$ is in place) and (iii) $evalConfig()=\sf False$.
Then this execution is a reduction to the execution of the self-stabilizing virtually synchronous SMR algorithm 
of~\cite{SSVS}.
\end{lemma}

\begin{proof}
The mapping of the fixed set of processors of~\cite{SSVS} (that was named $P$ where $|P|=n$) to the set of $config$, creates the fixed set (throughout $R$) that provides the supportive majority (of $\lfloor |config|/2 \rfloor $) for at least one processor of the configuration to eventually become the coordinator. 
While the views can also include non-member processors (that are participants), safety is provided by the majority of $config$ members.
We note that the algorithm in its structure has not changed less the fact of replacing the old fixed set with $config$, and the addition of the suspend mechanism which is not activated at any point, since $evalConfig()=\sf False$ throughout $R$.
We thus deduce that the lemma is a reduction to the non-reconfigurable VS SMR and we deduce the following corollary.
\end{proof}

\begin{corollary}
\label{thVS:redCorollary}
Consider an execution $R$ of Algorithm~\ref{alg:rvs} in which the following hold throughout: (i) the supportive majority assumption, (ii) no reconfiguration takes place (and therefore a valid configuration $config$ is in place) and (iii) $evalConfig()=\sf False$. 
Then, by Lemma~\ref{thVS:reduction} and the correctness of~\cite{SSVS},  starting from an arbitrary state, Algorithm~\ref{alg:rvs} simulates state machine replication preserving the virtual synchrony property.
\end{corollary}

\begin{lemma}
\label{thVS:eventReconf} 
\sloppy Consider an infinite execution $R$ of Algorithm~\ref{alg:rvs} with an established coordinator $p_i$. 
If $evalConfig_i() = \sf True$ from some system state $c \in R$ onwards, then we reach another system state $c' \in R$ in which a reconfiguration eventually take place (and after which $evalConfig() = \sf False$). 
\end{lemma}

\begin{proof}
From $c$ onwards, the established coordinator $p_i$ receives $\sf True$ whenever it calls $evalConfig()$.
This implies that because of line~\ref{VSln:evalSuspend} $state_i[i].suspend = {\sf True}$ always.
This value is sent in every iteration of the Algorithm by $p_i$ (lines~\ref{VSln:sendSet}--\ref{VSln:send}) to every processor in the view, and adopted by every processor in the view through lines~\ref{VSln:receive} and~\ref{VSln:replicate}.
Every view member $p_j$ adopting $p_i$'s state in $state_j[i]$ (along with $suspend = \sf True$) propagates $state_j[j].suspend = \sf True$ back to $p_i$.
Due to assumption that a message sent infinitely often is received infinitely often, eventually the coordinator learns that every processor in its view has adopted $suspend = \sf True$.

If during this process the coordinator needs to change the view (due to failure detector changes), then this should take place and the above procedure starts again after the view installation (because installation by line~\ref{VSln:proC} forces $suspend = \sf False$.
The same takes place if the coordinator is lost, in which case by line~\ref{VSln:noCrdSusp} and the previous remark about installation falsify $suspend$ flags, and the procedure for reconfiguring should be taken up by the next coordinator.

If every processor in the view has $suspend = \sf True$ and $p_i$ has been informed of this, then line~\ref{VSln:setReconfReady} sets $reconfReady$ to $\sf True$ for the coordinator, and line~\ref{VSln:updSuspMult} retains this value since $evalConfig() = \sf True$ by assumption.
Hence, whenever the coordinator uses the reconfiguration manager (Algorithm~\ref{alg:coordUpper}) to run line~\ref{coordUpper:delicate}, it will find that all conditions of $needDelicateReconf()$ are satisfied, thus the coordinator can move on to initiate the reconfiguration using the $configEstb()$ interface of the reconfiguration scheme, and hence the result.~\end{proof}

\begin{lemma} 
\label{thVS:stabAfterReconf}
Consider an infinite execution $R$ of Algorithm~\ref{alg:rvs}, starting from a system state $c'$ as in Lemma~\ref{thVS:eventReconf}. 
We eventually reach a new state $c'' \in R$ in which a new configuration is installed and a new valid coordinator is established in the primary component of $config$.
\end{lemma}

\begin{proof}
Starting from system state $c'$, a coordinator $p_i$ has already initiated a reconfiguration.
During reconfiguration, $p_i$ cannot proceed to increment the round number since line~\ref{VSln:setReconfReady} allows no round increments if reconfiguration is taking place.
Note that since it is the coordinator that triggers the reconfiguration, this processor is immediately informed about a reconfiguration through $\noReconfig()$.
By the correctness of the reconfiguration and stability assurance layer, we are guaranteed that the reconfiguration will complete and so we move from a configuration $config$ to a new one $config'$.

When the reconfiguration completes, every processor in the new configuration $config'$ should try to establish that there is a coordinator installed.
Even if $p_i$ survives to $config'$, it will need to change the view because of the change in configuration, although it will still consider itself as the coordinator.
By the second set of conditions that allow view proposal, if $((valCrd = \{p_i\})$ $\land$ $(FD_i[i].part \neq propV.set)$ $\land $
$(|\{p_k\in FD_i[i].part:$ $\rep_i[k].propV = propV_i\}|$ $> \lfloor |config'|/2 \rfloor)))$, namely if there is a majority of active maembers that follow the previously proposed view of $p_i$, then $p_i$ is eligible to create a new view.
Otherwise, any processor $p_\ell$ may trigger, given that the following conditions are satisfied. 
(i) $|FD_\ell[\ell].part|\cap config'| > \lfloor |config'|/2\rfloor $ ($p_j$ can trust a majority of the configuration).
(ii) The set of processors in $FD_\ell[\ell].part$ that $p_\ell$ trusts and $p_\ell$ knows to have their trust must state that they have no coordinator.
Note that by our majority supportive assumption and the eventual reception of messages, some processor will reach a state in which it will be able to propose.

The correctness arguments of~\cite{SSVS} complete the proof (by Corollary~\ref{thVS:redCorollary}), and so some processor manages to become the coordinator.
\end{proof}

\begin{lemma}
Consider an infinite execution $R$ of Algorithm~\ref{alg:rvs} that contains a system state $c'$ as defined in Lemma~\ref{thVS:eventReconf} and a following state $c''$ as guaranteed by Lemma~\ref{thVS:stabAfterReconf}. 
The replica state is preserved from $c'$ to $c''$.
\end{lemma}

\begin{proof}
We examine how the state at the time of reconfiguration triggering, during and after reconfiguration remains unchanged until a coordinator is in place and initiates multicasting.\\
\textbf{Step 1 --} We establish that the state of the replicas \textbf{before} reconfiguration, is preserved.
By Lemma~\ref{thVS:eventReconf}, processor are led to turn their flag $suspend$ to $\sf False$, when the coordinator $p_i$ suggests that a reconfiguration must take place.
In the next multicast round when every non-coordinator processor receives $p_i$'s $suspend = \sf True$ it adopts the last state (line~\ref{VSln:replicate}), applies changes to the state~\ref{VSln:applyF} and fetches multicast messages \emph{only if the suspend flag is $\sf False$}.
Similarly, the coordinator applies the changes to the state~\ref{VSln:applyF} after all the view members have adopted its state, and renews its suspend flag (which should always return $\sf True$, otherwise the suspension process stops). 
If all the participants have suspended (otherwise condition $(\forall$ $p_j$ $\in$ $view.set$ $:$ $\rep_i[j].(view$, $status$, $rnd)$ $=$ $\rep_i[i].(view$, $status$, $rnd))$  of line~\ref{VSln:incrCntr} would not hold), then line~\ref{VSln:setReconfReady} must return $\sf True$ so $reconfReady = \sf True$ and the coordinator does not fetch new messages (line~\ref{VSln:fetchCrd}) and stops incrementing round numbers. \\
\textbf{Step 2 -- } \textbf{During} reconfiguration, no new multicasts and message $fetch()$ may take place, since non-coordinators of $config$ cannot access line~\ref{VSln:optCond} due to the conditions of \ref{VSln:repF}.
They also cannot propose a new view during reconfiguration, by the last condition of line~\ref{VSln:incrCntr}.
So their $suspend$ flags are preserved to $sf True$ throughout the reconfiguration (also by line~\ref{VSln:suspOnRecon}).
The coordinator's flag is also maintained to $\sf True$ if reconfiguration is taking place. \\
%
\textbf{Step 3 --} \textbf{After} the new configuration $config'$ is in place, new processors are assumed to access the computation with default values in their state, and cannot therefore introduce new messages or replica state.
Processors coming from $config$, carry the last state and for liveness we assume a single of these processors is alive and included in the proposed view set of the next coordinator. 
That a new proposal takes place, it is suggested by Lemma~\ref{thVS:stabAfterReconf}.
So the new coordinator (possibly the same as the one before reconfiguration) will gather all available states and messages and create a synchronized state (which is only for stabilization purposes here, since no new messages were allowe to be fetched and there should be only one version of last state).
Note that the $suspend$ flags are set to $\sf False$ for every processor when reconfiguration finishes and when a valid view proposal is found (line~\ref{VSln:updSuspNoCrd}).
\end{proof}\\

\noindent The above lemmas lead to the following theorem.

\begin{theorem}
\label{thVS:finalApp}
Starting in an arbitrary state in an execution $R$ of Algorithm~\ref{alg:rvs}, the algorithm simulates state machine replication preserving the virtual synchrony property, even in the case of reconfiguration when this is delicate, i.e., when initiated by $needDelicateReconf()$.
\end{theorem}

\paragraph{Self-stabilizing reconfigurable emulation of shared memory.} Birman et al.~\cite{birmanMR2010} show how a virtual synchrony solution can lead to a reconfigurable emulation of shared memory. 
Following this approach, and using our self-stabilizing reconfigurable SMR solution discussed in this section, and our increment counter scheme (Section~\ref{sec:counter}), we can obtain a self-stabilizing reconfigurable emulation of shared memory.
Given a conflict-free configuration, a typical two-phase read and write protocol can be used for the shared memory emulation.
In the event of a delicate reconfiguration, the coordinator (of the virtual synchrony algorithm) suspends reads and writes on the register and once a new configuration is established, the emulation continues.
Virtual synchrony ensures that the state of the system, in this case the state of the object, is preserved (c.f., Theorem~\ref{thVS:finalApp}).
In the event of a brute force reconfiguration (e.g. due to transient faults or violation of the churn rate), the system will automatically recover and eventually reach a legal execution (in this case the state of the system may be lost).

We note that our proposed self-stablizing reconfigurable SMR and shared memory emulation solutions are suspending, in the sense that 
they do not provide service during a reconfiguration. With some extra care and under certain conditions we believe that they
can be modified to provide continuous service, but in general, it remains an interesting open question whether a \emph{self-stabilizing} 
service, such as reconfigurable SMR or distributed shared memory that does not suspend, is possible. In~\cite{birmanMR2010}, Birman et al. discuss the tradeoffs of suspending and non suspending reconfiguration (such as the ones provided in \cite{RAMBO} and \cite{DynaStore}). It is argued, that suspending services provide some simpler solutions, and may be enhanced for more efficient reconfiguration decisions so that the time for reconfiguration and state transfer before reconfiguration can be reduced.

\section{Conclusion}
We presented the first self-stabilizing reconfiguration scheme that recovers automatically from transient faults, such as temporary violation of the predefined churn rate or the unexpected activities of processors and communication channels, using a bounded amount of local storage and message size. We showed how this scheme can be used for the implementation of several dynamic distributed services, such as a self-stabilizing reconfigurable virtual synchrony, which in turn can be used for developing self-stabilizing reconfigurable SMR and shared memory emulation solutions. We use a number of bootstrapping techniques for allowing the system to always recover from arbitrary transient faults, for example, when the current configuration includes no active processors. 
We believe that the presented techniques provide a generic blueprint for different solutions that
are needed in the area of self-stabilizing high-level communication and synchronization primitives, which need to deal with processor joins and leaves as well as transient faults.

\clearpage




\end{document}